\documentclass[a4paper,UKenglish]{lipics-v2016}

\usepackage{epsfig}
\usepackage{amsmath}
\usepackage{color}
\usepackage{amsfonts,amssymb}
\usepackage{verbatim}
\usepackage{times}
\usepackage{ulem}    
\usepackage{enumerate}

\usepackage[linesnumbered,ruled,procnumbered,noend]{algorithm2e}
\usepackage{stmaryrd}
\usepackage{thmtools}
\usepackage{thm-restate}

\usepackage[version=0.96]{pgf}
\usepackage{tikz}
\usetikzlibrary{arrows,shapes,snakes,automata,backgrounds,petri}

\DeclareSymbolFont{largesymbolsA}{U}{txexa}{m}{n}
\SetSymbolFont{largesymbolsA}{bold}{U}{txexa}{bx}{n}
\DeclareFontSubstitution{U}{txexa}{m}{n}
\DeclareMathSymbol{\bigsqcupplus}{\mathop}{largesymbolsA}{"02}

\newcommand{\forget}[1]{}

\title{Checking Linearizability of Concurrent Priority Queues}

\author[1]{Ahmed Bouajjani}
\author[1]{Constantin Enea}
\author[1]{Chao Wang}
\affil[1]{Institut de Recherche en Informatique Fondamentale, \\ \texttt{\{abou,cenea,wangch\}@irif.fr}}

\begin{document}

\maketitle

\begin{abstract}
Efficient implementations of concurrent objects such as atomic collections are essential to modern computing. Programming such objects is error prone: in minimizing the synchronization overhead between concurrent object invocations, one risks the conformance to sequential specifications -- or in formal terms, one risks violating linearizability. Unfortunately, verifying linearizability is undecidable in general, even on classes of implementations where the usual control-state reachability is decidable. In this work we consider concurrent priority queues which are fundamental to many multi-threaded applications such as task scheduling or discrete event simulation, and show that verifying linearizability of such implementations can be reduced to control-state reachability. This reduction entails the first decidability results for verifying concurrent priority queues in the context of an unbounded number of threads, and it enables the application of existing safety-verification tools for establishing their correctness.
\end{abstract}

\forget{
}

\section{Introduction}
\label{sec:introduction}

Modern computer software is increasingly concurrent. Interactive applications
and services necessitate reactive asynchronous operations to handle requests
immediately as they happen, rather than waiting for long-running operations to
complete. Furthermore, as processor manufacturers approach clock-speed limits,
performance improvements are more-often achieved by parallelizing operations
across multiple processor cores.



 Multithreaded software is typically built with specialized “concurrent
  objects” like atomic integers, queues, maps, priority queues. These objects’ methods are
  designed to confom to better established sequential specifications, a property known as \emph{linearizability}~\cite{journals/toplas/HerlihyW90},
  despite being optimized to avoid blocking and exploit parallelism, e.g.,~by
  using atomic machine instructions like compare-and-swap. Intuitively, linearizability asks that every individual operation appears to take place instantaneously at some point between its invocation and its return. Verifying linearizability is intrinsically hard, and undecidable in general~\cite{conf/esop/BouajjaniEEH13}. 
However, recent work~\cite{DBLP:conf/icalp/BouajjaniEEH15} has shown that for particular classes of objects, 
i.e., registers, mutexes, queues, and stacks, the problem of verifying linearizability becomes decidable (for finite-state implementations).

In this paper, we consider another important object, namely the priority queue, which is essential for applications such as task scheduling and discrete event simulation. Numerous implementations have been proposed in the research literature, e.g.,~\cite{DBLP:conf/ppopp/AlistarhKLS15,DBLP:conf/wdag/CalciuMH14,DBLP:conf/opodis/LindenJ13,DBLP:conf/podc/ShavitZ99,DBLP:conf/ipps/ShavitL00}, and concrete implementations exist in many modern languages like C++ or Java.  
%
%
Priority queues are collections providing $\textit{put}$ and $\textit{rm}$ methods for adding and removing values. Every value added is associated to a priority and a remove operation returns a minimal priority value. For generality, we consider a partially-ordered set of priorities. Values with incomparable priorities can be removed in any order, and values having the same priority are removed in the FIFO order. Implementations like the PriorityBlockingQueue in Java where same priority values are removed in an arbitrary order can be modeled in our framework by renaming equal priorities to incomparable priorities (while preserving the order constraints).

Compared to previously studied collections like stacks and queues, the main challenge in dealing with priority queues is that the order in which values are removed is not fixed by the happens-before between add/remove operations (e.g., in the case of a queue, the values are removed in the order in which they were inserted), but by parameters of the $\textit{put}$ operations (the priorities) which come from an unbounded domain. For instance, the sequential behavior $\textit{put}(a,p_1)\cdot \textit{put}(b,p_3)\cdot \textit{put}(c,p_2)\cdot \textit{rm}(a,p_1)\cdot \textit{rm}(c,p_2)$ where the priority $p_1$ is less than $p_2$ which is less than $p_3$, is not admitted neither by the regular queue nor the stack. 



Following the approach in~\cite{DBLP:conf/icalp/BouajjaniEEH15}, we give a characterization of concurrent priority queue behaviors violating linearizability in terms of automata. This characterization enables a reduction of checking linearizability for arbitrary implementations to reachability or invariant checking, and implies decidability for checking linearizability of finite-state implementations. However, differently from the case of stacks and queues where finite-state automata are sufficient, we show that the case of priority queues needs \emph{register automata} where registers are used to store and compare priorities.

This characterization is obtained in several steps. We first define a recursive procedure which recognizes valid sequential executions which is then extended to recognize linearizable concurrent executions. Intuitively, for an input execution $e$, this procedure handles values occurring in $e$ one by one, starting with values of maximal priority (which are to be removed the latest). For each value $x$, it checks whether $e$ satisfies some property ``local'' to that value, i.e., which is agnostic to how the operations adding or removing other values are ordered between them (w.r.t. the happens-before), other than how they are ordered w.r.t. the operations on $x$. When this property holds, the same is done for the rest of the execution, without the operations on $x$. This procedure works only for executions where a value is added at most once, but this is not a limitation for \emph{data-independent} implementations whose behavior doesn't depend on the values that are added or removed. In fact, all the implementations that we are aware of are data-independent.

Next, we show that checking whether an execution violates this ``local'' property for a value $x$ can be done using a class of register automata~\cite{DBLP:journals/tcs/KaminskiF94,DBLP:conf/icalp/Cerans94,DBLP:conf/stacs/SegoufinT11} (transition systems where the states consist of a fixed set of registers that can receive values and be compared). Actually, only two registers are needed: one register $r_1$ for storing a priority guessed at the initial state, and one register $r_2$ for reading priorities as they occur in the execution and comparing them with the one stored in $r_1$. We show that registers storing values added to or removed from the priority queue are not needed, since any data-independent implementation admits a violation to linearizability whenever it admits a violation where the number of values is constant, and at most 5 (the number of priorities can still be unbounded).

The remainder of this article is organized as follows.
Section~\ref{sec:priority queue and data-independence} describes the priority queue ADT, lists several semantic properties like data-independence that are satisfied by implementations of this ADT, and formalizes the notion of linearizability. 
Sections~\ref{sec:checking inclusion by recursive procedure} defines a recursive procedure for checking linearizability of concurrent priority queue behaviors.
Section~\ref{sec:co-regular of extended priority queues} gives an automata characterization of the violations to linearizability, and
Section~\ref{sec:related} discusses related work.


\newcommand{\seqPQ}{\mathsf{SeqPQ}}

\section{The Priority Queue ADT}
\label{sec:priority queue and data-independence}

We consider priority queues whose interface contains two methods $\textit{put}$ and $\textit{rm}$ for adding and respectively, removing a value. Each value is assigned with a priority when being added to the data structure (using a call to  $\textit{put}$) and the remove method $\textit{rm}$ removes a value with a minimal priority. For generality, we assume that the set of priorities is partially-ordered. Incomparable priorities can be removed in any order. In case multiple values are assigned with the same priority, $\textit{rm}$ returns the least recent value (according to a FIFO semantics). Also, when the set of values stored in the priority queue is empty, $\textit{rm}$ returns the distinguished value $\textit{empty}$. Concurrent implementations of the priority queue allow the methods $\textit{put}$ and $\textit{rm}$ to be called concurrently from different threads. However, every method invocation should give the illusion that it takes place instantaneously at some point between its invocation and its return. We formalize (concurrent) executions and implementations in Section~\ref{ssec:exec}, Section~\ref{ssec:semantic_prop} introduces a set of properties satisfied by all the  implementations we are aware of, and Section~\ref{ssec:lin} defines the standard correctness criterion for concurrent implementations of ADTs known as \emph{linearizability}~\cite{journals/toplas/HerlihyW90}.

\subsection{Executions}\label{ssec:exec}

We fix a (possibly infinite) set $\mathbb{D}$ of data values, a (possibly infinite) set $\mathbb{P}$ of priorities, a partial order $\prec$ among elements in $\mathbb{P}$, and an infinite set $\mathbb{O}$ of operation identifiers.
The latter are used to match call and return actions of the same invocation. Call actions $\textit{call}_o(\textit{put},a,p)$ and $\textit{call}_o(\textit{rm},a')$ with $a\in \mathbb{D}$, $a'\in \mathbb{D}\cup\{\textit{empty}\}$, $p \in \mathbb{P}$, and $o \in \mathbb{O}$, combine a method name and a set of arguments with an operation identifier. The return value of a remove is transformed to an argument value for uniformity~\footnote{Method return values are guessed nondeterministically, and validated at return points.
This can be handled using the {\tt assume} statements of typical formal specification languages, which only admit executions satisfying a given predicate.}.
The return actions are denoted in a similar way as $\textit{ret}_o(\textit{put},a,p)$ and respectively, $\textit{ret}_o(\textit{rm},a')$.

An \emph{execution} $e$ is a sequence of call and return actions which satisfy the following well-formedness properties: each return is preceded by a matching call (i.e., having the same operation identifier), and each operation identifier is used in at most one call/return. We assume every set of executions is closed under isomorphic renaming of operation identifiers. An execution is called \emph{sequential} when no two operations overlap, i.e., each call action is immediately followed by its matching return action, and \emph{concurrent} otherwise. To ease the reading, we write a sequential execution as a sequence of $\textit{put}(a,p)$ and $\textit{rm}(a)$ symbols representing a pair of actions $\textit{call}_o(\textit{put},a,p)\cdot \textit{ret}_o(\textit{put},a,p)$ and $\textit{call}_o(\textit{rm},a)\cdot \textit{ret}_o(\textit{rm},a)$, respectively (where $o\in\mathbb{O}$). For example, given two priorities $p_1 \prec p_2$, $\textit{put}(a,p_2) \cdot \textit{put}(b,p_1) \cdot \textit{rm}(b)$ is a sequential execution of the priority queue ($\textit{rm}$ returns $b$ because it has smaller priority).

We define $\seqPQ$, the set of sequential priority queue executions, semantically via labelled transition system ($LTS$). An LTS is a tuple $A=(Q,\Sigma,\rightarrow,q_0)$, where $Q$ is a set of states, $\Sigma$ is an alphabet of transition labels, $\rightarrow\subseteq Q\times\Sigma\times Q$ is a transition relation and $q_0$ is the initial state. We model priority queue as an LTS $\textit{PQ}$ defined as follows: each state of the LTS $\textit{PQ}$ is a mapping associating priorities in $\mathbb{P}$ with sequences of values in $\mathbb{D}$ representing a snapshot of the priority queue (for each priority, the values are ordered as they were inserted); the transition labels are $\textit{put}(a,p)$ and $\textit{rm}(a)$ operations; Each transition modifies the state as expected. For example, $q_1 \xrightarrow{\textit{rm}(\textit{empty})} q_2$ if $q_1 = q_2$, and $q_1$ and $q_2$ map each priority to the empty sequence $\epsilon$. Then, $\seqPQ$ is the set of traces of $\textit{PQ}$. The detailed definition of $\textit{PQ}$ can be found in Appendix \ref{sec:appendix definition of seqPQ and proof of Lemma EQP rules and semantics}.

An implementation $\mathcal{I}$ is a set of executions. Implementations represent libraries whose methods are called by external programs. In the remainder of this work, we consider only completed executions, where each call action has a corresponding return action. This simplification is sound when implementation methods can always make progress in isolation: formally, for any execution $e$ with pending operations, there exists an execution $e'$ obtained by extending $e$ only with the return actions of the pending operations of $e$. Intuitively this means that methods can always return without any help from outside threads, avoiding deadlock.

\subsection{Semantic Properties of Priority Queues}\label{ssec:semantic_prop}

We define two properties which are satisfied by priority queue implementations and which are important for the results that follow: (1) \emph{data independence}~\cite{conf/popl/Wolper86,conf/tacas/AbdullaHHJR13} states that priority queue behaviors do not depend on the actual values which are added to the queue, and (2) \emph{closure under projection}~\cite{DBLP:conf/icalp/BouajjaniEEH15} states that intuitively, remove operations can return the same values no matter how many other different values are in the queue, assuming they don't have more important priorities.

An execution $e$ is \emph{data-differentiated} if every value is added at most once, i.e., for each $d \in \mathbb{D}$, $e$ contains at most one action $\textit{call}_o(\textit{put},d,p)$ with $o\in\mathbb{O}$ and $p\in \mathbb{P}$. Note that this property concerns only values, a data-differentiated execution $e$ may contain more than one value with the same priority. The subset of data-differentiated executions of a set of executions $E$ is denoted by $E_{\neq}$.

A renaming function $r$ is a function from $\mathbb{D}$ to $\mathbb{D}$. Given an execution $e$, we denote by $r(e)$ the execution obtained from $e$ by replacing every data value $x$ by $r(x)$. Note that $r$ renames only the values and keep the priorities unchanged. Intuitively, renaming values has no influence on the behavior of the priority queue, contrary to renaming priorities.

\begin{definition}\label{def:priority-value data-independence}
A set of executions $E$ is \emph{data independent} iff
\begin{itemize}
\setlength{\itemsep}{0.5pt}
\item[-] for all $e \in E$, there exists $e' \in E_{\neq}$ and a renaming function $r$, such that $e=r(e')$,

\item[-] for all $e \in E$ and for all renamings $r$, $r(e) \in E$.
\end{itemize}
\end{definition}

The following lemma is a direct consequence of definitions.

\begin{lemma}
$\seqPQ$ is data independent.
\end{lemma}

Beyond the sequential executions, every (concurrent) implementation of the priority queue that we are aware of is data-independent. Therefore, from now on, we consider only data-independent implementations. This assumption enables a reduction from checking the correctness of an implementation $\mathcal{I}$ to checking the correctness of only its data-differentiated executions in $\mathcal{I}_{\neq}$.

Besides data independence, the sequential behaviors of the priority queue satisfy the following closure property: a behavior remains valid when removing all the operations with an argument in some set of values $D \subseteq \mathbb{D}$ and any $\textit{rm}(\textit{empty})$ operation (since they are read-only and they don't affect the queue's state).
In order to distinguish between different $\textit{rm}(\textit{empty})$ operations while simplifying the technical exposition, we assume that they receive as argument a value, i.e., call actions are of the form $\textit{call}_o(\textit{rm},\textit{empty},a)$ for some $a\in \mathbb{D}$. We will make explicit this argument only when needed in our technical development. The projection $e \vert D$ of an execution $e$ to a set of values $D \subseteq \mathbb{D}$ is obtained from $e$ by erasing all call/return actions with an argument not in $D$. We write $e \setminus x$ for the projection $e \vert_{ \mathbb{D} \setminus \{ x \} }$. Let $\textit{proj}(e)$ be the set of all projections of $e$ to a set of values. The proof of the following lemma can be found in Appendix \ref{sec:appendix proofs in section priority queue and data-independence}.

\begin{restatable}{lemma}{EPQClosedUnderProjection}
\label{lem:closure_proj}
$\seqPQ$ is closed under projection, i.e., $\textit{proj}(e)\subseteq \seqPQ$ for each $e\in \seqPQ$.
\end{restatable}

\subsection{Linearizability}\label{ssec:lin}

%
We recall the notion of \emph{linearizability}~\cite{journals/toplas/HerlihyW90} which is the \emph{de facto} standard correctness condition for concurrent data structures.
Given an execution $e$, the happen-before relation $<_{\textit{hb}}$ between operations~\footnote{In general, we refer to operations using their identifiers.} is defined as follows: $o_1 <_{\textit{hb}} o_2$, if the return action of $o_1$ occurs before the call action of $o_2$ in $e$. The happens-before relation is an interval order \cite{DBLP:conf/popl/BouajjaniEEH15}: for distinct $o_1,o_2,o_3,o_4$, if $o_1 <_{\textit{hb}} o_2$ and $o_3 <_{\textit{hb}} o_4$, then either $o_1 <_{\textit{hb}} o_4$, or $o_3 <_{\textit{hb}} o_2$. Intuitively, this comes from the fact that concurrent threads share a notion of global time.

Given a (concurrent) execution $e$ and a sequential execution $s$, we say that $e$ is linearizable w.r.t $s$, denoted $e \sqsubseteq s$, if there is a bijection $f: O_1 \rightarrow O_2$, where $O_1$ and $O_2$ are the set of operations of $e$ and $s$, respectively, such that (1) $o$ and $f(o)$ is the same operation\footnote{An $m(a)$-operation in an execution $e$ is an operation identifier $o$ s.t. $e$ contains the actions $\textit{call}_o(m,a)$, $\textit{ret}_o(m,a)$.}, and (2) if $o_1 <_{\textit{hb}} o_2$, then $f(o_1) <_{\textit{hb}} f(o_2)$.
A (concurrent) execution $e$ is linearizable w.r.t a set $S$ of sequential executions, denoted $e \sqsubseteq S$, if there exists $s \in S$ such that $e \sqsubseteq s$. A set of concurrent executions $E$ is linearizable w.r.t $S$, denoted $E \sqsubseteq S$, if $e \sqsubseteq S$ for all $e \in E$.

The following lemma states that by data-independence, it is enough to consider only data-differentiated executions when checking linearizability (see Appendix \ref{sec:appendix proofs in section priority queue and data-independence}). This is similar to that in \cite{conf/tacas/AbdullaHHJR13,DBLP:conf/icalp/BouajjaniEEH15}, where they use the notion of data-independence in \cite{conf/popl/Wolper86}.
Section~\ref{sec:checking inclusion by recursive procedure} will focus on characterizing linearizability for data-differentiated executions. 

\begin{restatable}{lemma}{DataDifferentiatedisEnoughforPQ}
\label{lemma:data differentiated is enough for PQ}
A data-independent implementation $\mathcal{I}$ is linearizable w.r.t a data-independent set $S$ of sequential executions, if and only if $\mathcal{I}_{\neq}$ is linearizable w.r.t. $S_{\neq}$.
\end{restatable}

\section{Checking Linearizability of Priority Queue Executions}
\label{sec:checking inclusion by recursive procedure}

We define a recursive procedure for checking linearizability of an execution w.r.t. $\seqPQ$.
To ease the exposition, Section~\ref{ssec:seq_exec} introduces a recursive procedure for checking whether a data-differentiated \emph{sequential} execution is admitted by the priority queue which is then extended to the concurrent case in Section~\ref{ssec:conc_exec}.

\subsection{Characterizing Data-Differentiated Sequential Executions}\label{ssec:seq_exec}

Checking whether a data-differentiated sequential execution belongs to $\seqPQ$ could be implemented by checking membership into the set of traces of the LTS $PQ$. The recursive procedure $\textit{Check-PQ-Seq}$ outlined in Algorithm~\ref{alg:seq_check} is an alternative to this membership test. Roughly, it selects one or two operations in the input execution, checks whether their return values are correct by ignoring the order between the other operations other than how they are ordered w.r.t. the selected ones, and calls itself recursively on the execution without the selected operations. 

We explain how the procedure works on the following execution:
{
\begin{align}
\hspace{-5mm}\textit{put}(c,p_2) \cdot \textit{put}(a,p_1) \cdot \textit{rm}(a) \cdot \textit{rm}(c) \cdot \textit{rm}(\textit{empty}) \cdot \textit{put}(d,p_2) \cdot \textit{put}(f,p_3) \cdot \textit{rm}(f) \cdot \textit{put}(b,p_1)\label{eq:ex_rec1}
\end{align}
}
where $p_1$, $p_2$, $p_3$ are priorities such that $p_1 \prec p_2$ and $p_1 \prec p_3$, and $p_2$ and $p_3$ are incomparable. Since the $\textit{rm}(\textit{empty})$ operations 
are read-only (they don't affect the state of the queue), they are selected first. Ensuring that an operation $o=\textit{rm}(\textit{empty})$ is correct boils down to checking that every $\textit{put}(x,p)$ operation before $o$ is matched to a $\textit{rm}(x)$ operation which also occurs before $o$. This is true in this case for $x\in \{a,c\}$. Therefore, the correctness of (\ref{eq:ex_rec1}) reduces to the correctness of
{
\begin{align*}
\textit{put}(c,p_2) \cdot \textit{put}(a,p_1) \cdot \textit{rm}(a) \cdot \textit{rm}(c) \cdot \textit{put}(d,p_2) \cdot \textit{put}(f,p_3) \cdot \textit{rm}(f) \cdot \textit{put}(b,p_1)
\end{align*}
}
When there are no more $\textit{rm}(\textit{empty})$ operations, the procedure selects a $\textit{put}$ operation adding a value with maximal priority which is not removed, and then a pair of $\textit{put}$ and $\textit{rm}$ operations adding and removing the same maximal priority value. For instance, since $p_2$ is a maximal priority, it selects the operation $\textit{put}(d,p_2)$. {This operation is correct if $d$ is the last value with priority $p_2$, and the correctness of (\ref{eq:ex_rec1}) reduces to the correctness of
{
\begin{align*}
\textit{put}(c,p_2) \cdot \textit{put}(a,p_1) \cdot \textit{rm}(a) \cdot \textit{rm}(c) \cdot \textit{put}(f,p_3) \cdot \textit{rm}(f) \cdot \textit{put}(b,p_1)
\end{align*}
}
Since there is no other value of maximal priority which is not removed, the procedure selects a pair of $\textit{put}$/$\textit{rm}$ operations with an argument of maximal priority $p_2$,} for instance, $\textit{put}(c,p_2)$ and $\textit{rm}(c)$. The value returned by $\textit{rm}$ is correct if all the values of priority smaller than $p_2$ added before $\textit{rm}(c)$ are also removed before $\textit{rm}(c)$. In this case, $a$ is the only value of priority smaller than $p_2$ and it satisfies this property. Applying a similar reasoning for all the remaining values, it can be proved that this execution is correct.

\begin{algorithm}[t]
\KwIn {A data-differentiated sequential execution $e$}
\KwOut{$\mathsf{true}$ iff $e\in \seqPQ$}

\If {$e = \epsilon$}
{\Return $\mathsf{true}$\;}

\If {$\mathsf{Has\text{-}EmptyRemoves}(e)$}
{
    \If {$\exists\ o=\textit{rm}(\textit{empty})\in e$ such that $\mathsf{EmptyRemove\text{-}Seq}(e,o)$ holds}
    {
        \KwRet $\textit{Check-PQ-Seq}(e \setminus o)$\;
    }
}
\ElseIf{$\mathsf{Has\text{-}UnmatchedMaxPriority}(e)$}
{
    \If {$\exists\ x \in \textit{values}(e)$ such that $\mathsf{UnmatchedMaxPriority\text{-}Seq}(e,x)$ holds}
    {
        \KwRet $\textit{Check-PQ-Seq}(e \setminus x)$\;
    }
}

\Else
{
    \If {$\exists\ x \in \textit{values}(e)$ such that $\mathsf{MatchedMaxPriority\text{-}Seq}(e,x)$ holds}
    {
        \KwRet $\textit{Check-PQ-Seq}(e \setminus x)$\;
    }
    \Else {\KwRet $\mathsf{false}$\;}
}
\caption{$\textit{Check-PQ-Seq}$}
\label{alg:seq_check}
\end{algorithm}

In formal terms, the operations which are selected depend on the following set of predicates on executions:
{\small
\begin{align*}
\mathsf{Has\text{-}EmptyRemoves}(e)=\mathsf{true} & \mbox{ iff  $e$ contains a $\textit{rm}(\textit{empty})$ operation} \\
\mathsf{Has\text{-}UnmatchedMaxPriority}(e)=\mathsf{true} & \mbox{ iff $p\in \textit{unmatched-priorities}(e)$ for a maximal priority} \\
&\hspace{4mm}\mbox{$p\in priorities(e)$}
\end{align*}}
where $\textit{priorities}(e)$, resp., $\textit{unmatched-priorities}(e)$, is the set of priorities occurring in $\textit{put}$ operations of $e$, resp., in $\textit{put}$ operations of $e$ for which there is no $\textit{rm}$ operation removing the same value. We call the latter \emph{unmatched} put operations. A put operation which is not unmatched is called \emph{matched}. For simplicity, we consider the following syntactic sugar $\mathsf{Has\text{-}MatchedMaxPriority}(e)=\neg \mathsf{Has\text{-}EmptyRemoves}(e)\land \neg \mathsf{Has\text{-}UnmatchedMaxPriority}(e)$. By an abuse of notation, we also assume that  $\mathsf{Has\text{-}UnmatchedMaxPriority}(e) \Rightarrow \neg \mathsf{Has\text{-}EmptyRemoves}(e)$ (this is sound by the order of the conditionals in $\textit{Check-PQ-Seq}$).

The predicates defining the correctness of the selected operations are defined as follows:
{\small
\begin{align*}
\mathsf{EmptyRemove\text{-}Seq}(e,o)=\mathsf{true} & \mbox{ iff  $e= u\cdot o\cdot v$ and $\textit{matched}(u)$} \\
\mathsf{UnmatchedMaxPriority\text{-}Seq}(e,x)=\mathsf{true} & \mbox{ iff  $e= u\cdot \textit{put}(x,p)\cdot v$, $p\not\prec \textit{priorities}(u\cdot v)$, and $p\not\in \textit{priorities}(v)$} \\
\mathsf{MatchedMaxPriority\text{-}Seq}(e,x)=\mathsf{true} & \mbox{ iff  $e= u\cdot \textit{put}(x,p)\cdot v\cdot \textit{rm}(x)\cdot w$, $p\not\prec \textit{priorities}(u\cdot v\cdot w)$,} \\
&\hspace{4mm}\mbox{$p\not\preceq \textit{unmatched-priorities}(u\cdot v\cdot w)$, $\textit{matched}_\prec(u\cdot v,p)$,} \\
&\hspace{4mm}\mbox{and $p\not\in \textit{priorities}(v\cdot w)$}
\end{align*}
}

\vspace{-3mm}
\noindent
where $p\prec \textit{priorities}(e)$ when $p\prec p'$ for some $p'\in \textit{priorities}(e)$ (and similarly for $p\prec \textit{unmatched-p}$ $\textit{riorities}(e)$ or $p\preceq \textit{unmatched-priorities}(e)$),
$\textit{matched}_\prec(e,p)$ holds when each value with priority strictly smaller than $p$ is removed in $e$, and $\textit{matched}(e)$ holds when $\textit{matched}_\prec(e,p)$ holds for each $p\in\mathbb{P}$. Compared to the example presented at the beginning of the section, these predicates take into consideration that multiple values with the same priority are removed in FIFO order: the predicate $\mathsf{MatchedMaxPrioritySeq}(e,x)$ holds when $x$ is the last value with priority $p$ added in $e$.

When $o$ is a $\textit{rm}(\textit{empty})$ operation, we write $e\setminus o$ for the maximal subsequence of $e$ which doesn't contain $o$. For an execution $e$, $\textit{values}(e)$ is the set of values occurring in call/return actions of $e$.

The following lemma states the correctness of $\textit{Check-PQ-Seq}$ (see Appendix~\ref{sec:appendix definition of seqPQ and proof of Lemma EQP rules and semantics} for the proof).

\begin{restatable}{lemma}{EPQRulesAndSemantics}
\label{lemma:EPQ rules and semantics}
$\textit{Check-PQ-Seq}(e)=\mathsf{true}$ iff $e\in \seqPQ$, for every data-differentiated sequential execution $e$.
\end{restatable}

\subsection{Checking Linearizability of Data-Differentiated Concurrent Executions}\label{ssec:conc_exec}

The extension of $\textit{Check-PQ-Seq}$ to concurrent executions, checking whether they are linearizable w.r.t. $\seqPQ$, is obtained by replacing every predicate $\Gamma\mathsf{\text{-}Seq}$ with
\begin{align*}
\Gamma\mathsf{\text{-}Conc}(e,\alpha) = \mathsf{true}\mbox{ iff there exists a sequential execution $s$ such that $e\sqsubseteq s$ and $\Gamma\mathsf{\text{-}Seq}(s,\alpha)$}
\end{align*}
for each $\Gamma\in \{\mathsf{EmptyRemove}, \mathsf{UnmatchedMaxPriority}, \mathsf{MatchedMaxPriority}\}$. Let $\textit{Check-PQ-Conc}$ denote the thus obtained procedure (we assume recursive calls are modified accordingly).

%
%
%

The following lemma states the correctness of $\textit{Check-PQ-Conc}$. Completeness follows easily from the properties of $\seqPQ$. Thus, if $\textit{Check-PQ-Conc}(e) = \textit{false}$, then there exists a set $D$ of values s.t. either  $\mathsf{EmptyRemove\text{-}Conc}(e \vert D)$ is false, or $\mathsf{UnmatchedMaxPriority\text{-}Conc}(e \vert D,x)$ is false for all the values $x$ of maximal priority that are not removed (and there exists at least one such value), or $\mathsf{MatchedMaxPriority\text{-}Conc}(e \vert D,x)$ is false for all the values $x$ of maximal priority (and these values are all removed in $e \vert D$). It can be easily seen that we get $e \vert D\not\sqsubseteq \seqPQ$ in all cases, which by the closure under projection of $\seqPQ$ implies, $e \not\sqsubseteq \seqPQ$ (since every linearization of $e$ includes as a subsequence a linearization of $e \vert D$).

\begin{restatable}{lemma}{ConCheckEPQIsCorrect}
\label{lemma:con-check-EPQ is correct}
$\textit{Check-PQ-Conc}(e)=\mathsf{true}$ iff $e \sqsubseteq \seqPQ$, for every data-differentiated execution $e$.
\end{restatable}

Proving soundness is highly non-trivial and one of the main technical contributions of this paper (see Appendix~\ref{sec:appendix subsection proof of lemma con-check-EPQ is correct} for a complete proof). The main technical difficulty is showing that for any execution $e$, any linearization of $e\setminus x$ for some maximal priority value $x$ can be extended to a linearization of $e$ provided that $\mathsf{UnmatchedMaxPriority}$ or $\mathsf{MatchedMaxPriority}$ holds (depending on whether there are values with the same priority as $x$ in $e$ which are not removed).
%
%
%
%
%
%

We explain the proof of this property on the execution $e$ in \figurename~\ref{fig:concurrent execution for EPQ1}(a) where $p_1 \prec p$, $p_1 \prec p_2$, and the predicate $\mathsf{Has\text{-}MatchedMaxPriority}(e)$ holds. Assume that there exist two sequential executions $l$ and $l'$ such that $e \sqsubseteq l=u \cdot \textit{put}(x,p) \cdot v \cdot \textit{rm}(x) \cdot w$, $\mathsf{MatchedMaxPriority\text{-}Seq}(l,x)$ holds, and $e \setminus x \sqsubseteq l' \in \seqPQ$. Let $u=\epsilon$, $w$ be any sequence containing the set of operations $\textit{put}(z_2,p_2)$ and $\textit{rm}(z_1)$ (we distinguish them by adding the suffix $-w$ to their name, e.g., $\textit{rm}(z_1)-w$), and $v$ be any sequence containing the remaining operations. In general, the linearization $l'$ can be defined by choosing for each operation, a point in time between its call and return actions, called \emph{linearization point}. The order between the linearization points defines the sequence $l'$. \figurename~\ref{fig:concurrent execution for EPQ1}(a) draws linearization points for the operations in $e \setminus x$ defining the linearization $l'$~\footnote{In general, there may exist multiple ways of choosing linearization points to define the same linearization. Our construction is agnostic to this choice.}.
We show how to construct a sequence $l''= l''_1 \cdot \textit{put}(x,p) \cdot l''_2 \cdot \textit{rm}(x) \cdot l''_3\in\seqPQ$ such that $e \sqsubseteq l''$.
\begin{itemize}
\setlength{\itemsep}{0.5pt}
\item[-] An operation is called $p$-comparable (resp., $p$-incomparable) when it receives as argument a value of priority comparable to $p$ (resp., incomparable to $p$). We could try to define $l''_1$, $l''_2$ and $l''_3$ as the projection of $l'$ to the operations in $u$, $v$ and $w$, respectively. However, this is incorrect, since $\mathsf{MatchedMaxPriority\text{-}Seq}(l,x)$ imposes no restriction on  $p$-incomparable operations in $u \cdot v$, and thus, there is no guarantee that the projection of $l'$ to $p$-incomparable operations in $u \cdot v$ is correct. In this example, this projection is $\textit{put}(z_1,p_2) \cdot \textit{rm}(z_2)$ which is incorrect.

\item[-] We define the sets of operations $U'$, $V'$ and $W'$ such that $l''_1$, $l''_2$ and $l''_3$ are the projections of $l'$ to $U'$, $V'$, and $W'$, respectively. This is done in two steps:

\item[-] The first step is to define $W'$. The $p$-comparable operations in $W'$ are the same as in $w$. To identify the $p$-incomparable operations in $W'$, we search for a $p$-incomparable operation $o$ which either happens before some $p$-comparable operation in $w$, or whose linearization point occurs after $\textit{ret}(\textit{rm},x)$. We add to $W'$ the operation $o$ and all the $p$-incomparable operations occurring after $o$ in $l'$. In this example, $o$ is $\textit{rm}(z_1)$ and the only $p$-incomparable operation occurring after $o$ in $l'$ is $\textit{rm}(z_2)$ (they are surrounded by boxes in the figure). In this process, whether a $p$-incomparable operation is in $W'$ or not only relies on whether it is before or after such an $o$ in $l'$.

\item[-] The second step is to define $U'$ and $V'$. $U'$ contains two kinds of operations: (1) operations whose linearization points are before $\textit{ret}(\textit{put},x,p)$, and (2) other $\textit{put}$ operations with priority $p$. $V'$ contains the remaining operations. In this example, $U'$ contains $\textit{put}(z_1,p_2)$ and $\textit{put}(x_2,p)$.

\item[-] In conclusion, we have that $l''_1 = \textit{put}(z_1,p_2) \cdot \textit{put}(x_2,p)$, $l''_2 = \textit{put}(z_2,p_2) \cdot \textit{rm}(x_2) \cdot \textit{put}(y_1,p_1) \cdot \textit{rm}(y_1)$, and $l''_3 = \textit{rm}(z_1) \cdot \textit{rm}(z_2)$. \figurename~\ref{fig:concurrent execution for EPQ1}(b) draws linearization points for each operation in $e$ defining the linearization $l''$.
\end{itemize}


%
%
\begin{figure}[t]
  \centering
  \includegraphics[width=1 \textwidth]{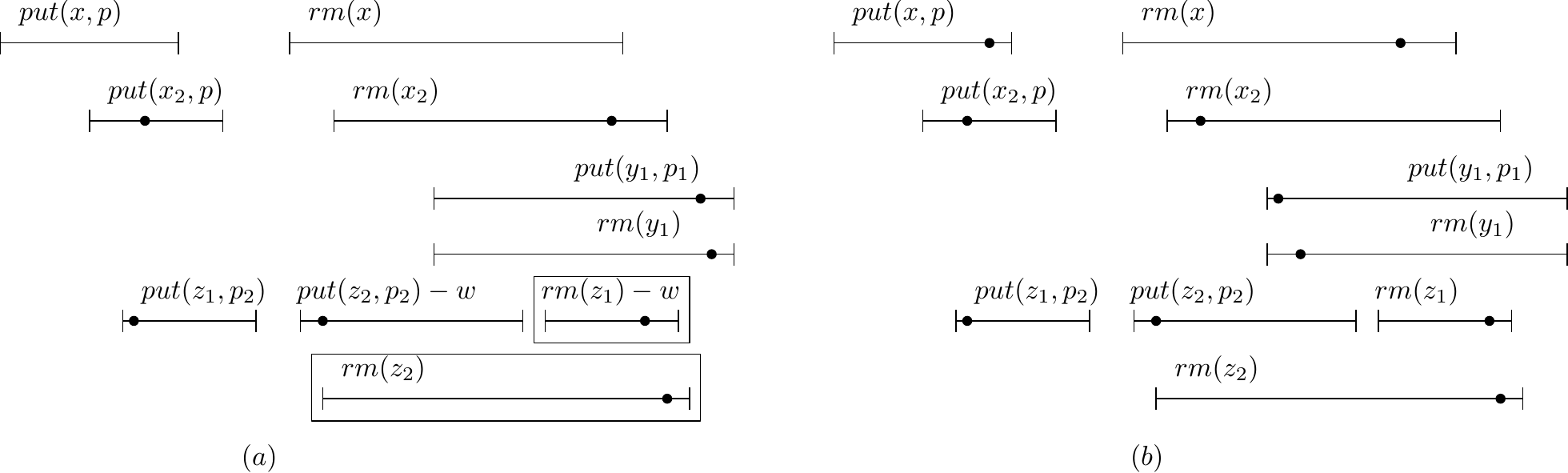}
  \caption{The process of obtaining linearization of $e$}
  \label{fig:concurrent execution for EPQ1}
\end{figure}

Section~\ref{sec:co-regular of extended priority queues} introduces a characterization of concurrent priority queue violations using a set of \emph{non-recursive} automata (i.e., whose states consist of a fixed number of registers) whose standard synchronized product is equivalent to $\textit{Check-PQ-Conc}$ (modulo renaming of values which is possible by data-independence). Since $\seqPQ$ is closed under projection (Lemma~\ref{lem:closure_proj}), the recursion in $\textit{Check-PQ-Conc}$ can be eliminated by checking that each projection of a given execution $e$ passes a non-recursive version of $\textit{Check-PQ-Conc}$ where every recursive call is replaced by $\mathsf{true}$. More precisely, every occurrence of $\text{{\bf return}}\ \textit{Check-PQ-Conc}$ is replaced by  $\text{{\bf return}}\ \mathsf{true}$. Let $\textit{Check-PQ-Conc-}$ $\textit{NonRec}$ be the thus obtained procedure.

\begin{restatable}{lemma}{EPQasMultiInMRpriforHistory}
\label{lemma:EPQ as multi in MRpri for history}
Given a data-differentiated execution $e$, $e \sqsubseteq \seqPQ$ if and only if for each $e' \in \textit{proj}(e)$, $\textit{Check-PQ-Conc-NonRec}(e')$ returns $\mathsf{true}$.
\end{restatable}

\section{Reducing Linearizability of Priority Queues to Reachability}
\label{sec:co-regular of extended priority queues}

We show that the set of executions for which some projection fails the test $\textit{Check-PQ-Conc-NonRec}$ can be characterized using a set of register automata, modulo a value renaming. The possibility of renaming values (which is complete for checking data independent implementations) allows to simplify the reasoning about projections. Thus, we assume that all the operations which are not in the projection failing this test use the same distinguished value $\top$, different from those used in the projection. Then, it is enough to find an automata characterization for the set of executions $e$ for which $\textit{Check-PQ-Conc-NonRec}$ fails, or equivalently, for which one
of the following three formulas is false:
\begin{align*}
\hspace{-5mm}
\Gamma(e) := \mathsf{Has\text{-}}\Gamma(e) \Rightarrow \exists \alpha.\ \Gamma\mathsf{\text{-}Conc}(e,\alpha)\mbox{ with $\Gamma\in \{\mathsf{EmptyRemove}$, $\mathsf{UnmatchedMaxPriority}$, $\mathsf{MatchedMaxPriority}\}$}
\end{align*}
Intuitively, $\Gamma(e)$ states that $e$ is linearizable w.r.t. the set of sequential executions described by $\Gamma\mathsf{\text{-}Seq}$ (provided that $\mathsf{Has\text{-}}\Gamma(e)$ holds). Therefore, by an abuse of terminology, an execution $e$ satisfying $\Gamma(e)$ is called \emph{linearizable w.r.t. $\Gamma$}, or \emph{$\Gamma$-linearizable}.
Extending the automaton characterizing executions which are not $\Gamma$-linearizable, with self-loops that allow any operation with parameter $\top$ results in an automaton satisfying the following property called \emph{$\Gamma$-completeness}.

\begin{definition}
For $\Gamma\in \{\mathsf{EmptyRemove}$, $\mathsf{UnmatchedMaxPriority}$, $\mathsf{MatchedMaxPriority}\}$, an automaton $A$ is called \emph{$\Gamma$-complete} when for each data-independent implementation $\mathcal{I}$:

$A \cap \mathcal{I} \neq \emptyset$ if and only if there exists $ e \in \mathcal{I}$ and $e' \in \textit{proj}(e)$ such that $e'$ is not $\Gamma$-linearizable.
\end{definition}

Section~\ref{ssec:aut} describes a $\mathsf{MatchedMaxPriority}$-complete automaton, the other automata being defined in Appendix~\ref{sec:appendix lemma and register automata for FIFO of single-priority executions} 
and Appendix \ref{subsec:co-regular of EPQ3}. Therefore, the following holds.

\begin{lemma}
There exists a $\Gamma$-complete automaton for each $\Gamma\in \{\mathsf{EmptyRemove}$, $\mathsf{UnmatchedMax}$ $\mathsf{Priority}$, $\mathsf{MatchedMaxPriority}\}$.
\end{lemma}

When defining $\Gamma$-complete automata, we assume that every implementation $\mathcal{I}$ behaves correctly, i.e., as a FIFO queue, when only values with the same priority are observed. More precisely, we assume that for every execution $e\in\mathcal{I}$ and every priority $p\in\mathbb{P}$, the projection of $e$ to values with priority $p$ is linearizable (w.r.t. $\seqPQ$). This property can be checked separately using register automata obtained from the finite automata in~\cite{DBLP:conf/icalp/BouajjaniEEH15} for FIFO queue (see Appendix~\ref{sec:appendix lemma and register automata for FIFO of single-priority executions} for more details). Note that this assumption excludes some obvious violations, such as a $\textit{rm}(a)$ operation happens before a $\textit{put}(a,p)$ operation, for some $p$.

Also, we consider $\Gamma$-complete automata for $\Gamma\in \{\mathsf{UnmatchedMaxPriority}, \mathsf{MatchedMaxPriority}\}$, recognizing executions which contain only one maximal priority. This is possible because any data-differentiated execution for which $\Gamma(e)$ is false has such a projection.
Formally, given a data-differentiated execution $e$ and $p$ a maximal priority in $e$, $e\vert_{\preceq p}$ is the projection of $e$ to the set of values with priorities smaller than $p$. Then,

\begin{restatable}{lemma}{priExecutionIsEnough}
\label{lemma:pri execution is enough}
Let $\Gamma\in \{\mathsf{UnmatchedMaxPriority}, \mathsf{MatchedMaxPriority}\}$ and $e$ a data-differentiated execution. Then, $e$ is $\Gamma$-linearizable iff $e\vert_{\preceq p}$ is $\Gamma$-linearizable for some maximal priority $p$ in $e$.
\end{restatable}
\begin {proof} (Sketch)
To prove the $\textit{only if}$ direction, let $e$ be a data-differentiated execution linearizable w.r.t. $l = u \cdot \textit{put}(x,p) \cdot v \cdot \textit{rm}(x) \cdot w \in \mathsf{MatchedMaxPriority}\mathsf{\text{-}Seq}(s,x)$. Since $\mathsf{MatchedMaxPriority}\mathsf{\text{-}Seq}$ $(s,x)$ imposes no restriction on the operations in $u$, $v$ and $w$ with priorities incomparable to $p$, erasing all these operations results in a sequential execution which still satisfies this property. Similarly, for $\Gamma=\mathsf{UnmatchedMaxPriority}$.

The $\textit{if}$ direction follows from the fact that if the projection of an execution to a set of operations $O_1$ has a linearization $l_1$ and the projection of the same execution to the remaining set of operations has a linearization $l_2$, then the execution has a linearization which is defined as an interleaving of $l_1$ and $l_2$ (see Appendix~\ref{sec:appendix proof of Lemma pri execution is enough} for more details).

Thus, let $e$ be an execution such that $e\vert_{\preceq p}$ is linearizable w.r.t. $l = u \cdot \textit{put}(x,p) \cdot v \cdot \textit{rm}(x) \cdot w \in \mathsf{MatchedMaxPriority}\mathsf{\text{-}Seq}(s,x)$. By the property above, we know that $e$ has a linearization $l' = u' \cdot \textit{put}(x,p) \cdot v' \cdot \textit{rm}(x) \cdot w'$, such that the projection of $l'$ to values of priority comparable to $p$ is $l$.
Since $\mathsf{MatchedMaxPriority}\mathsf{\text{-}Seq}(s,x)$ does not have a condition on values of priority incomparable to $p$, we obtain that $l' \in \mathsf{MatchedMaxPriority}\mathsf{\text{-}Seq}(s,\alpha)$.
\end {proof}

The following shows that $\Gamma$-complete automata enable an effective reduction of checking linearizability of concurrent priority queue implementations to state reachability. It is a direct consequence of the above definitions. Section~\ref{subsec:combine step-by-step linearizability and co-regular} discusses decidability results implied by this reduction.

\begin{restatable}{theorem}{ReduceEPQIntoStateReachability}
\label{lemma:reduce EPQ into state reachability}
Let $\mathcal{I}$ be a data-independent implementation, and $A(\Gamma)$ be a $\Gamma$-complete automaton for each $\Gamma$. Then,
$\mathcal{I} \sqsubseteq \seqPQ$ if and only if $\mathcal{I} \cap A(\Gamma) = \emptyset$ for all $\Gamma$.
\end{restatable}

\subsection{A $\mathsf{MatchedMaxPriority}$-complete automaton}\label{ssec:aut}

%
%
%

A differentiated execution $e$ is not $\mathsf{MatchedMaxPriority}$-linearizable when all the $\textit{put}$ operations in $e$ using the maximal priority $p$ are matched, and $e$ is not linearizable w.r.t. the set of sequential executions satisfying $\mathsf{MatchedMaxPriority\text{-}Seq}(e,x)$ for each value $x$ of priority $p$. We consider two cases depending on whether $e$ contains exactly one value with priority $p$ or at least two values. We denote by $\mathsf{MatchedMaxPriority}^{>}$ the strengthening of $\mathsf{MatchedMaxPriority}$ with the condition that all the values other than $x$ have a priority strictly smaller than $p$ (corresponding to the first case), and by $\mathsf{MatchedMaxPriority}^{=}$ the strengthening of the same formula with the negation of this condition (corresponding to the second case).
We use particular instances of register automata~\cite{DBLP:journals/tcs/KaminskiF94,DBLP:conf/icalp/Cerans94,DBLP:conf/stacs/SegoufinT11} whose states include only two registers, one for storing a priority guessed at the initial state, and one for storing the priority of the current action in the execution. The transitions can check equality or the order relation $\prec$ between the values stored in the two registers. Instead of formalizing the full class of register automata, we consider a simpler class which suffices our needs. More precisely, we consider a class of labeled transition systems whose states consist of a finite control part and a register $r$ interpreted to elements of $\mathbb{P}$. The transition labels can be one of the following:
\begin{itemize}
	\item $r=*$ for storing an arbitrary value to $r$,
	\item $\textit{call}(\textit{rm},a)$ and $\textit{ret}(\textit{rm},a)$ for reading call/return actions of a remove,
	\item $\textit{call}(\textit{put},d,g)$ where $g\in\{=r,\prec r,true\}$ is a guard, for reading a call action $\textit{call}(\textit{put},d,p)$ of a put and checking whether $p$ is either equal to or smaller than the value stored in $r$, or arbitrary,
	\item $\textit{ret}(\textit{put},d,true)$ for reading a return action $\textit{ret}(\textit{put},d,p)$ for any $p$.
\end{itemize}
The set of words accepted by such a transition system can be defined as usual.

\subsubsection{A $\mathsf{MatchedMaxPriority}^>$-complete automaton}
\label{subsec:co-regular of EPQ1Lar}

We give a typical example of an execution $e$ which is not $\mathsf{MatchedMaxPriority}^>$-linearizable in
\figurename~\ref{fig:introduce gap for EPQ1Lar}. Intuitively, this is a violation because during the whole execution of $\textit{rm}(b)$, the priority queue stores a smaller priority value (which should be removed before $b$). To be more precise, we define \emph{the interval of a value $x$} as the time interval from the return of a put $\textit{ret}(\textit{put},x,p)$ to the call of the matching remove $\textit{call}(rm,x)$, or to the end of the execution if such a call action doesn't exist. Intuitively, it represents the time interval in which a value is guaranteed to be stored into the concurrent priority queue. Concretely, for a standard indexing of actions in an execution, a time interval is a closed interval between the indexes of two actions in the execution.
In \figurename~\ref{fig:introduce gap for EPQ1Lar}, we draw the interval of each value by dashed line. Here we assume that $p_1 \prec p_4$, $p_2 \prec p_4$, and $p_3 \prec p_4$. We can not find a sequence $s$ where $e \sqsubseteq s$ and $\mathsf{MatchedMaxPriority}\mathsf{\text{-}Seq}(s,b)$ holds, since each time point from $\textit{call}(\textit{rm},b)$ to $\textit{ret}(\textit{rm},b)$ is included in the interval of some smaller priority value, and $\textit{rm}(b)$ can't take effect in the interval of a smaller priority value.

\begin{figure}[htbp]
  \centering
  \includegraphics[width=0.4 \textwidth]{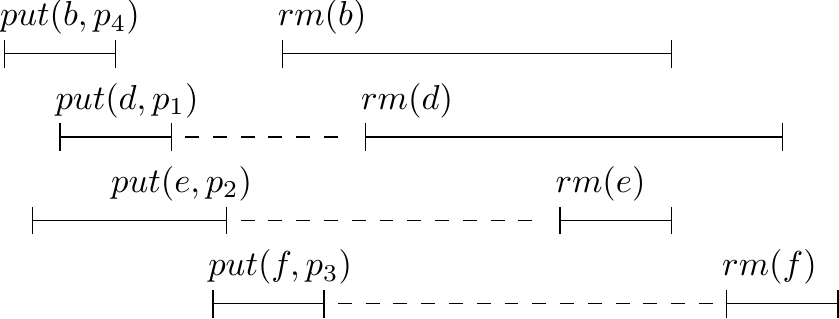}
  \caption{An execution that is not $\mathsf{MatchedMaxPriority}^{>}$-linearizable. We represent each operation as a time interval whose left, resp., right, bound corresponds to the call, resp., return action. Operations adding and removing the same value are aligned vertically.}
  \label{fig:introduce gap for EPQ1Lar}
\end{figure}

To formalize the scenario in \figurename~\ref{fig:introduce gap for EPQ1Lar} we use the notion of \emph{left-right constraint} defined below.

\begin{definition}\label{def:left-right constraint for matched put and rm operations}
Let $e$ be a data-differentiated execution which contains only one maximal priority $p$, and only one value $x$ of priority $p$ (and no $\textit{rm}(\textit{empty})$ operations).
The \emph{left-right constraint of $x$} is the graph $G$ where: 
\begin{itemize}
\item the nodes are the values occurring in $e$, 
\item there is an edge from $d_1$ to $x$, if $\textit{put}(d_1,\_) <_{\textit{hb}} \textit{put}(x,p)$ or $\textit{put}(d_1,\_) <_{\textit{hb}} \textit{rm}(x)$,
\item there is an edge from $x$ to $d_1$, if $\textit{rm}(x)<_{\textit{hb}}\textit{rm}(d_1)$ or $\textit{rm}(d_1)$ does not exists,
\item there is an edge from $d_1$ to $d_2$, if $\textit{put}(d_1,\_) <_{\textit{hb}} \textit{rm}(d_2,\_)$.
\end{itemize}
\end{definition}

The execution in \figurename~\ref{fig:introduce gap for EPQ1Lar} is not $\mathsf{MatchedMaxPriority}^>$-linearizable because the left-right constraint of the maximal priority value $b$ contains a cycle: $f \rightarrow e \rightarrow d \rightarrow b \rightarrow f$. The presence of such a cycle is equivalent to the execution not being $\mathsf{MatchedMaxPriority}^>$-linearizable (see Appendix~\ref{sec:appendix proof and definition in section co-regular of EPQ1Lar}), as indicated by the following lemma:

\begin{restatable}{lemma}{LinEqualsConstraintforEPQOneLar}
\label{lemma:Lin Equals Constraint for EPQ1Lar}
Given a data-differentiated execution $e$ where $\mathsf{Has\text{-}MatchedMaxPriority}(e)$ holds, let $p$ be its maximal priority and $\textit{put}(x,p),\textit{rm}(x)$ are only operations of priority $p$ in $e$. Let $G$ be the graph representing the left-right constraint of $x$. $e$ is $\mathsf{MatchedMaxPriority}$-linearizable, if and only if $G$ has no cycle going through $x$.
\end{restatable}

When the left-right constraint of the maximal priority value $x$ contains a cycle of the form $d_1 \rightarrow \ldots \rightarrow d_m \rightarrow x \rightarrow d_1$ for some $d_1$,$\ldots$,$d_n\in \mathbb{D}$, we say that $x$ is \emph{covered} by $d_1,\ldots,d_m$. The shape of such a cycle (i.e., the alternation between call/return actions of $\textit{put}$/$\textit{rm}$ operations) can be detected using our class of automata, the only complication being the unbounded number of values $d_1$,$\ldots$,$d_n$. However, by data independence, whenever an implementation contains such an execution it also contains an execution where all the values $d_1$,$\ldots$,$d_n$ are renamed to the same value $a$, and $x$ is renamed to $b$. Therefore, our automata can be defined over a fixed set of values $a$, $b$, and $\top$ (recall that $\top$ is used for the operations outside of the non-linearizable projection).

To define a $\mathsf{MatchedMaxPriority}^>$-complete automaton we need to consider several cases depending on the order between the call/return actions of the $\textit{put}$/$\textit{rm}$ operations that add and respectively, remove the value $b$. For example, the case where the put happens-before the remove (as in \figurename~\ref{fig:introduce gap for EPQ1Lar}) is pictured in \figurename~\ref{fig:automata APQ1Lar-1 in paper}. This automaton captures the three possible ways of ordering the first action $\textit{ret}(\textit{put},a,\_)$ w.r.t. the actions with value $b$, which are pictured in \figurename~\ref{fig:executions APQ1Lar-1 in paper} (a) (this action cannot occur after $\textit{call}(\textit{rm},b,\_)$ since $b$ must be covered by the $a$-s). The paths corresponding to these three possible orders are: $q_1 \rightarrow q_2 \rightarrow q_3 \ldots \rightarrow q_7$, $q_1 \rightarrow q_2 \rightarrow q_3 \ldots \rightarrow q_{10}$ and $q_1 \rightarrow q_9 \rightarrow q_{10} \ldots \rightarrow q_7$. In \figurename~\ref{fig:executions APQ1Lar-1 in paper}, we show all the four orders of the call/return actions of adding and removing $b$, and also possible orders of the first $\textit{ret}(\textit{put},a,\_)$ w.r.t the actions with value $b$. In Appendix \ref{sec:appendix proof and definition in section co-regular of EPQ1Lar}, three register automata is constructed according to the cases of \figurename~\ref{fig:executions APQ1Lar-1 in paper} (b), (c) and (d), respectively.

%
%

\begin{figure}[htbp]
  \centering
  \includegraphics[width=1 \textwidth]{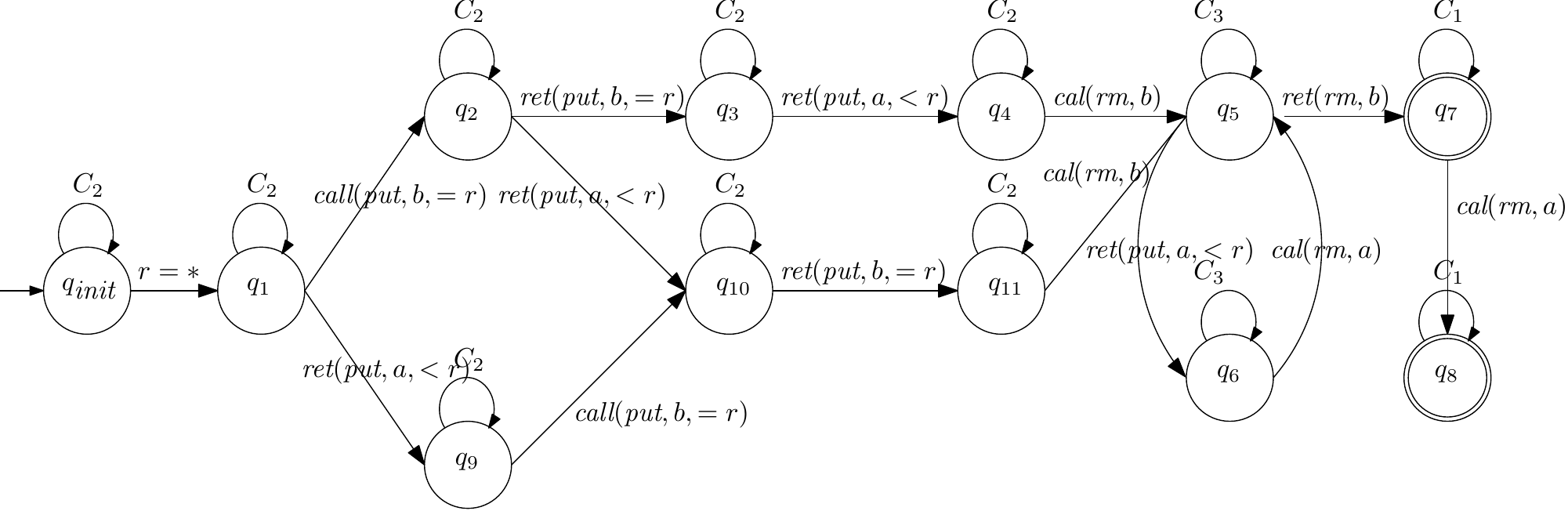}
  \caption{Register automaton $\mathcal{A}_{\textit{l-lar}}^1$. We use the following notations: $C_1 = C \cup \{ \textit{ret}(\textit{rm},a) \}$, $C_2 = C \cup \{ \textit{call}(\textit{put},a,=r) \}$, $C_3 = C_2 \cup \{ \textit{ret}(\textit{rm},a) \}$, where $C = \{ \textit{call}(\textit{put},\top,\textit{true}),\textit{ret}(\textit{put},\top,\textit{true}), \textit{call}(\textit{rm},d)$, $\textit{ret}(\textit{rm},d),\textit{call}(\textit{rm},\textit{empty}),\textit{ret}(\textit{rm},\textit{empty}) \}$.}
  \label{fig:automata APQ1Lar-1 in paper}
\end{figure}

\begin{figure}[htbp]
  \centering
  \includegraphics[width=1 \textwidth]{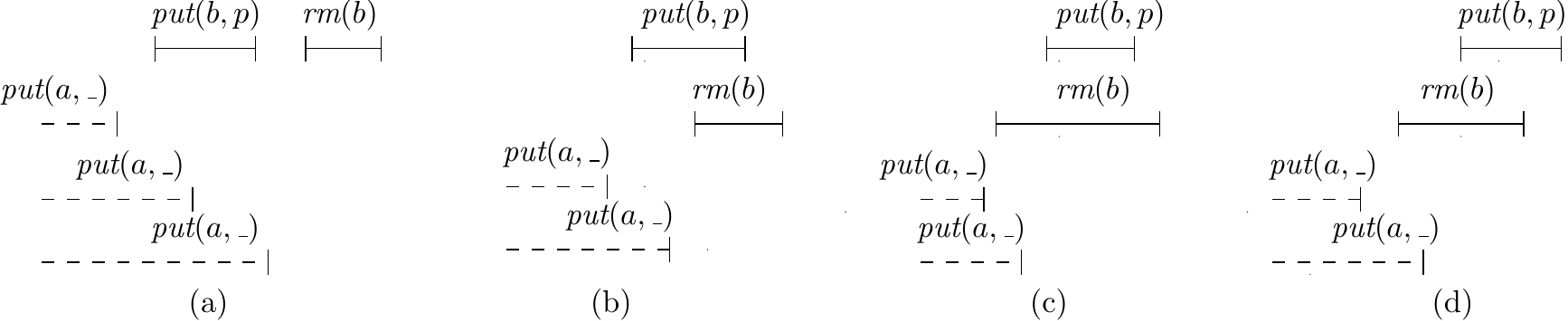}
  \caption{Four cases of ordering actions with value $b$, and possible orders of the first $\textit{ret}(\textit{put},a,\_)$ w.r.t the actions with value $b$}
  \label{fig:executions APQ1Lar-1 in paper}
\end{figure}

%

\subsubsection{A $\mathsf{MatchedMaxPriority}^=$-complete automaton}
\label{subsec:co-regular of EPQ1Equal}

When an execution contains at least two values of maximal priority, the acyclicity of the left-right constraints (for all the maximal priority values) is not enough to conclude that the execution is $\mathsf{MatchedMaxPriority}$-linearizable.
Intuitively, there may exist a value $a$ which is added before another value $b$ such that all the possible linearization points of $\textit{rm}(b)$ are disabled by the position of $\textit{rm}(a)$ in the happens-before. We give an example of such an execution $e$ in \figurename~\ref{fig:introduce pb order}. This execution  is not linearizable w.r.t. $\mathsf{MatchedMaxPriority}$ (or $\mathsf{MatchedMaxPriority}^{=}$) even if
neither $a$ nor $b$ are covered by values with smaller priority. 
Since $\textit{put}(a,p_4) <_{\textit{hb}} \textit{put}(b,p_4)$ and values of the same priority are removed in FIFO order, $\textit{rm}(a)$ should be linearized before $\textit{rm}(b)$ (i.e., this execution should be linearizable w.r.t. a sequence where $\textit{rm}(a)$ occurs before $\textit{rm}(b)$).
Since $\textit{rm}(b)$ cannot take effect during the interval of a smaller priority value, it could be linearized only in one of the two time intervals pictured with dotted lines in \figurename~\ref{fig:introduce pb order}. However, each of these time intervals ends before $\textit{call}(\textit{rm},a)$, and thus $\textit{rm}(a)$ cannot be linearized before $\textit{rm}(b)$.

\begin{figure}[htbp]
  \centering
  \includegraphics[width=0.6 \textwidth]{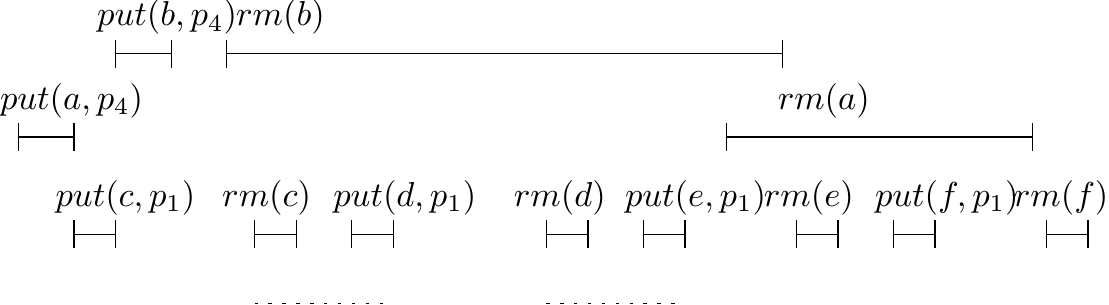}
  \caption{An execution that is not linearizable w.r.t $\mathsf{MatchedMaxPriority}^=$.}
  \label{fig:introduce pb order}
\end{figure}

To recognize the scenarios in \figurename~\ref{fig:introduce pb order}, we introduce an order $<_{\textit{pb}}$ between values which intuitively, can be thought of as ``a value $a$ is put before another value $b$''. 
Thus, given a data-differentiated execution $e$ and two values $a$ and $b$ of maximal priority, $a <_{\textit{pb}} b$ if one of the following holds: (1) $\textit{put}(a,\_) <_{\textit{hb}} \textit{put}(b,\_)$, (2) $\textit{rm}(a) <_{\textit{hb}} \textit{rm}(b)$, or (3) $\textit{rm}(a) <_{\textit{hb}} \textit{put}(b,\_)$. Sometimes we use $a <_{\textit{pb}}^A b$, $a <_{\textit{pb}}^B b$ and $a <_{\textit{pb}}^C b$ to explicitly distinguish between these three cases. Let $<_{\textit{pb}}^*$ be the transitive closure of $<_{\textit{pb}}$.

Then, to model the time intervals in which a remove operation like $\textit{rm}(b)$ in \figurename~\ref{fig:introduce pb order}, can be linearized (outside of intervals of smaller priority values) we introduce the notion of gap-point. Here we assume that the index of actions of an execution starts from $0$.

\begin{definition}\label{def:gap-point for matched put and rm operations}
Let $e$ be a data-differentiated execution which contains only one maximal priority $p$, and $\textit{put}(x,p)$ and $\textit{rm}(x)$ two operations in $e$. An index $i\in [0,|e|-1]$ is a \emph{gap-point of $x$} if $i$ is greater than or equal to the index of both $\textit{call}(\textit{put},x,p)$ and $\textit{call}(\textit{rm},x)$, smaller than the index of $\textit{ret}(\textit{rm},x)$, and it is not included in the interval of some value with priority smaller than $p$.
\end{definition}

The case of \figurename~\ref{fig:introduce pb order} can be formally described as follows: $a <_{\textit{pb}}^* b$ while the right-most gap-point of $b$ is before $\textit{call}(\textit{rm},a)$ or $\textit{call}(\textit{put},a,p_4)$. The following lemma states that these conditions are enough to characterize non-linearizability w.r.t $\mathsf{MatchedMaxPriority}^{=}$ (see Appendix~\ref{sec:appendix proof and definition in section co-regular of EPQ1Equal}).

\begin{restatable}{lemma}{EPQOneEqualAsPBandGP}
\label{lemma:EPQ1Equal as pb order and gap-point}
Let $e$ be a data-differentiated execution which contains only one maximal priority $p$ such that $\mathsf{Has\text{-}MatchedMaxPriority}(e)$ holds.
Then, $e$ is not linearizable w.r.t $\mathsf{MatchedMaxPriority}^{=}$ iff $e$ contains two values $x$ and $y$ of maximal priority $p$ such that $y <_{\textit{pb}}^* x$, and the rightmost gap-point of $x$ is strictly smaller than the index of $\textit{call}(\textit{put},y,p)$ or $\textit{call}(\textit{rm},y)$.
\end{restatable}

To characterize violations to linearizability w.r.t. $\mathsf{MatchedMaxPriority}^{=}$ using an automaton that tracks a bounded number of values, we show that the number of values needed to witness that $y <_{\textit{pb}}^* x$ for some $x$ and $y$ is bounded.
%
%

\begin{restatable}{lemma}{OBOrderHasBoundedLength}
\label{lemma:ob order has bounded length}
Let $e$ be a data-differentiated execution such that $a <_{\textit{pb}} a_1 <_{\textit{pb}} \ldots <_{\textit{pb}} a_m <_{\textit{pb}} b$ holds for some set of values $a$, $a_1$,$\ldots$,$a_m$, $b$. Then, one of the following holds:
\begin{itemize}
\setlength{\itemsep}{0.5pt}
\item[-] $a <_{\textit{pb}}^A b$, $a <_{\textit{pb}}^B b$, or $a <_{\textit{pb}}^C b$,

\item[-] $a <_{\textit{pb}}^A a_i <_{\textit{pb}}^B b$ or $a <_{\textit{pb}}^B a_i <_{\textit{pb}}^A b$, for some $i$.
\end{itemize}
\end{restatable}

\smallskip
To characterize violations to linearizability w.r.t. $\mathsf{MatchedMaxPriority}^{=}$, one has to consider all the possible orders between call/return actions of the operations on values $a$, $b$, and $a_i$ in Lemma \ref{lemma:ob order has bounded length}, and the right-most gap point of $b$. Excluding the inconsistent cases, we are left with 5 possible orders that are shown in \figurename~\ref{fig:five enumerations}, where $o$ denotes the rightmost gap-point of $b$.
For each case, we define an automaton recognizing the induced set of violations. For instance, the register automata for the case of \figurename~\ref{fig:five enumerations}(a) is shown in \figurename~\ref{fig:an enumeration and its witness automaton}. In this case, the conditions in Lemma~\ref{lemma:EPQ1Equal as pb order and gap-point} are equivalent to the fact that intuitively, the time interval from $\textit{call}(\textit{rm},a)$ to $\textit{ret}(\textit{rm},b)$ is covered by lower priority values (and thus, there is no gap-point of $b$ which occurs after $\textit{call}(\textit{rm},a)$). Using again the data-independence property, these lower priority values can all be renamed to a fixed value $d$, and the other values to a fixed value $\top$.


\begin{figure}[htbp]
  \centering
  \includegraphics[width=0.8 \textwidth]{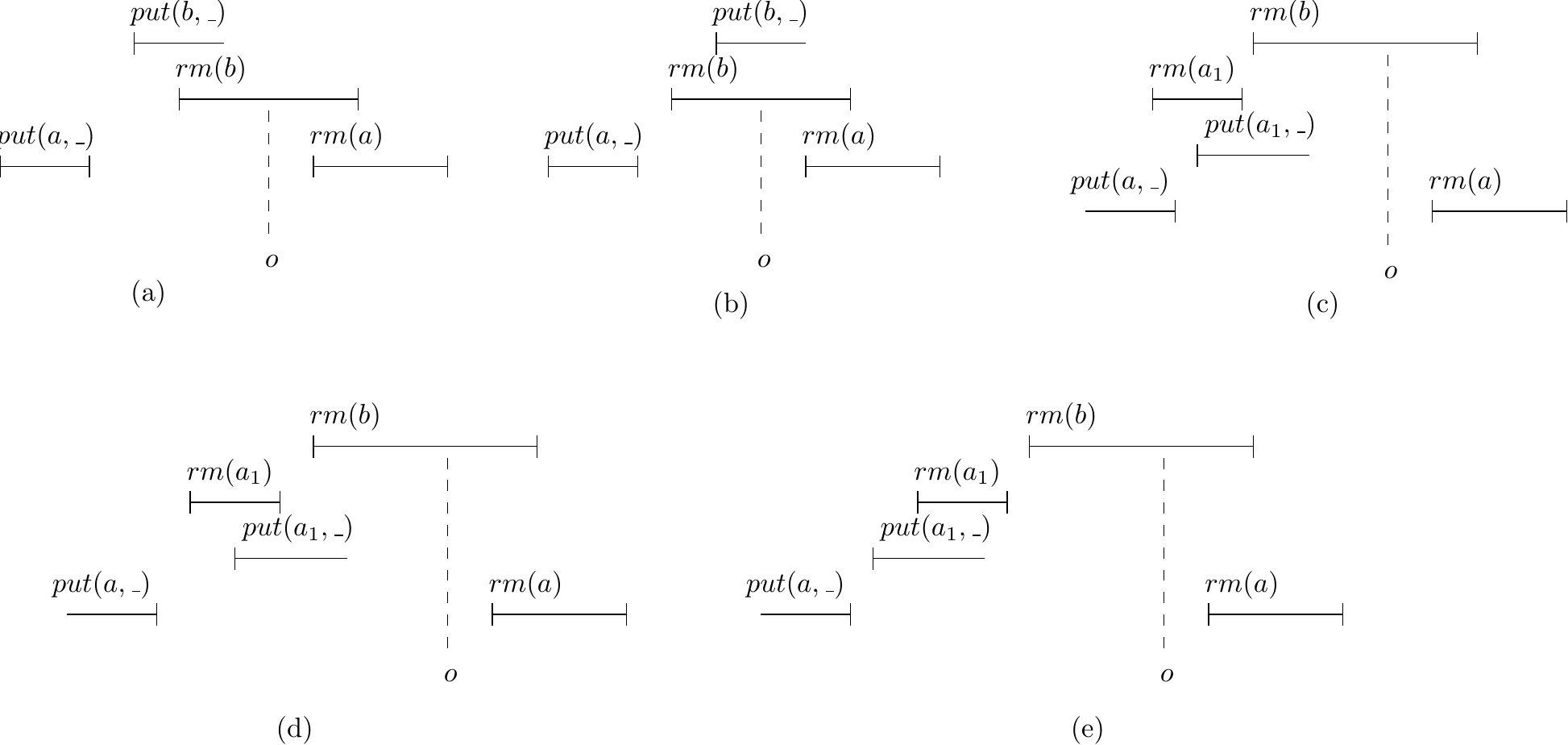}
  \caption{Five cases that need to be considered.}
  \label{fig:five enumerations}
\end{figure}

\begin{figure}[htbp]
  \centering
  \includegraphics[width=0.8 \textwidth]{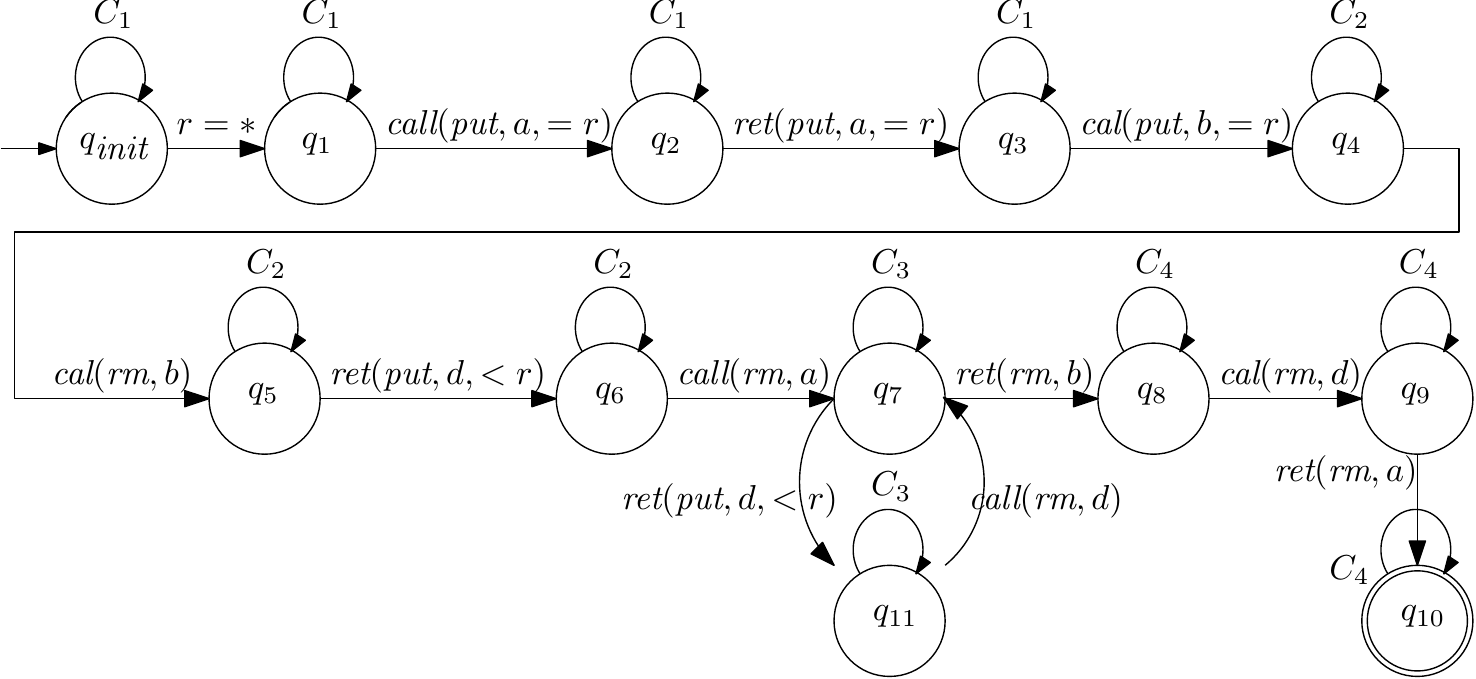}
  \caption{A case in deriving a $\mathsf{MatchedMaxPriority}^{=}$-complete automaton. We use the following notations: $C = \{ \textit{call}(\textit{put},\top,\textit{true}),\textit{ret}(\textit{put},\top,\textit{true})$, $\textit{call}(\textit{rm},\top), \textit{ret}(\textit{rm},\top),\textit{call}(\textit{rm},\textit{empty}),\textit{ret}(\textit{rm},\textit{empty}) \}$, $C_1 = C \cup \{ \textit{call}(\textit{put},d,<r) \}$, $C_2 = C_1 \cup \{ \textit{ret}(\textit{put},b,=r) \}$, $C_3 = C_2 \cup \{ \textit{ret}(\textit{rm},d) \}$, $C_4 = C \cup \{ \textit{ret}(\textit{put},b,=r), \textit{ret}(\textit{rm},d) \}$.}
  \label{fig:an enumeration and its witness automaton}
\end{figure}


%

\subsection{Decidability Result}
\label{subsec:combine step-by-step linearizability and co-regular}

We describe a class $\mathcal{C}$ of data-independent implementations for which linearizability w.r.t. $\seqPQ$ is decidable. The implementations in $\mathcal{C}$ allow an unbounded number of values but a bounded number of priorities. Each method manipulates a finite number of local variables which store Boolean values, or data values from $\mathbb{D}$. Methods communicate through a finite number of shared variables that also store Boolean values, or data values from $\mathbb{D}$. Data values may be assigned, but never used in program predicates (e.g., in the conditions of if-then-else and while statements) so as to ensure data independence. This class captures typical implementations, or finite-state abstractions thereof, e.g., obtained via predicate abstraction. Since the $\Gamma$-complete automata $A(\Gamma)$ uses a fixed set $D=\{a,b,a_1,d,e,\top\}$ of values, we have that $\mathcal{C}\cap A(\Gamma)\neq\emptyset$ for some $\Gamma$ iff $\mathcal{C}_D\cap A(\Gamma)\neq\emptyset$ where $\mathcal{C}_D$ is the subset of $\mathcal{C}$ that uses only values in $D$. 

The set of executions $\mathcal{C}_D$ can be represented by a Vector Addition Systems with States (VASS), since both values and priorities are finite, which implies that each thread and register automata can be transformed into finite-state automata. To obtain this transformation, the states of the VASS represent the global variables of $\mathcal{C}_D$, while each counter of the VASS then represents the number of threads which are at a particular control location within a method, with a certain valuation of the local variables. In this way we transform linearizability problem of priority queue into state-reachability problem of VASS, which is EXPSPACE-complete. When we consider only finite number of threads, the value of counters are bounded and then this problem is PSPACE \cite{DBLP:conf/ac/Esparza96}.

On the other hand, we can mimic transitions of VASS by a priority queue implementation. To obtain this, we keep $\textit{rm}$ unchanged, while in each $\textit{put}$ method, we first put a value into priority queue, and then either simulate one step of transitions or record one increase/decrease of a counter, similarly as that in \cite{conf/esop/BouajjaniEEH13}. In this way we transform the state-reachability problem of VASS into linearizability problem of priority queue. Based on above discussion, we have the following complexity result. The detailed proof can be found in Appendix \ref{subsec:appendix proof and definition in subsection decidability result}.


\begin{theorem}
\label{theorem:complexity of priority queue}
Verifying whether an implementation in $\mathcal{C}$ is linearizable w.r.t. $\seqPQ$ is PSPACE-complete for a fixed number of threads, and EXPSPACE-complete otherwise.
\end{theorem}

%
%
%
%

\section{Related work}\label{sec:related}

The theoretical limits of checking linearizability have been investigated in previous works. 
Checking linearizability of a single execution w.r.t. an arbitrary ADT is NP-complete~\cite{journals/siamcomp/GibbonsK97} while checking linearizability of all the executions  
of a finite-state implementation w.r.t. an arbitrary ADT 
specification (given as a regular language) is EXPSPACE-complete when the number of program 
threads is bounded~\cite{journals/iandc/AlurMP00,netys-lin}, and
undecidable otherwise~\cite{conf/esop/BouajjaniEEH13}. 

Existing automated methods for proving linearizability of a concurrent object
implementation are also based on reductions to safety
verification, e.g.,~\cite{conf/tacas/AbdullaHHJR13, conf/concur/HenzingerSV13,
conf/cav/Vafeiadis10}. The approach in~\cite{conf/cav/Vafeiadis10} considers
implementations where 
operations' \emph{linearization points}
are 
manually specified.
Essentially, this approach instruments the
implementation with ghost variables simulating the ADT specification at
linearization points. This approach is incomplete since not all implementations
have fixed linearization points. Aspect-oriented
proofs~\cite{conf/concur/HenzingerSV13} reduce linearizability to the
verification of four simpler safety properties. However, this approach has only
been applied to queues, and has not produced a fully automated
and complete proof technique. The work in~\cite{Dodds:2015:SCT:2676726.2676963} proves 
linearizability of stack implementations with an automated proof assistant. 
Their approach does not lead to full automation however, e.g.,~by reduction to 
safety verification.

Our previous work~\cite{DBLP:conf/icalp/BouajjaniEEH15}
shows that checking linearizability of finite-state implementations of concurrent queues and stacks is decidable.
Roughly, we follow the same schema: the recursive procedure in Section~\ref{ssec:seq_exec} is similar to the inductive rules in~\cite{DBLP:conf/icalp/BouajjaniEEH15}, and its extension to concurrent executions in Section~\ref{ssec:conc_exec} corresponds to the notion of step-by-step linearizability in~\cite{DBLP:conf/icalp/BouajjaniEEH15}. Although similar in nature, defining these procedures and establishing their correctness require proof techniques which are specific to the priority queue semantics. The order in which values are removed from a priority queue is encoded in their priorities which come from an unbounded domain, and not in the happens-before order as in the case of stacks and queues. Therefore, the results we introduce in this paper cannot be inferred from those in~\cite{DBLP:conf/icalp/BouajjaniEEH15}. At a technical level, characterizing the priority queue violations requires a more expressive class of automata (with registers) than the finite-state automata in~\cite{DBLP:conf/icalp/BouajjaniEEH15}.

%
%


  \bibliography{violin}
\bibliographystyle{plainurl}

%
%
%
%
%
%
%
%
%
%
%
\newpage

\appendix

\section{Definition of $\seqPQ$ and Proof of Lemma \ref{lemma:EPQ rules and semantics}}
\label{sec:appendix definition of seqPQ and proof of Lemma EQP rules and semantics}

A $\textit{labelled transition system}$ ($LTS$) is a tuple $\mathcal{A}=(Q,\Sigma,\rightarrow,q_0)$, where $Q$ is a set of states, $\Sigma$ is an alphabet of transition labels, $\rightarrow\subseteq Q\times\Sigma\times Q$ is a transition relation and $q_0$ is the initial state.

Let us model priority queue as an LTS $\textit{PQ} = (Q,\Sigma,\rightarrow,q_0)$ as follows:

\begin{itemize}
\setlength{\itemsep}{0.5pt}
\item[-] Each state of $Q$ is a function from $\mathbb{P}$ into a finite sequence over $\mathbb{D}$.

\item[-] The initial state $q_0$ is a function that maps each element in $\mathbb{P}$ into $\epsilon$.

\item[-] $\Sigma = \{ \textit{put}(a,p),\textit{rm}(a),\textit{rm}(\textit{empty}) \vert a \in \mathbb{D}, p \in \mathbb{P} \}$.

\item[-] The transition relation $\rightarrow$ is defined as follows:

    \begin{itemize}
    \setlength{\itemsep}{0.5pt}
    \item[-] $q_1 \xrightarrow{\textit{put}(a,p)} q_2$, if $q_1$ maps $p$ into some finite sequence $l$, and $q_2$ is the same as $q_1$, except for $p$, where it maps $p$ into $a \cdot l$.

    \item[-] $q_1 \xrightarrow{\textit{rm}(a)} q_2$, if $q_1$ maps $p$ into $l \cdot a$ for some finite sequence $l$, and $q_2$ is the same as $q_1$, except for $p$, where it maps $p$ into $l$. We also require that for each priority $p'$ such that $p' \prec p$, $q_1$ and $q_2$ map $p'$ into $\epsilon$.

    \item[-] $q_1 \xrightarrow{\textit{rm}(\textit{empty})} q_2$, if $q_1 = q_2$, and they maps each element in $\mathbb{P}$ into $\epsilon$.
    \end{itemize}
\end{itemize}

A path of an LTS is a finite transition sequence $q_0\xrightarrow{\beta_1}q_1\overset{\beta_2}{\longrightarrow}\ldots\overset{\beta_k}{\longrightarrow}q_k$ for $k\geq 0$, where $q_0$ is the initial state of the LTS. A trace of an LTS is a finite sequence $\beta_1 \cdot \beta_2 \cdot \ldots \cdot \beta_k$, where $k \geq 0$ if there exists a path $q_0\overset{\beta_1}{\longrightarrow}q_1\overset{\beta_2}{\longrightarrow}\ldots\overset{\beta_k}{\longrightarrow}q_k$ of the LTS. Let $\seqPQ$ be the set of traces of $\textit{PQ}$. The following lemma states that $\seqPQ$ is indeed the set of sequences obtained by renaming sequences accepted by $\textit{Check-PQ-Seq}$. Given $\Gamma\in \{\mathsf{EmptyRemove}, \mathsf{Unmatched-}$ $\mathsf{MaxPriority}, \mathsf{MatchedMaxPriority}\}$, let us use $l_2 \xrightarrow{\Gamma} l_1$ to mean that when we use $\textit{Check-PQ-Seq}$ to check $l_2$, we choose the branch of $\mathsf{Has\text{-}\Gamma}$ and finally recursively call $\textit{Check-PQ-Seq}$ to check $l_1$.


$\newline$

{\noindent \bf Lemma \ref{lemma:EPQ rules and semantics}}: $\textit{Check-PQ-Seq}(e)=\mathsf{true}$ iff $e\in \seqPQ$, for every data-differentiated sequential execution $e$.

\begin {proof}

Let $\textit{SeqPQ}_f$ be the set of data-differentiated sequences, such that $e \in \textit{SeqPQ}_f$, if $\textit{Check-PQ-}$ $\textit{Seq}(e)=\mathsf{true}$. We need to prove that $\textit{SeqPQ}_f = \seqPQ_{\neq}$.
We prove $\textit{SeqPQ}_f \subseteq \seqPQ_{\neq}$ by induction.

\begin{itemize}
\setlength{\itemsep}{0.5pt}
\item[-] It is obvious that $\epsilon \in \seqPQ$.

\item[-] If $l_1 \in \seqPQ_{\neq}$ and $l_2 \xrightarrow{\mathsf{MatchedMaxPriority}} l_1$. Then we need to prove that $l_2 \in \seqPQ$. We know that $l_1 = u \cdot v \cdot w$, such that $\mathsf{MatchedMaxPriority\text{-}Seq}(l_2,x)$ holds and $l_2 = u \cdot \textit{put}(x,p) \cdot v \cdot \textit{rm}(x) \cdot w$.

    Assume that $u = \alpha_1 \cdot \ldots \cdot \alpha_i$, $v = \alpha_{\textit{i+1}} \cdot \ldots \cdot \alpha_j$ and $w = \alpha_{\textit{j+1}} \cdot \ldots \cdot \alpha_m$. Assume that $q_0 \xrightarrow{\alpha_1} q_1 \ldots \xrightarrow{\alpha_i} q_i \xrightarrow{\alpha_{\textit{i+1}}} q_{\textit{i+1}} \ldots  \xrightarrow{\alpha_j} q_j \xrightarrow{\alpha_{\textit{j+1}}} q_{\textit{j+1}} \ldots \xrightarrow{\alpha_m} q_m$ is the path of $l_1$ on $\textit{PQ}$. For each $i \leq k \leq j$, let $q'_k$ be the same as $q_k$, except that $q'_k$ maps $p$ into $x \cdot l_k$ and $q_k$ maps $p$ into $l_k$ for some finite sequence $l_k$.

    We already know that $q_0 \xrightarrow{\alpha_1} q_1 \ldots \xrightarrow{\alpha_i} q_i$, and it is obvious that $q_i \xrightarrow{\textit{put}(x,p)} q'_i$. Since (1) all $\textit{put}$ with priority $p$ is in $u$, and (2) in $u \cdot v$, only values with priority either incomparable, or less, or equal than $p$ is removed, we can see that it is safe to add to each $q_k$ ($1 \leq k \leq j$) with a newest $x$ with priority $p$. Or we can say, $q'_i \xrightarrow{\alpha_{\textit{i+1}}} q'_{\textit{i+1}} \ldots \xrightarrow{\alpha_j} q'_j$ are transitions of $\textit{PQ}$. Since $\textit{matched}_{\prec}(u \cdot v,p)$ holds, we can see that $q_j$ maps each priority that is smaller than $p$ into $\epsilon$ and maps $p$ into $\epsilon$, and $q'_j$ maps each priority that is smaller than $p$ into $\epsilon$ and maps $p$ into $x$. Then, we can see that $q'_j \xrightarrow{\textit{rm}(x)} q_j$. We already know that that $q_j \xrightarrow{\alpha_{\textit{j+1}}} q_{\textit{j+1}} \ldots \xrightarrow{\alpha_m} q_m$. Therefore, we can see that $l_2 = u \cdot \textit{put}(x,p) \cdot v \cdot \textit{rm}(x) \in \seqPQ$.

\item[-] If $l_1 \in \seqPQ_{\neq}$ and $l_2 \xrightarrow{\mathsf{UnmatchedMaxPriority}} l_1$. Then we need to prove that $l_2 \in \seqPQ$. We know that $l_1 = u \cdot v$, such that $\mathsf{UnmatchedMaxPriority\text{-}Seq}(l_2,x)$ holds and $l_2 = u \cdot \textit{put}(x,p) \cdot v$.

    Assume that $u = \alpha_1 \cdot \ldots \cdot \alpha_i$ and $v = \alpha_{\textit{i+1}} \cdot \ldots \cdot \alpha_m$. Assume that $q_0 \xrightarrow{\alpha_1} q_1 \ldots \xrightarrow{\alpha_i} q_i \xrightarrow{\alpha_{\textit{i+1}}} q_{\textit{i+1}} \ldots  \xrightarrow{\alpha_m} q_m$ is the path of $l_1$ on $\textit{PQ}$. For each $i \leq k \leq m$, let $q'_k$ be the same as $q_k$, except that $q'_k$ maps $p$ into $x \cdot l_k$ and $q_k$ maps $p$ into $l_k$ for some finite sequence $l_k$.

    We already know that $q_0 \xrightarrow{\alpha_1} q_1 \ldots \xrightarrow{\alpha_i} q_i$, and it is obvious that $q_i \xrightarrow{\textit{put}(x,p)} q'_i$. Since (1) all $\textit{put}$ with priority $p$ is in $u$, (2) in $u \cdot v$, only values with priority either incomparable, or less, or equal than $p$ is removed, we can see that it is safe to add to each $q_k$ ($1 \leq k \leq m$) with a newest $x$ with priority $p$. Or we can say, $q'_i \xrightarrow{\alpha_{\textit{i+1}}} q'_{\textit{i+1}} \ldots \xrightarrow{\alpha_m} q'_m$ are transitions of $\textit{PQ}$. Therefore, we can see that $l_2 = u \cdot \textit{put}(x,p) \cdot v \in \seqPQ$.

\item[-] If $l_1 \in \seqPQ_{\neq}$ and $l_2 \xrightarrow{\mathsf{EmptyRemove}} l_1$. Then we need to prove that $l_2 \in \seqPQ$. We know that $l_1 = u \cdot v$, such that $\mathsf{EmptyRemove\text{-}Seq}(l_2,o)$ holds, and $l_2 = u \cdot \textit{rm}(\textit{empty}) \cdot v$ with $o=\textit{rm}(\textit{empty})$.

    Assume that $u = \alpha_1 \cdot \ldots \cdot \alpha_i$ and $v = \alpha_{\textit{i+1}} \cdot \ldots \cdot \alpha_m$. Assume that $q_0 \xrightarrow{\alpha_1} q_1 \ldots \xrightarrow{\alpha_i} q_i \xrightarrow{\alpha_{\textit{i+1}}} q_{\textit{i+1}} \ldots  \xrightarrow{\alpha_m} q_m$ is the path of $l_1$ on $\textit{PQ}$.

    We already know that $q_0 \xrightarrow{\alpha_1} q_1 \ldots \xrightarrow{\alpha_i} q_i$. Since $\textit{matched}(u)$ holds, we can see that $q_i$ maps each element in $\mathbb{P}$ into $\epsilon$, and then $q_i \xrightarrow{\textit{rm}(\textit{empty})} q_i$. We already know that $q_i \xrightarrow{\alpha_{\textit{i+1}}} q_{\textit{i+1}} \ldots \xrightarrow{\alpha_m} q_m$. Therefore, we can see that $l_2 = u \cdot \textit{rm}(\textit{empty}) \cdot v \in \seqPQ$.
\end{itemize}

To prove that $\seqPQ_{\neq} \subseteq \textit{SeqPQ}_f$, we show that given $l_2 \in \seqPQ_{\neq}$, how to construct a sequence $l_1$, such that $l_2 \xrightarrow{\Gamma} l_1$ for some $\Gamma$, and $l_1 \in \seqPQ$. Based on this, we can decompose a sequence of $\seqPQ$ into $\epsilon$, and this process ensures that this sequence is in $\textit{SeqPQ}_f$. Note that from a $l_2$ we may construct more than one $l_1$, and this does not influence the correctness of our proof.

\begin{itemize}
\setlength{\itemsep}{0.5pt}
\item[-] If $\mathsf{Has\text{-}EmptyRemoves}(l_2)$: Assume that $l_2 = u \cdot \textit{rm}(\textit{empty}) \cdot v$. It is easy to see that $\textit{matched}(u)$ holds. Let $l_1 = u \cdot v$. It is easy to see that $l_2 \xrightarrow{\mathsf{EmptyRemove}} l_1$, and $\mathsf{UnmatchedMaxPriority\text{-}Seq}(l_2,o)$ holds for some $o=\textit{rm}(\textit{empty})$.

    Assume that $u = \alpha_1 \cdot \ldots \cdot \alpha_i$ and $v = \alpha_{\textit{i+1}} \cdot \ldots \cdot \alpha_m$. Since We already know that $q_0 \xrightarrow{\alpha_1} q_1 \ldots \xrightarrow{\alpha_i} q_i \xrightarrow{\textit{rm}(\textit{empty})} q'_i \xrightarrow{\alpha_{\textit{i+1}}} q_{\textit{i+1}} \ldots  \xrightarrow{\alpha_m} q_m$ is transitions of $\textit{PQ}$. It is easy to see that $q_i = q'_i$, and they map each element in $\mathbb{P}$ into $\epsilon$. Then we can see that $q_0 \xrightarrow{\alpha_1} q_1 \ldots \xrightarrow{\alpha_i} q_i \xrightarrow{\alpha_{\textit{i+1}}} q_{\textit{i+1}} \ldots  \xrightarrow{\alpha_m} q_m$ is transitions of $\textit{PQ}$, and $l_1 \in \seqPQ$.

\item[-] If $\mathsf{Has\text{-}UnmatchedMaxPriority}(l_2)$: Assume that $l_2 = u \cdot \textit{put}(x,p) \cdot v$, such that all $\textit{put}$ with priority $p$ of $u \cdot v$ is in $u$. Let $l_1 = u \cdot v$. According to construction of $\textit{PQ}$, we can see that $\mathsf{UnmatchedMaxPriority\text{-}Seq}(l_2,x)$ holds, and $l_2 \xrightarrow{\mathsf{UnmatchedMaxPriority}} l_1$.

    Assume that $u = \alpha_1 \cdot \ldots \cdot \alpha_i$ and $v = \alpha_{\textit{i+1}} \cdot \ldots \cdot \alpha_m$. We already know that $\textit{pa} = q_0 \xrightarrow{\alpha_1} q_1 \ldots \xrightarrow{\alpha_i} q_i \xrightarrow{\textit{put}(x,p)} q_{\textit{i+}} \xrightarrow{\alpha_{\textit{i+1}}} q_{\textit{i+1}} \ldots \xrightarrow{\alpha_m} q_m$ are transitions of $\textit{PQ}$. For each $\textit{i+1} \leq k \leq m$, let $q'_k$ be the same as $q_k$, except that $q_k$ maps $p$ into some $x \cdot l_k$ for some finite sequence $l_k$, and $q'_k$ maps $p$ into $l_k$. Since (1) all $\textit{put}$ with priority $p$ of $u \cdot v$ is in $u$ and (2) $p$ is one of maximal priority of $l_2$, it is safe to remove $x$ without influence other transitions of $\textit{pa}$. Or we can say, $q_0 \xrightarrow{\alpha_1} q_1 \ldots \xrightarrow{\alpha_i} q_i \xrightarrow{\alpha_{\textit{i+1}}} q'_{\textit{i+1}} \ldots \xrightarrow{\alpha_m} q'_m$ are transitions of $\textit{PQ}$. Therefore, $l_1 \in \seqPQ$.

\item[-] If $\mathsf{Has\text{-}MatchedMaxPriority}(l_2)$: Assume that $l_2 = u \cdot \textit{put}(x,p) \cdot v \cdot \textit{rm}(x) \cdot w$, such that all $\textit{put}$ with priority $p$ of $u \cdot v \cdot w$ is in $u$. Let $l_1 = u \cdot v \cdot w$. According to construction of $\textit{PQ}$, we can see that $\mathsf{MatchedMaxPriority\text{-}Seq}(l_2,x)$ holds, and $l_2 \xrightarrow{\mathsf{MatchedMaxPriority}} l_1$.

    Assume that $u = \alpha_1 \cdot \ldots \cdot \alpha_i$, $v = \alpha_{\textit{i+1}} \cdot \ldots \cdot \alpha_j$ and $w = \alpha_{\textit{j+1}} \cdot \ldots \cdot \alpha_m$. We already know that $q_0 \xrightarrow{\alpha_1} q_1 \ldots \xrightarrow{\alpha_i} q_i \xrightarrow{\textit{put}(x,p)} q_{\textit{i+}} \xrightarrow{\alpha_{\textit{i+1}}} q_{\textit{i+1}} \ldots \xrightarrow{\alpha_j} q_j \xrightarrow{\textit{rm}(x)} q_{\textit{j+}} \xrightarrow{\alpha_{\textit{j+1}}} q_{\textit{j+1}} \ldots \xrightarrow{\alpha_m} q_m$. For each $\textit{i+1} \leq k \leq j$, let $q'_k$ be the same as $q_k$, except that $q_k$ maps $p$ into $x \cdot l_k$ for some finite sequence $l_k$, and $q'_k$ maps $p$ into $l_k$. Since (1) $p$ is one of maximal priority in $l_2$, (2) $x$ is the newest value with priority $p$ in $l_2$, and (3) $x$ is not removed until $\textit{rm}(x)$, we know that whether we keep $x$ or remove it will not influence transitions from $q_{\textit{i+1}}$ to $q_j$. Then we can see that $q_0 \xrightarrow{\alpha_1} q_1 \ldots \xrightarrow{\alpha_i} q_i \xrightarrow{\alpha_{\textit{i+1}}} q'_{\textit{i+1}} \ldots \xrightarrow{\alpha_j} q'_j \xrightarrow{\alpha_{\textit{j+1}}} q_{\textit{j+1}} \ldots \xrightarrow{\alpha_m} q_m$ are transitions of $\textit{PQ}$. Therefore, $l_1 \in \seqPQ$.
\end{itemize}

This completes the proof of this lemma. \qed
\end {proof}

\section{Proofs in Section \ref{sec:priority queue and data-independence}}
\label{sec:appendix proofs in section priority queue and data-independence}


{\noindent \bf Lemma \ref{lem:closure_proj}:} $\seqPQ$ is closed under projection, i.e., $\textit{proj}(e)\subseteq \seqPQ$ for each $e\in \seqPQ$.

\begin {proof}

By Lemma \ref{lemma:EPQ rules and semantics}, $\seqPQ$ is equivalent to the set of sequences obtained by renaming sequences accepted by $\textit{Check-PQ-Seq}$. It is easy to see that for the predicates of $\textit{Check-PQ-Seq}$, if a sequential execution satisfy it, then its sub-sequence also satisfy it. For example, if $\textit{matched}(u)$ holds, then $\textit{matched}(u \vert D)$ holds for each set $D$ of values. This completes the proof of this lemma. \qed
\end {proof}


{\noindent \bf Lemma \ref{lemma:data differentiated is enough for PQ}:} A data-independent implementation $\mathcal{I}$ is linearizable w.r.t a data-independent set $S$ of sequential executions, if and only if $\mathcal{I}_{\neq}$ is linearizable w.r.t. $S_{\neq}$.

\begin {proof}

To prove the $\textit{only if}$ direction, given a data-differentiated execution $e \in \mathcal{I}_{\neq}$. By assumption, it is linearizable with respect to a sequential execution $l \in S$, and the bijection between the operations of $e$ and the operations of $l$ ensures that $l$ is differentiated and belongs to $S_{\neq}$.

To prove the $\textit{if}$ direction, given an execution $e \in \mathcal{I}$. By data independence of $\mathcal{I}$, we know that there exists $e' \in \mathcal{I}_{\neq}$ and a renaming function $r$, such that $r(e') = e$. By assumption, $e'$ is linearizable with respect to a sequential execution $l' \in S_{\neq}$. Let $l=r(l')$. By data independence of $S$ it is easy to see that $l \in S$, and it is easy to see that $e \sqsubseteq l$  using the same bijection used for $e' \sqsubseteq l'$. \qed
\end {proof}

\section{Proofs in Section \ref{sec:checking inclusion by recursive procedure}}
\label{sec:appendix in section checking lnclusion by recursive procedure}

\subsection{Definition of Step-by-Step linearizablity}
\label{sec:appendix definition of step-by-step linearizability}

We introduce the notion of step-by-step linearizability, which means that from a linearization of $e \setminus x$ that satisfy the requirements of priority queue, we can obtain a linearization of $e$ that satisfy the requirements of priority queue. Its formal definition is as follows:

\begin{definition}\label{def:step-by-step linearizability}
Given $\Gamma\in \{\mathsf{EmptyRemove}, \mathsf{UnmatchedMaxPriority}, \mathsf{MatchedMaxPriority}\}$, $\Gamma$ is step-by-step linearizability, if for each data-differentiated execution $e$ where $\Gamma\mathsf{\text{-}Conc}(e,\alpha)$ for some $\alpha$, then $e \setminus \alpha \sqsubseteq \seqPQ \Rightarrow e \sqsubseteq \seqPQ$.

$\seqPQ$ is step-by-step linearizability, if each $\Gamma$ is step-by-step linearizability.
\end{definition}

Our notion of step-by-step linearizability is inspired by the step-by-step linearizability in ~\cite{DBLP:conf/icalp/BouajjaniEEH15}.

Given a data-differentiated execution $e$, we can obtain a sequence $e'$ from $e$ by adding $\textit{put}(a,p)$ (resp., $\textit{rm}(a)$, $\textit{rm}(\textit{empty})$) between each pair of $\textit{call}(\textit{put},a,p)$ and $\textit{ret}(\textit{put},a,p)$ (resp., $\textit{call}(\textit{rm},a)$ and $\textit{ret}(\textit{rm},a)$, $\textit{call}(\textit{rm},\textit{empty})$ and $\textit{ret}(\textit{rm},\textit{empty})$). Such $e'$ is called an execution with linearization points, and we call the projection of $e'$ into $m(a,b)$ the linearization of $e$.

\subsection{Obtaining New Sequences while Ensuring that they are in $\seqPQ$}
\label{sec:obtaining new sequences while ensureing that they are in SeqPQ}

Before we prove the step-by-step linearizability of $\seqPQ$, we introduce several lemmas, which are used to ensure some sub-sequences of $\seqPQ$ still belongs to $\seqPQ$. Given a data-differentiated sequence $l$ and one of its maximal priority $p$, let $O_c(l,p)$ and $O_i(l,p)$ be the set of operations with priorities comparable with $p$ and incomparable with $p$ in $l$, respectively. Similarly we can define $D_c(l,p)$ and $D_i(l,p)$ for se of values instead of set of operations. We can see that each priority of values in $O_i(l,p)$ is either larger or incomparable with priorities of values in  $O_c(l,p)$.

The following lemma shows that if a new sequence is generated by erasing some operations in $O_c(l,p)$ while keeping the remaining $O_c(l,p)$ sub-sequences in $\seqPQ$, then this new sequence is still in $\seqPQ$. Note that this is different from projection on value.

\begin{restatable}{lemma}{EraseOcStillinEPQ}
\label{lemma:erase Oc still in EPQ}
Given a data-differentiated sequential execution $l \in \seqPQ$ and a maximal priority $p$ in $l$, where $l$ does not contain $\textit{rm}(\textit{empty})$. Let $l'$ be generated from $l$ by discarding some operations in $O_c(l,p)$, and $l' \vert_{ O_c(l,p) } \in \seqPQ$. Then, $l' \in \seqPQ$.
\end{restatable}

\begin {proof}
Let $l=o_1 \cdot \ldots \cdot o_m$, and $q_0 \xrightarrow{o_1} q_1 \ldots \xrightarrow{o_m} q_m$ be the path of $l$ in $\textit{PQ}$. Assume that $l'$ is generated from $l$ by discarding $o_{\textit{ind1}},\ldots,o_{\textit{indn}}$. Let $D$ be the set such that $D$ contains $a$, if $\textit{put}(a,\_)$ is in $o_{\textit{ind1}},\ldots,o_{\textit{indn}}$. For each $i$, let $q'_i$ be generated from $q_i$ by erasing values in $D$.

For each $q'_j$ with $j \neq \textit{ind}\_\textit{-1}$, if $o_{\textit{j+1}}$ is $\textit{put}$, then it is obvious that $q'_j \xrightarrow{o_{\textit{j+1}}} q'_{\textit{j+1}}$. Else, assume $o_{\textit{j+1}} = \textit{rm}(a)$,

\begin{itemize}
\setlength{\itemsep}{0.5pt}
\item[-] If $\textit{rm}(a) \in O_c(l,p)$: By assumption, $l' \vert_{ O_c(l,p) } \in \seqPQ$. Therefore, $a$ is in $q'_j$ and is the should-be-removed value in $O_c(l,p)$. Since each priority of values in $O_i(l,p)$ is either larger or incomparable with priorities of values in $O_c(l,p)$, we can removed $a_c$ from $q'_j$, and then $q'_j \xrightarrow{o_{\textit{j+1}}} q'_{\textit{j+1}}$.

\item[-] If $\textit{rm}(a) \in O_i(l,p)$: By assumption we know that $q_j \xrightarrow{\textit{rm}(a)} q_{\textit{j+1}}$. Since $q'_j$ contains the same $D_i(l,p)$ values as $q_j$ and $q'_j$ contains less $D_c(l,p)$ values than $q_j$, we can see that $q'_j \xrightarrow{o_{\textit{j+1}}} q'_{\textit{j+1}}$.
\end{itemize}

For each $q'_{\textit{indj-1}}$, it is easy to see that $q'_{\textit{indj-1}} = q'_{\textit{indj}}$. Therefore, we can see that $l' = o_1 \cdot \ldots o_{\textit{ind1-1}} \cdot o_{\textit{ind1+1}} \cdot \ldots \in \seqPQ$. \qed
\end {proof}

The following lemma shows that if a new sequence is generated by from some time point, erasing operations in $O_i(l,p)$, then this new sequence is still in $\seqPQ$.

\begin{restatable}{lemma}{EraseOiFromSomeTimePointStillinEPQ}
\label{lemma:erase Oi from some time point still in EPQ}
Given a data-differentiated sequential execution $l \in \seqPQ$ and a maximal priority $p$ in $l$, where $l$ does not contain $\textit{rm}(\textit{empty})$. Let $l'$ be generated from $l$ by discarding operations in $O_i(l,p)$ from some time point, then, $l' \in \seqPQ$.
\end{restatable}

\begin {proof}
Let $l=o_1 \cdot \ldots \cdot o_m$, and $q_0 \xrightarrow{o_1} q_1 \ldots \xrightarrow{o_m} q_m$ be the path of $l$ in $\textit{PQ}$. Assume that $l'$ is generated from $l$ by discarding all operations $o_i$ if (1) $o_i \in O_i(l,p)$ and (2) $i \geq k$ for a specific index $k$. Let $D$ be a set such that $a \in D$, if $\textit{put}(a,\_)$ is in $l$ and not in $l'$. For each $0 \leq i \leq m$, let $q'_i$ be generated from $q_i$ by erasing values in $D$.

Let $l'=o'_1 \cdot \ldots \cdot o'_n$, and let $f$ be a function, such that $f(i)=j$, if $o'_i = o_j$.

\begin{itemize}
\setlength{\itemsep}{0.5pt}
\item[-] We can see that $f$ maps each $0 \leq i \leq \textit{k-1}$ into $i$, and $q'_0 \xrightarrow{o_1} q'_1 \ldots \xrightarrow{o_{\textit{k-1}}} q'_{\textit{k-1}}$.

\item[-] It is easy to see that $q'_{\textit{k-1}} = q'_{f(k)-1}$, and for each $i>k$, $q'_{f(k)} = q'_{f(k+1)-1}$.

\item[-] If $o_{f(k)}$ is a $\textit{put}$ operation, then it is obvious that $q'_{f(k)-1} \xrightarrow{o_k} q'_{f(k)}$. Else, if $o_{f(k)} = \textit{rm}(a)$, we can see $a$ is in $q'_{f(k)-1}$, and since (1) $q'_{f(k)-1}$ contains the same $D_c(l,p)$ values as $q_{f(k)-1}$ and $q'_{f(k)-1}$ contains less $D_i(l,p)$ values than $q_{f(k)-1}$, and (2) each priority of values in $O_i(l,p)$ is either larger or incomparable with priorities of values in  $O_c(l,p)$, we can see that $q'_{f(k)-1} \xrightarrow{o_k} q'_{f(k)}$. Similarly we can prove the case of $o_{f(j)}$ with $j > k$.
\end{itemize}

This completes the proof of this lemma. \qed
\end {proof}

The following lemma shows that we can make $\textit{put}$ with maximal priority to happen earlier.

\begin{restatable}{lemma}{MakePutWithMaxPriorityHappenEarlier}
\label{lemma:make put with maximal priority happen earlier}
Given a data-differentiated sequential execution $l \in \seqPQ$ and a maximal priority $p$ in $l$, where $l$ does not contain $\textit{rm}(\textit{empty})$. Let $l=l_1 \cdot l_2$. Let $l_3$ be the projection of $l_2$ into $\{ \textit{put}(\_,p) \}$, and $l_4$ be the projection of $l_2$ into other operations. Then, $l'=l_1 \cdot l_3 \cdot l_4 \in \seqPQ$.
\end{restatable}

\begin {proof}
Let $l=o_1 \cdot \ldots \cdot o_m$, and $q_0 \xrightarrow{o_1} q_1 \ldots \xrightarrow{o_m} q_m$ be the path of $l$ in $\textit{PQ}$. Let $l_1 = o_1 \cdot \ldots \cdot o_n$, let $l_2 = o_{\textit{n+1}} \cdot \ldots \cdot o_m$, let $l_3 = o'_{\textit{n+1}} \cdot \ldots \cdot o'_k$, let $l_4 = o'_{\textit{k+1}} \cdot \ldots \cdot o'_m$. Let $f$ be a function, such that $f(i)=j$, if $o'_i = o_j$.

Let $q'_i$ be constructed as follows:

\begin{itemize}
\setlength{\itemsep}{0.5pt}
\item[-] For $0 \leq i \leq n$, let $q'_i = q_i$.

\item[-] For $\textit{n+1} \leq i \leq k$, let $q'_i$ be obtained from $q_n$ by adding values in $o'_{\textit{n+1}} \cdot \ldots \cdot o'_i$ with priority $p$ and in the same order.

\item[-] For $\textit{k+1} \leq i \leq m$, let $q'_i$ be obtained from $q_{f(i)}$ by adding values which are (1) with priority $p$, (2) in $o'_{\textit{n+1}} \cdot \ldots \cdot o'_i$ and not removed by $o_1 \cdot \ldots \cdot o_{f(i)}$. The order of adding them is the same as $o'_{\textit{n+1}} \cdot \ldots \cdot o'_i$.
\end{itemize}

Then, our proof proceeds as follows:

\begin{itemize}
\setlength{\itemsep}{0.5pt}
\item[-] It is obvious that $q'_0 \xrightarrow{o_1} q'_1 \ldots \xrightarrow{o_n} q'_n$ and $q'_n \xrightarrow{o_{\textit{n+1}}} q'_{\textit{n+1}} \ldots \xrightarrow{o_k} q'_k$.

\item[-] For $q'_{\textit{k+1}}$: We already know that $q_{f(k+1)-1} \xrightarrow{o_{f(k+1)}} q_{f(k+1)}$, and it is easy to see that $q_{f(k+1)-1}$ is obtained from $q_n$ by adding values in $o_{\textit{n+1}} \cdot \ldots \cdot o_{f(k+1)-1}$.

    We can see that $q'_k$ is obtained from $q_n$ by adding values in $o'_{\textit{n+1}} \cdot \ldots \cdot o'_k$, and $q'_{k+1}$ is obtained from $q_{f(k+1)}$ by adding values which are (1) with priority $p$, (2) in $o'_{\textit{n+1}} \cdot \ldots \cdot o'_i$ and not removed by $o_1 \cdot \ldots \cdot o_{f(k+1)}$.

    \begin{itemize}
    \setlength{\itemsep}{0.5pt}
    \item[-] If $o_{f(k+1)}$ is an operation of non-$p$ values, then $q'_{k+1}$ is obtained from $q_{f(k+1)}$ by adding values in $o'_{\textit{n+1}} \cdot \ldots \cdot o'_i$. Since non-$p$ priority is either smaller or incomparable with $p$, we can see that $q'_k \xrightarrow{o_{f(k+1)}} q'_{\textit{k+1}}$.

    \item[-] Otherwise, it is only possible that $o_{f(k+1)} = \textit{rm}(a)$ for some value $a$ with priority $p$. We can see that $q'_{k+1}$ is obtained from $q_{f(k+1)}$ by adding values in $o'_{\textit{n+1}} \cdot \ldots \cdot o'_k$ and then remove $a$. Since $q_{f(k+1)-1} \xrightarrow{o_{f(k+1)}} q_{f(k+1)}$, in $q_{f(k+1)-1}$( and also in $q'_k$), there is no value with priority less than $p$, and $a$ is the first-input value of priority $p$. Therefore, we can see that $q'_k \xrightarrow{o_{f(k+1)}} q'_{\textit{k+1}}$.
    \end{itemize}

\item[-] For $q'_{\textit{k+i}}$ with $i>1$: We already know that $q_{f(k+i)-1} \xrightarrow{o_{f(k+i)}} q_{f(k+i)}$, and $q_{f(k+i)-1}$ is obtained from $q_{f(k+i-1)}$ by adding values in $o_{f(k+i-1)+1} \cdot \ldots \cdot o_{f(k+i)-1}$.

    We can see that $q'_{\textit{k+i}-1}$ (resp., $q'_{\textit{k+i}}$) is obtained from $q_{f(k+i-1)}$ (resp., $q_{f(k+i)}$) by adding values which are (1) with priority $p$, (2) in $o'_{\textit{n+1}} \cdot \ldots \cdot o'_i$ and not removed by $o_1 \cdot \ldots \cdot o_{f(k+i-1)}$ (resp., $o_1 \cdot \ldots \cdot o_{f(k+i)}$).

    \begin{itemize}
    \setlength{\itemsep}{0.5pt}
    \item[-] If $o_{f(k+i)}$ is an operation of non-$p$ values, then $q'_{k+i}$ is obtained from $q_{f(k+i)}$ by adding values which are (1) with priority $p$, (2) in $o'_{\textit{n+1}} \cdot \ldots \cdot o'_i$ and not removed by $o_1 \cdot \ldots \cdot o_{f(k+i-1)}$. Since non-$p$ priority is either smaller or incomparable with $p$, we can see that $q'_{\textit{k+i-1}} \xrightarrow{o_{f(k+i)}} q'_{\textit{k+i}}$.

    \item[-] Otherwise, it is only possible that $o_{f(k+i)} = \textit{rm}(a)$ for some value $a$ with priority $p$. We can see that $q'_{k+i}$ is obtained from $q_{f(k+i)}$ by adding values in $o'_{\textit{n+1}} \cdot \ldots \cdot o'_k$ and then remove $a$. Since $q_{f(k+i)-1} \xrightarrow{o_{f(k+i)}} q_{f(k+i)}$, in $q_{f(k+i)-1}$( and also in $q'_{\textit{k+i-1}}$), there is no value with priority less than $p$, and $a$ is the first-input value of priority $p$. Therefore, we can see that $q'_{\textit{k+i-1}} \xrightarrow{o_{f(k+i)}} q'_{\textit{k+i}}$.
    \end{itemize}
\end{itemize}

This completes the proof of this lemma. \qed
\end {proof}

The following lemma shows that if a new sequence is generated by make some $O_i(l,p)$ behaviors to happen earlier, then this new sequence is still in $\seqPQ$.

\begin{restatable}{lemma}{MakeOiHappenEarlierStillinEPQ}
\label{lemma:make Oi happen earlier still in EPQ}
Given a data-differentiated sequential execution $l \in \seqPQ$ and a maximal priority $p$ in $l$, where $l$ does not contain $\textit{rm}(\textit{empty})$. Let $l \vert_{ O_i(l,p) } = l_1 \cdot l_2$, let $l' = l_1 \cdot l_3$, where $l_3$ is the projection of $l$ into non-$l_1$ operations. Then, $l' \in \seqPQ$.
\end{restatable}

\begin {proof}
Let $l=o_1 \cdot \ldots \cdot o_m$, and $q_0 \xrightarrow{o_1} q_1 \ldots \xrightarrow{o_m} q_m$ be the path of $l$ in $\textit{PQ}$. Let $l_1 = o'_1 \cdot \ldots \cdot o'_n$, let $l_3 = o'_{\textit{n+1}} \cdot \ldots \cdot o'_m$. Let $f$ be a function, such that $f(i)=j$, if $o'_i = o_j$. Let $D$ be the set of values which are added and not removed in $o'_1 \cdot \ldots \cdot o'_n$.

Let $q'_i$ be constructed as follows:

\begin{itemize}
\setlength{\itemsep}{0.5pt}
\item[-] It is easy to see that $l_1 \in \seqPQ$, and let $q'_0 \xrightarrow{o'_1} q'_1 \ldots \xrightarrow{o'_n} q'_n$ be the path of $l_1$ in $\textit{PQ}$.

\item[-] For $\textit{n+1} \leq i \leq m$, let $q'_i$ be obtained from $q_{f(i)}$ by adding values in $D$. The order of adding them is the same as $o'_1 \cdot \ldots \cdot o'_n$.
\end{itemize}

Then, our proof proceeds as follows:

\begin{itemize}
\setlength{\itemsep}{0.5pt}
\item[-] We already know that $q'_0 \xrightarrow{o'_1} q'_1 \ldots \xrightarrow{o'_n} q'_n$.

\item[-] For $q'_{\textit{n+i}}$: We already know that $q_{f(n+i)-1} \xrightarrow{o_{f(n+i)}} q_{f(n+i)}$, and it is easy to see that $q_{f(n+i)-1}$ is obtained from $q_{f(n+i-1)}$ by adding $D$-values in $o_{f(n+i-1)+1} \cdot \ldots \cdot q_{f(n+i)-1}$.

     We can see that $q'_{\textit{n+i}-1}$ (resp., $q'_{\textit{n+i}}$) is obtained from $q_{f(n+i-1)}$ (resp., $q_{f(n+i)}$) by adding remanning values in $D$.

    \begin{itemize}
    \setlength{\itemsep}{0.5pt}
    \item[-] If $o_{f(n+1)}$ is an operation of $O_c(l,p)$ values, since priority in $O_i(l,p)$ is either larger or incomparable with priority in $O_c(l,p)$, we can see that $q'_n \xrightarrow{o_{f(n+1)}} q'_{\textit{n+1}}$.

    \item[-] Otherwise, it is only possible that $o_{f(k+1)} = \textit{rm}(a)$ for some value $a$ in $O_c(l,p)$. Since $q_{f(n+i)-1} \xrightarrow{o_{f(n+i)}} q_{f(n+i)}$, in $q_{f(n+i)-1}$( and also in $q'_{\textit{n+i-1}}$), there is no value with priority less than $p$, and $a$ is the first-input value of priority $p$. Therefore, we can see that $q'_{\textit{n+i-1}} \xrightarrow{o_{f(n+i)}} q'_{\textit{n+i}}$.
    \end{itemize}
\end{itemize}

This completes the proof of this lemma. \qed
\end {proof}

The following lemma shows that if a new sequence is generated by replacing a prefix with another one which make the priority queue has same content, then this new sequence is still in $\seqPQ$.

\begin{restatable}{lemma}{ReplaceEquivalentPrefixStillinEPQ}
\label{lemma:replace equivalent prefix still in EPQ}
Given a data-differentiated sequential execution $l \in \seqPQ$. Let $l=l_1 \cdot l_2$. Given $l_3 \in \seqPQ$. Assume that the priority queue has same content after executing $l_1$ and $l_3$. Let $l' = l_3 \cdot l_2$. Then, $l' \in \seqPQ$.
\end{restatable}

\begin {proof}
Let $l=o_1 \cdot \ldots \cdot o_m$ and let $q_0 \xrightarrow{o_1} q_1 \ldots \xrightarrow{o_m} q_m$ be the path of $l$ in $\textit{PQ}$. Let $l_1 = o_1 \cdot \ldots \cdot l_k$, $l_3 = o'_1 \cdot \ldots \cdot o'_n$ and let $q_0 \xrightarrow{o'_1} q'_1 \ldots \xrightarrow{o'_n} q'_n$ be the path of $l_3$ in $\textit{PQ}$.

By assumption we know that $q_k = q'_n$. Then it is not hard to see that $q_0 \xrightarrow{o'_1} q'_1 \ldots \xrightarrow{o'_n} q'_n \xrightarrow{o_{\textit{k+1}}} q_{\textit{k+1}} \ldots \xrightarrow{o_m} q_m$ is a path in $\textit{PQ}$, and then $l' \in \seqPQ_s$. By Lemma \ref{lemma:EPQ rules and semantics}, we know that $l' \in \seqPQ$. \qed
\end {proof}

\subsection{Proving Step-by-Step Linearizability of $\seqPQ$}
\label{sec:obtaining new sequences while ensureing that they are in SeqPQ}

With the help of Lemma \ref{lemma:erase Oc still in EPQ}, Lemma \ref{lemma:erase Oi from some time point still in EPQ}, Lemma \ref{lemma:make put with maximal priority happen earlier}, Lemma \ref{lemma:make Oi happen earlier still in EPQ} and Lemma \ref{lemma:replace equivalent prefix still in EPQ}, we can now prove that $\mathsf{MatchedMaxPriority}$ is step-by-step linearizability. Here we use $\textit{call}(o)$ and $\textit{ret}(o)$ as the call and return action of operation $o$, respectively.

\begin{restatable}{lemma}{EPQ1isStepByStepLinearizability}
\label{lemma:EPQ1 is step-by-step linearizability}
For each data-differentiated execution $e$ where $\mathsf{MatchedMaxPriority}\mathsf{\text{-}Conc}(e,x)$ for some $x$, then $e \setminus x \sqsubseteq \seqPQ \Rightarrow e \sqsubseteq \seqPQ$.
\end{restatable}

\begin {proof}

By assumption we know that there exists $l$, such that $e \sqsubseteq l$ and $\mathsf{MatchedMaxPriority\text{-}Seq}(l,x)$ holds. Let $p$ be the priority of $x$, then $l=u \cdot \textit{put}(x,p) \cdot v \cdot \textit{rm}(x) \cdot w$ for some $u$, $v$ and $w$. Let $e' = e \setminus x$. By assumption there exists sequence $l'$, such that $e' \sqsubseteq l' \in \seqPQ$. Let $e_{\textit{lp}}$ be an execution with linearization points of $e$ and the linearization points is added according to $l'$. Or we can say, $e_{\textit{lp}}$ is generated from $e$ by instrumenting linearization points, and the projection of $e_{\textit{lp}}$ into operations is $l'$. Let $l'_v$ be the shortest prefix of $l'$ that contains all operation of $u \cdot v$.

Let $U$, $V$ and $W$ be the set of operations of $u$, $v$ and $w$, respectively. Let us change $O_i(l,p)$ elements in $U$, $V$ and $W$, while keep $O_c(l,p)$ elements unchanged. We proceed by a loop: In the first round of the loop, we start from the first $O_i(l,p)$-value of $l'_v$, and let it be $o_h$,

\begin{itemize}
\setlength{\itemsep}{0.5pt}
\item[-] Case $1$: If in $e_{\textit{lp}}$, the linearization point of $o_h$ is before $\textit{ret}(\textit{rm},x)$, and no $O_c(l,p)$-value in $W$ happens before $o_h$. Then, $o_h$ is in the new version of $U \cup V$.

\item[-] Case $2$: Else, if in $e_{\textit{lp}}$, the linearization point of $o_h$ is before $\textit{ret}(\textit{rm},x)$, and there exists $O_c(l,p)$-value $o_w$ in $W$, such that $o_w <_{\textit{hb}} o_h$. Then, in $l'_v$, we put $o_h$ and all $O_i(l,p)$-value whose linearization points is after the linearization point of $o_h$ into new version of $W$, and then stop the process of changing $U$, $V$ and $W$.

\item[-] Case $3$: Else, if in $e_{\textit{lp}}$, the linearization point of $o_h$ is after $\textit{ret}(\textit{rm},x)$. Then, in $l'_v$, we put $o_h$ and all $O_i(l,p)$-value whose linearization points is after the linearization point of $o_h$ into new version of $W$, and then stop the process of changing $U$, $V$ and $W$.
\end{itemize}

In the next round of the loop, we consider the second $O_i(l,p)$-value of $l'_v$, and so on. Our process proceed, until either all element in $l'_v$ are in new version of $U \cup V$, or case $2$ or case $3$ happens and this process terminates. Let $U' \cup V'$ and $W'$ be the new version of $U \cup V$ and $W$ after the process terminates, respectively. Let $O_+$ be the set of operations that are moved into $U' \cup V'$ in the process, and let $O_-$ be the set of operations that are moved into $W'$ in the process.

Let $l'_{\textit{u'v'}}$ be the projection of $l'$ into $U' \cup V'$, let $O_x$ be the set of $\textit{put}(\_,p)$ while the value is not $x$ in $h$. Let $l''_a$ be the longest prefix of $l'_{\textit{u'v'}}$, where linearization of each operation of $l''_a$ is before $\textit{ret}(\textit{put},b)$ in $e_{\textit{lp}}$. Let $l''_d$ be the projection of $l'$ into operations of $O_x$ that are not in $l''_a$. Let $l''_1 = l''_a \cdot l''_d$. Let $l''_2$ be the projection of $l'$ into operations of $l'_{\textit{u'v'}}$ that are not in $l''_1$. Let $l''_3$ be the projection of $l'$ into operations which are not in $l'_{\textit{u'v'}}$. Let $l'' = l''_1 \cdot \textit{put}(x,p) \cdot l''_2 \cdot \textit{rm}(x) \cdot l''_3$.

To prove $e \sqsubseteq l''$, we define a graph $G$ whose nodes are the operations of $h$ and there is an edge from operation $o_1$ to $o_2$, if one of the following case holds

\begin{itemize}
\setlength{\itemsep}{0.5pt}
\item[-] $o_1$ happens-before $o_2$ in h,

\item[-] the operation corresponding to $o_1$ in $l''$ is before the one corresponding to $o_2$.
\end{itemize}

Assume there is a cycle in $G$. According the the property of interval order and the fact that the order of $l''$ is total, we know that there must exists $o_1$ and $o_2$, such that $o_1$ happens-before $o_2$ in $h$, but the corresponding operations are in the opposite order in $l''$. Then, we consider all possible case of $o_1$ and $o_2$ as follows: Let $O_a$ and $O_d$ be the set of operations in $l''_a$ and $l''_d$, respectively.

\begin{itemize}
\setlength{\itemsep}{0.5pt}
\item[-] If $o_2 \in l''_1 \wedge o_1 \in l''_1$:
    \begin{itemize}
    \setlength{\itemsep}{0.5pt}
    \item[-] If $o_1,o_2 \in O_a$ or $o_1,o_2 \in O_d$: Then $l'$ contradicts with happen before relation of $h$.

    \item[-] If $o_2 \in O_a \wedge o_1 \in O_d$: Then the order of linearization points of $e_{\textit{lp}}$ contradicts with happen before relation of $h$.
    \end{itemize}

\item[-] If $o_2 \in l''_1 \wedge o_1 = \textit{put}(x,p)$:
    \begin{itemize}
    \setlength{\itemsep}{0.5pt}
    \item[-] If $o_2 \in O_a$: This is impossible, since the linearization point of operations in $O_a$ is before $\textit{ret}(\textit{put},x,p)$ in $e_{\textit{lp}}$.

    \item[-] If $o_2 \in O_d$: Then $l$ contradicts with happen before relation of $h$.
    \end{itemize}

\item[-] If $o_2 \in l''_1 \wedge o_1 \in l''_2$:
    \begin{itemize}
    \setlength{\itemsep}{0.5pt}
    \item[-] If $o_2 \in O_a$: This violates the order of linearization point in $e_{\textit{lp}}$.

    \item[-] If $o_2 \in O_d$: According to $l$, we can see that $\textit{put}(x,p)$ does not happen before any operation in $O_x$. Then we can see that the linearization point of $o_1$ is before $\textit{ret}(\textit{put},x,p)$ and $o_1 \in O_a$. This violates that $o_1 \in l''_2$.
    \end{itemize}

\item[-] If $o_2 \in l''_1 \wedge o_1 = \textit{rm}(x)$:
    \begin{itemize}
    \setlength{\itemsep}{0.5pt}
    \item[-] If $o_2 \in U \cup V$: Then $l$ contradicts with happen before relation of $h$.

    \item[-] If $o_2 \in O_+$: This is impossible, since the linearization point of operations in $O_+$ is before $\textit{ret}(\textit{rm},x)$ in $e_{\textit{lp}}$.
    \end{itemize}

\item[-] If $o_2 \in l''_1 \wedge o_1 \in l''_3$:
    \begin{itemize}
    \setlength{\itemsep}{0.5pt}
    \item[-] If $o_1 \in W \wedge o_2 \in U \cup V$: Then $l$ contradicts with happen before relation of $h$.

    \item[-] If $o_1 \in W \wedge o_2 \in O_+$:

    \begin{itemize}
    \setlength{\itemsep}{0.5pt}
    \item[-] If $o_1 \in O_i(l,p)$: Then according to the construction process of $U' \cup V'$ and $W'$, we can see that $o_1 \in O_+$ and then $o_1 \in U' \cup V'$, which contradicts that $o_1 \in l''_3$.

    \item[-] If $o_1 \in O_c(l,p)$: Then according to the construction process of $U' \cup V'$ and $W'$, we can see that $o_2 \in O_-$, which contradicts that $o_2 \in O_+$.
    \end{itemize}

    \item[-] If $o_1 \in O_- \wedge o_2 \in U \cup V$:
        \begin{itemize}
        \setlength{\itemsep}{0.5pt}
        \item[-] If the reason of $o_1 \in O_-$ is case $2$: Let $o_h$ be as in case $2$. Then there exists $O_c(l,p)$-value $o_w \in W$, and in $e_{\textit{lp}}$, $\textit{ret}(o_w)$ is before $\textit{call}(o_h)$, the linearization point of $o_h$ is before the linearization point of $o_1$, and $\textit{ret}(o_1)$ is before $\textit{call}(o_2)$. Therefore, we can see that $o_w <_{\textit{hb}} o_2$, and then $l$ contradicts with happen before relation of $h$.

        \item[-] If the reason of $o_1 \in O_-$ is case $3$: Let $o_h$ be as in case $3$. Then in $e_{\textit{lp}}$, $\textit{ret}(\textit{rm},x)$ is before the linearization point of $o_h$, the linearization point of $o_h$ is before the linearization point of $o_1$, and $\textit{ret}(o_1)$ is before $\textit{call}(o_2)$. Therefore, we can see that $\textit{rm}(x) <_{\textit{hb}} o_2$, and then $l$ contradicts with happen before relation of $h$.
        \end{itemize}
    \item[-] If $o_1 \in O_- \wedge O_2 \in O_+$: This is impossible, since in $e_{\textit{lp}}$, the linearization points of operations in $O_+$ is before the linearization points of operations in $O_-$.
    \end{itemize}

\item[-] If $o_2 = \textit{put}(x,p) \wedge o_1 \in l''_2$: This is impossible, since in $e_{\textit{lp}}$, the linearization points of operations in $l''_2$ is after $\textit{ret}(\textit{put},x,p)$.

\item[-] If $o_2 = \textit{put}(x,p) \wedge o_1 = \textit{rm}(x)$: Then $l$ contradicts with happen before relation of $h$.

\item[-] If $o_2 = \textit{put}(x,p) \wedge o_1 \in l''_3$:
    \begin{itemize}
    \setlength{\itemsep}{0.5pt}
    \item[-] If $o_1 \in W$: Then $l$ contradicts with happen before relation of $h$.

    \item[-] If $o_1 \in O_-$:
         \begin{itemize}
         \setlength{\itemsep}{0.5pt}
         \item[-] If the reason of $o_1 \in O_-$ is case $2$: Let $o_h$ be as in case $2$. Then there exists $O_c(l,p)$-value $o_w \in W$, and in $e_{\textit{lp}}$, $\textit{ret}(o_w)$ is before $\textit{call}(o_h)$, the linearization point of $o_h$ is before the linearization point of $o_1$, and $\textit{ret}(o_1)$ is before $\textit{call}(\textit{put},x,p)$. Therefore, we can see that $o_w <_{\textit{hb}} \textit{put}(x,p)$, and then $l$ contradicts with happen before relation of $h$.

         \item[-] If the reason of $o_1 \in O_-$ is case $3$: Let $o_h$ be as in case $3$. Then in $e_{\textit{lp}}$, $\textit{ret}(\textit{rm},x)$ is before the linearization point of $o_h$, the linearization point of $o_h$ is before the linearization point of $o_1$, and $\textit{ret}(o_1)$ is before $\textit{call}(\textit{put},x,p)$. Therefore, we can see that $\textit{rm}(x) <_{\textit{hb}} \textit{put}(x,p)$, and then $l$ contradicts with happen before relation of $h$.
         \end{itemize}
    \end{itemize}

\item[-] If $o_2 \in l''_2 \wedge o_1 \in l''_2$: Then $l'$ contradicts with happen before relation of $h$.

\item[-] If $o_2 \in l''_2 \wedge o_1 = \textit{rm}(x)$: We can prove this similarly as the case of $o_2 \in l''_1 \wedge o_1 = \textit{rm}(x)$.

\item[-] If $o_2 \in l''_2 \wedge o_1 \in l''_3$: We can prove this similarly as the case of $o_2 \in l''_1 \wedge o_1 \in l''_3$.

\item[-] If $o_2 = \textit{rm}(x) \wedge o_1 \in l''_3$:
    \begin{itemize}
    \setlength{\itemsep}{0.5pt}
    \item[-] If $o_1 \in W$: Then $l$ contradicts with happen before relation of $h$.

    \item[-] If $o_1 \in O_-$:
         \begin{itemize}
         \setlength{\itemsep}{0.5pt}
         \item[-] If the reason of $o_1 \in O_-$ is case $2$: Let $o_h$ be as in case $2$. Then there exists $O_c(l,p)$-value $o_w \in W$, and in $e_{\textit{lp}}$, $\textit{ret}(o_w)$ is before $\textit{call}(o_h)$, the linearization point of $o_h$ is before the linearization point of $o_1$.

             Since $l$ is consistent with the happen before order of $h$, we can see that $\textit{call}(\textit{rm},x)$ is before $\textit{ret}(o_w)$. Therefore, we can see that the linearization point of $o_1$ is after $\textit{call}(\textit{rm},x)$, and then it is impossible that $o_1 <_{\textit{hb}} \textit{rm}(x)$.

         \item[-] If the reason of $o_1 \in O_-$ is case $3$: Let $o_h$ be as in case $3$. Then in $e_{\textit{lp}}$, $\textit{ret}(\textit{rm},x)$ is before the linearization point of $o_h$, and the linearization point of $o_h$ is before the linearization point of $o_1$. Therefore, we can see that the linearization point of $o_1$ is after $\textit{ret}(\textit{rm},x)$, and then it is impossible that $o_1 <_{\textit{hb}} \textit{rm}(x)$.
         \end{itemize}
    \end{itemize}

\item[-] If $o_2 \in l''_3 \wedge o_1 \in l''_3$: Then $l'$ contradicts with happen before relation of $h$.
\end{itemize}

Therefore, we know that $G$ is acyclic, and then we know that $h \sqsubseteq l''$.

It remains to prove that $l'' \in \seqPQ$. The process for proving $l'' \in \seqPQ$ is as follows:

\begin{itemize}
\setlength{\itemsep}{0.5pt}
\item[-] Since $l' \in \seqPQ$ and $l'_v$ is a prefix of $l'$, it is obvious that $l'_v \in \seqPQ$.

\item[-] $l'_{\textit{u'v'}}$ can be obtained from $l'_v$ as follows:
    \begin{itemize}
    \setlength{\itemsep}{0.5pt}
    \item[-] Discard $O_c(l,p)$-value that are in $W$ and keep $O_c(l,p)$-value in $U \cup V$ unchanged.

    \item[-] From some time point, discard all the $O_i(l,p)$ operations after this time point.
    \end{itemize}

From $\textit{matched}_{\prec}(u \cdot v,p)$, we can see that $O_c(l,p)$-value in $U \cdot V$ is matched, and then it is easy to see that the projection of $l'_v$ in to $O_c(l,p)$-value in $U \cdot V$ is still in $\seqPQ$. By Lemma \ref{lemma:erase Oc still in EPQ} and Lemma \ref{lemma:erase Oi from some time point still in EPQ}, we can see that $l'_{\textit{u'v'}} \in \seqPQ$.

\item[-] $l''_1 \cdot l''_2$ can be obtained from $l'_{\textit{u'v'}}$ as follows: Execute until reaching some time point $t$, then first execute all $O_x$ operations after $t$, and then execute remanning operations. By Lemma \ref{lemma:make put with maximal priority happen earlier}, we can see that $l''_1 \cdot l''_2 \in \seqPQ$.

\item[-] Let $l'_e$ be obtained from $l'$ by discarding $O_c(l,p)$-value in $U \cdot V$. By Lemma \ref{lem:closure_proj}, we can see that $l'_e \in \seqPQ$.

\item[-] Let $l'_e \vert_{O_i(l,p)}  = l'_f \cdot l'_g$, where $l'_f$ is the projection of $O_i(l,p)$-value in $U \cdot V$. Let $l'_h$ be obtained from $l'_e$ by discarding operations in $l'_f$. By Lemma \ref{lemma:make Oi happen earlier still in EPQ}, we can see that $l'_f \cdot l'_h \in \seqPQ$.

\item[-] By Lemma \ref{lem:closure_proj}, it is obvious that $l'_f \in \seqPQ$. From $\textit{matched}_{\prec}(u \cdot v,p)$, we can see that the content of priority queue after executing $l''_1 \cdot l''_2$ is the same as after executing $l'_f$. By Lemma \ref{lemma:replace equivalent prefix still in EPQ}, we can see that $l''_1 \cdot l''_2 \cdot l'_h \in \seqPQ$. It is easy to see that $l'_h = l''_3$, and then $l''_1 \cdot l''_2 \cdot l''_3 \in \seqPQ$.

\item[-] Since $\mathsf{MatchedMaxPriority\text{-}Seq}(l,x)$ holds, it is easy to see that $l'' = l''_1 \cdot \textit{put}(x,p) \cdot l''_2 \cdot \textit{rm}(x) \cdot l''_3 \in \seqPQ$.
\end{itemize}

Therefore, we prove that $e \sqsubseteq l'' \in \seqPQ$. This completes the proof of this lemma. \qed
\end {proof}

\begin{restatable}{lemma}{EPQ2isStepByStepLinearizability}
\label{lemma:EPQ2 is step-by-step linearizability}
For each data-differentiated execution $e$ where $\mathsf{UnmatchedMaxPriority}\mathsf{\text{-}Conc}(e,x)$ for some $x$, then $e \setminus x \sqsubseteq \seqPQ \Rightarrow e \sqsubseteq \seqPQ$.
\end{restatable}

\begin {proof}

By assumption we know that there exists $l$, such that $e \sqsubseteq l$ and $\mathsf{UnmatchedMaxPriority\text{-}Seq}(l,x)$ holds. Let $p$ be the priority of $x$, then $l=u \cdot \textit{put}(x,p) \cdot v$ for some $u$ and $v$. Let $e' = e \setminus x$. By assumption there exists sequence $l'$, such that $e' \sqsubseteq l' \in \seqPQ$. Let $e_{\textit{lp}}$ be an execution with linearization points of $e$ and the linearization points is added according to $l'$. Or we can say, $e_{\textit{lp}}$ is generated from $e$ by instrumenting linearization points, and the projection of $e_{\textit{lp}}$ into operations is $l'$. Let $l'_v$ be the shortest prefix of $l'$ that contains all operations of $u \cdot v$.

Let $l''_a$ be the longest prefix of $l'$ such that linearization point of each operation of $l''_a$ is before $\textit{ret}(\textit{put},x,p)$ in $e_{\textit{lp}}$. Let $O_x$ be the set of $\textit{put}(\_,p)$ while the value is not $x$ in $h$. Let $l''_s$ be the projection of $l'$ into operations of $O_x$ that are not in $l''_a$. Let $l''_1 = l''_a \cdot l''_s$. Let $l''_2$ be the projection of $l'$ into operations of $l'$ that are not in $l''_1$. Let $l'' = l''_1 \cdot \textit{put}(x,p) \cdot l''_2$.

To prove $h \sqsubseteq l''$, we define graph $G$ as in Lemma \ref{lemma:EPQ1 is step-by-step linearizability}. Assume that there is a cycle in $G$, then there must exists $o_1$ and $o_2$, such that $o_1$ happens-before $o_2$ in $h$, but the corresponding operations are in the opposite order in $l''$. Then, we consider all possible case of $o_1$ and $o_2$ as follows: Let $O_a$ and $O_s$ be the set of operations in $l''_a$ and $l''_s$, respectively.

\begin{itemize}
\setlength{\itemsep}{0.5pt}
\item[-] If $o_2 \in l''_1 \wedge o_1 \in l''_1$:
    \begin{itemize}
    \setlength{\itemsep}{0.5pt}
    \item[-] If $o_1,o_2 \in O_a$ or $o_1,o_2 \in O_s$: Then $l'$ contradicts with happen before relation of $h$.

    \item[-] If $o_2 \in O_a \wedge o_1 \in O_s$: It is not hard to see that $\textit{put}(x,p) <_{\textit{hb}} o_2$. Then, it is impossible to locate the linearization point of $o_2$ before $\textit{ret}(\textit{put},x,p)$ in $e_{\textit{lp}}$.
    \end{itemize}

\item[-] If $o_2 \in l''_1 \wedge o_1 = \textit{put}(x,p)$:
    \begin{itemize}
    \setlength{\itemsep}{0.5pt}
    \item[-] If $o_2 \in O_a$: This is impossible, since in $e_{\textit{lp}}$, the linearization point of $o_1$ is before $\textit{ret}(\textit{put},x,p)$.

    \item[-] If $o_2 \in O_s$: This is impossible, since $o_2 = \textit{put}(\_,p)$, and $l$ is consistent with happen before relation of $h$.
    \end{itemize}

\item[-] If $o_2 \in l''_1 \wedge o_1 \in l''_2$:
    \begin{itemize}
    \setlength{\itemsep}{0.5pt}
    \item[-] If $o_2 \in O_a$: This is impossible, since in $e_{\textit{lp}}$, the linearization point of operation in $l''_a$ is before the linearization point of operations in $l''_2$.

    \item[-] If $o_2 \in O_s$: Since no $\textit{put}(\_,p)$ happens before $\textit{put}(x,p)$ in $h$, $\textit{call}(o_2)$ is before $\textit{ret}(\textit{put},x,p)$. Since $o_1 <_{\textit{hb}} o_2$, we can see that $\textit{ret}(o_1)$ is before $\textit{call}(o_2)$, and then $\textit{ret}(o_1)$ is before $\textit{ret}(\textit{put},x,p)$. Then the linearization point of $o_1$ can only be before $\textit{ret}(\textit{put},x,p)$, and $o_1 \in l''_a$, which contradicts that $o_1 \in l''_2$.
    \end{itemize}

\item[-] If $o_2 = \textit{put}(x,p) \wedge o_1 \in l''_2$: Then since the linearization point of $o_1$ can only be before $\textit{ret}(\textit{put},x,p)$, we can see that $o_1 \in l''_a$, which contradicts that $o_1 \in l''_2$.

\item[-] If $o_2 \in l''_2 \wedge o_1 \in l''_2$: Then $l'$ contradicts with happen before relation of $h$.
\end{itemize}

Therefore, we know that $G$ is acyclic, and then we know that $e \sqsubseteq l''$.

It remains to prove that $l'' \in \seqPQ$. $l''_a \cdot l''_s \cdot l''_2$ can be obtained from $l'$ as follows: Execute until reaching some time point $t$, then first execute all $O_x$ operations after $t$, and then execute remanning operations. By Lemma \ref{lemma:make put with maximal priority happen earlier}, we can see that $l''_1 \cdot l''_2 = l''_a \cdot l''_s \cdot l''_2 \in \seqPQ$. Since $\mathsf{UnmatchedMaxPriority\text{-}Seq}(l,x)$ holds, it is easy to see that $l'' = l''_1 \cdot \textit{put}(x,p) \cdot l''_2 \in \seqPQ$.

Therefore, we prove that $e \sqsubseteq l'' \in \seqPQ$. This completes the proof of this lemma.\qed
\end {proof}

\begin{restatable}{lemma}{EPQ3isStepByStepLinearizability}
\label{lemma:EPQ3 is step-by-step linearizability}
For each data-differentiated execution $e$ where $\mathsf{EmptyRemove}\mathsf{\text{-}Conc}(e,x)$ for some $o=\textit{rm}(\textit{empty})$, then $e \setminus o \sqsubseteq \seqPQ \Rightarrow e \sqsubseteq \seqPQ$.
\end{restatable}

\begin {proof}

By assumption we know that there exists $l$, such that $e \sqsubseteq l$ and $\mathsf{EmptyRemove\text{-}Seq}(l,o)$ holds. Then $l=u \cdot o \cdot v$ for some $u$ and $v$. Let $e' = e \setminus o$. By assumption there exists sequence $l'$, such that $e' \sqsubseteq l' \in \seqPQ$.

Let $E_L$ be the set of operations in $u$ and $E_R$ be the set of operations in $v$. Let $l'_L = l' \vert_{E_L}$ and $l'_R = l' \vert_{E_R}$. Let sequence $l'' = l'_L \cdot o \cdot L'_R$. Since priority queue is closed under projection (Lemma \ref{lem:closure_proj}) and all the $\textit{put}$ operations and $\textit{rm}$ in $u$ are matched, we know that $l'_L \in \seqPQ$ and the the priority queue is empty after executing $l'_L$. Then we know that $l'_L \cdot \textit{rm}(\textit{empty}) \in \seqPQ$. Since $l'_R$ is obtained from $l'$ by discarding pairs of matched $\textit{put}$ and $\textit{rm}$ operations, it is easy to see that $L'_R \in \seqPQ$, and then we know that $l'' = l'_L \cdot o \cdot L'_R \in \seqPQ$.

It remains to prove that $h \sqsubseteq l''$. To prove $h \sqsubseteq l''$, we define graph $G$ as in Lemma \ref{lemma:EPQ1 is step-by-step linearizability}. Assume that there is a cycle in $G$, then there must exists $o_1$ and $o_2$, such that $o_1$ happens-before $o_2$ in $h$, but the corresponding operations are in the opposite order in $l''$. Then, we consider all possible case of $o_1$ and $o_2$ as follows:

\begin{itemize}
\setlength{\itemsep}{0.5pt}
\item[-] $o_1,o_2 \in l'_L$, or $o_1,o_2 \in l'_R$: Then $l'$ contradicts with happen before relation of $h$.

\item[-] If $o_1=o \wedge o_2 \in l'_L$, or $o_1 \in l'_R \wedge o_2 \in l'_L$, or $o_1 \in l'_R \wedge o_2 = o$, then $l$ contradicts with happen before relation of $h$.
\end{itemize}

Therefore, we know that $G$ is acyclic, and then we know that $h \sqsubseteq \seqPQ$. \qed
\end {proof}

The following lemma states that $\seqPQ$ is step-by-step linearizability, which is a direct consequence of Lemma \ref{lemma:EPQ1 is step-by-step linearizability}, Lemma \ref{lemma:EPQ2 is step-by-step linearizability} and Lemma \ref{lemma:EPQ3 is step-by-step linearizability}.


\begin{restatable}{lemma}{EPQueueisStepByStepLinearizability}
\label{lemma:EPQ is step-by-step linearizability}
$\seqPQ$ is step-by-step linearizability.
\end{restatable}

\begin {proof}
This is a direct consequence of Lemma \ref{lemma:EPQ1 is step-by-step linearizability}, Lemma \ref{lemma:EPQ2 is step-by-step linearizability} and Lemma \ref{lemma:EPQ3 is step-by-step linearizability}. \qed
\end {proof}

\subsection{Proof of Lemma \ref{lemma:con-check-EPQ is correct}}
\label{sec:appendix subsection proof of lemma con-check-EPQ is correct}


{\noindent \bf Lemma \ref{lemma:con-check-EPQ is correct}}: $\textit{Check-PQ-Conc}(e)=\mathsf{true}$ iff $e \sqsubseteq \seqPQ$, for every data-differentiated execution $e$.

\begin {proof}

To prove the $\textit{if}$ direction, given a data-differentiated $e \sqsubseteq l \in \seqPQ$. Then we being a loop as follows,

\begin{itemize}
\setlength{\itemsep}{0.5pt}
\item[-] If $\mathsf{Has\text{-}EmptyRemoves}(e)$ holds, then $l = u \cdot o \cdot v$ for some $o=\textit{rm}(\textit{empty})$. It is not hard to see that $\mathsf{EmptyRemove\text{-}Seq}(l,o)$ holds. Let $e'$ be the projection of $e$ into operations of $u \cdot v$. It is easy to see that $e' \sqsubseteq u \cdot v$ and then $e' \sqsubseteq \seqPQ$. Then we start the next round and choose $e'$ to be the ``$e$ in the next round''.

\item[-] If $\mathsf{Has\text{-}UnmatchedMaxPriority}(e)$ holds, then there exists a maximal priority $p$ in $e$, such that $p$ has unmatched $\textit{put}$. Let $l = u \cdot \textit{put}(x,p) \cdot v$, where in $v$ there is no $\textit{put}$ with priority $p$. It is not hard to see that $\mathsf{UnmatchedMaxPriority\text{-}Seq}(l,x)$ holds. Let $e'$ be the projection of $e$ into operations of $u \cdot v$. It is easy to see that $e' \sqsubseteq u \cdot v$ and then $e' \sqsubseteq \seqPQ$. Then we start the next round and choose $e'$ to be the ``$e$ in the next round''.

\item[-] If $\mathsf{Has\text{-}MatchedMaxPriority}(e)$ holds, then there exists a maximal priority $p$ in $e$, such that $p$ has only matched $\textit{put}$. Let $l = u \cdot \textit{put}(x,p) \cdot v \cdot \textit{rm}(x) \cdot w$, where in $v \cdot w$ there is no $\textit{put}$ with priority $p$. It is not hard to see that $\mathsf{MatchedMaxPriority\text{-}Seq}(l,x)$ holds. Let $e'$ be the projection of $e$ into operations of $u \cdot v \cdot w$. It is easy to see that $e' \sqsubseteq u \cdot v \cdot w$ and then $e' \sqsubseteq \seqPQ$. Then we start the next round and choose $e'$ to be the ``$e$ in the next round''.
\end{itemize}

The loop terminates when at some round $e=\epsilon$. This process ensures that $\textit{Check-PQ-Conc}(e)=\mathsf{true}$.

To prove the $\textit{only if}$ direction, given a data-differentiated execution $e$ and assume that $\textit{Check-PQ-}$ $\textit{Conc}(e)=\mathsf{true}$. Then we have $e \xrightarrow{\Gamma_1} e_1 \ldots \xrightarrow{\Gamma_m} e_m$ with $e_0=e$ and $e_m=\epsilon$. Assume that for each $i$, we found $\mathsf{\Gamma\text{-}Seq}(e_i,\alpha_i)$ holds for some $\alpha_i$. From $\mathsf{\Gamma\text{-}Seq}(e_{\textit{m-1}},\alpha_{\textit{m-1}})$ holds and $e_m = e_{\textit{m-1}} \setminus \alpha_{\textit{m-1}} \sqsubseteq \seqPQ$, by Lemma \ref{lemma:EPQ is step-by-step linearizability}, we can see that $e_{\textit{m-1}} \sqsubseteq \seqPQ$. Similarly, we can prove that $e_{\textit{m-2}},\ldots,e_0=e \sqsubseteq \seqPQ$. \qed
\end {proof}

\subsection{Proof of Lemma \ref{lemma:EPQ as multi in MRpri for history}}

{\noindent \bf Lemma \ref{lemma:EPQ as multi in MRpri for history}}: Given a data-differentiated execution $e$, $e \sqsubseteq \seqPQ$ if and only if for each $e' \in \textit{proj}(e)$, $\textit{Check-PQ-Conc-NonRec}(e')$ returns $\mathsf{true}$.


\begin {proof}

To prove the $\textit{only if}$ direction, assume that $e \sqsubseteq l \in \seqPQ$. Given $e' = e \vert_{D}$ and $l' = l \vert_{D}$, it is easy to see that $e' \sqsubseteq l'$, and by Lemma \ref{lem:closure_proj}, we can see that $l' \in \seqPQ$. Then it is obvious that $\textit{Check-PQ-Conc-NonRec}(e')$ returns $\mathsf{true}$.

To prove the $\textit{if}$ direction, assume that for each $e' \in \textit{proj}(e)$, $\textit{Check-PQ-Conc-NonRec}(e')$ returns $\mathsf{true}$. Then similarly as by Lemma \ref{lemma:con-check-EPQ is correct}, we can implies that $\textit{Check-PQ-Conc}(e)=\mathsf{true}$, and then by Lemma \ref{lemma:con-check-EPQ is correct}, we know that $e \sqsubseteq \seqPQ$. \qed
\end {proof}

\section{Proofs and Definitions in Section \ref{sec:co-regular of extended priority queues}}
\label{sec:appendix proof and definition in section co-regular of extended priority queues}

To facilitate our proof, we consider two cases of $\mathsf{MatchedMaxPriority\text{-}Seq}(e,x)$ depending on whether $e$ contains exactly one value with priority $p$ or at least two values. We denote by $\mathsf{MatchedMax-}$ $\mathsf{Priority}^{>}(e,x)$ the strengthening of $\mathsf{MatchedMaxPriority}(e,x)$ with the condition that all the values other than $x$ have a priority strictly smaller than $p$, and by $\mathsf{MatchedMaxPriority}^{=}(e,x)$ the strengthening of the same formula with the negation of this condition.

Similarly, we consider two cases of $\mathsf{UnmatchedMaxPriority\text{-}Seq}(e,x)$ depending on whether $e$ some values of priority $p$ has matched $\textit{put}$. We denote by $\mathsf{UnmatchedMaxPriority}^{>}(e,x)$ the strengthening of $\mathsf{UnmatchedMaxPriority}(e,x)$ with the condition that non value of priority $p$ has matched $\textit{put}$, and by $\mathsf{UnmatchedMaxPriority}^{=}(e,x)$ the strengthening of the same formula with the negation of this condition.

\subsection{Lemma and Register Automata For FIFO of Single-Priority Executions}
\label{sec:appendix lemma and register automata for FIFO of single-priority executions}

Before we go to investigate $\Gamma$-linearizable, we can use the result in \cite{DBLP:conf/icalp/BouajjaniEEH15} to simplify our work. \cite{DBLP:conf/icalp/BouajjaniEEH15} states that checking linearizability w.r.t queue can be reduced into checking emptiness of intersection between $\mathcal{I}$ and a set of automata. Given a data-differentiated execution $e$, let $e \vert_{i}$ be an execution generated from $e$ by erasing call and return actions of values that does not use priority $i$ (does not influence $\textit{rm}(\textit{empty})$). We call a extended priority queue execution with only one priority a single-priority execution. Let $\textit{transToQueue}(e)$ be an execution generated from $e$ by transforming $\textit{put}$ and $\textit{rm}$ into $\textit{enq}$ and $\textit{deq}$, respectively, and then discarding priorities. We can see that for each $e \in \seqPQ$ and each priority $i$, $\textit{transToQueue}(e \vert_{i})$ satisfy FIFO (first in first out) property.

Given an execution of queue, we say that it is differentiated, if each value is enqueued at most once. \cite{DBLP:conf/icalp/BouajjaniEEH15} states that, given a differentiated queue execution $e$ without $\textit{deq}(\textit{empty})$, $e$ is not linearizable with respect to queue, if one of the following cases holds for some $a,b$: (1) $\textit{deq}(b) <_{hb} \textit{enq}(b)$, (2) there are no $\textit{enq}(b)$ and at least one $\textit{deq}(b)$, (3) there are one $\textit{enq}(b)$ and more than one $\textit{deq}(b)$, and (4) $\textit{enq}(a) <_{\textit{hb}} \textit{enq}(b)$, and $\textit{deq}(b) <_{\textit{hb}} \textit{deq}(a)$, or $\textit{deq}(a)$ does not exists. For each such case, we can construct a register automata for extended priority queue.

We generate register automata $\mathcal{A}_{\textit{SinPri}}^1$ for the first case, and it is shown in \figurename~\ref{fig:automata for FIFO-1 in appendix}. Here $C_1 = \{ \textit{call}(\textit{put},a,\textit{true})$, $\textit{ret}(\textit{put},a,\textit{true}), \textit{call}(\textit{rm},a),\textit{ret}(\textit{rm},a),\textit{call}(\textit{rm},b),\textit{call}(\textit{rm},\textit{empty}),\textit{ret}(\textit{rm},\textit{empty}) \}$, $C_2$ $= C_1 \cup \{ \textit{ret}(\textit{rm},b) \}$, $C_3 = C_2 \cup \{ \textit{ret}(\textit{put},b,\textit{true}) \}$.

\begin{figure}[htbp]
  \centering
  \includegraphics[width=0.5 \textwidth]{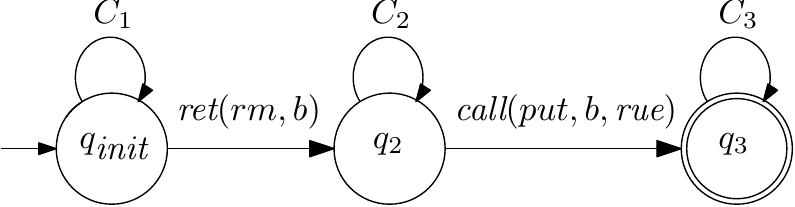}
  \caption{Automaton $\mathcal{A}_{\textit{SinPri}}^1$}
  \label{fig:automata for FIFO-1 in appendix}
\end{figure}

We generate register automata $\mathcal{A}_{\textit{SinPri}}^2$ for the second case, and it is shown in \figurename~\ref{fig:automata for FIFO-2}. Here $C_1 = \{ \textit{call}(\textit{put},a,\textit{true}),\textit{ret}(\textit{put},a,\textit{true}), \textit{call}(\textit{rm},a),\textit{ret}(\textit{rm},a),\textit{call}(\textit{rm},\textit{empty}),\textit{ret}(\textit{rm},\textit{empty}) \}$, $C_2 = C_1 \cup \{ \textit{call}(\textit{rm},b) + \textit{ret}(\textit{rm},b) \}$.

\begin{figure}[htbp]
  \centering
  \includegraphics[width=0.3 \textwidth]{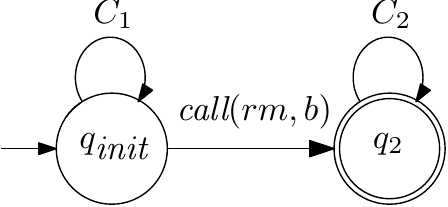}
  \caption{Automaton $\mathcal{A}_{\textit{SinPri}}^2$}
  \label{fig:automata for FIFO-2}
\end{figure}

We generate register automata $\mathcal{A}_{\textit{SinPri}}^3$ for the third case, and it is shown in \figurename~\ref{fig:automata for FIFO-3}. Here $C_1 = \{ \textit{call}(\textit{put},a,\textit{true}),\textit{ret}(\textit{put},a,\textit{true}), \textit{call}(\textit{rm},a),\textit{ret}(\textit{rm},a),\textit{call}(\textit{rm},\textit{empty}),\textit{ret}(\textit{rm},\textit{empty}) \}$, $C_2 = C_1 \cup \{ \textit{ret}(\textit{put},b,\textit{true}) \}$, $C_3 = C_2 \cup \{ \textit{ret}(\textit{rm},b) \}$, $C_4 = C_3 \cup \{ \textit{call}(\textit{rm},b) \}$, $C_5 = C_1 \cup \{ \textit{ret}(\textit{rm},b) \}$, $C_6 = C_5 \cup \{ \textit{call}(\textit{rm},b) \}$.

\begin{figure}[htbp]
  \centering
  \includegraphics[width=0.7 \textwidth]{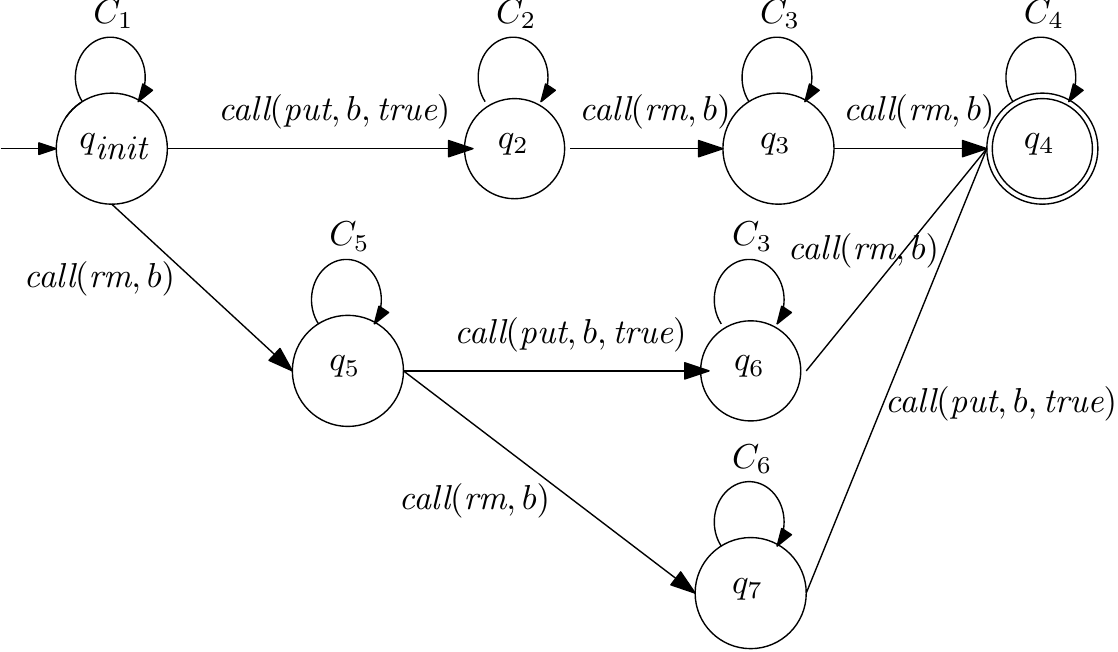}
  \caption{Automaton $\mathcal{A}_{\textit{SinPri}}^3$}
  \label{fig:automata for FIFO-3}
\end{figure}

We generate register automata $\mathcal{A}_{\textit{SinPri}}^4$ for the forth case, and it is shown in \figurename~\ref{fig:automata for FIFO-4}. Here $C_1 = C \cup \{ \textit{call}(\textit{rm},b) \}$, and $C_2 = C \cup \{ \textit{ret}(\textit{put},b,=r), \textit{call}(\textit{rm},a), \textit{ret}(\textit{rm},a) \}$, where $C = \{ \textit{call}(\textit{put},d,\textit{true}),\textit{ret}(\textit{put},d,\textit{true}), \textit{call}(\textit{rm},d),\textit{ret}(\textit{rm},d),\textit{call}(\textit{rm},\textit{empty}),\textit{ret}(\textit{rm},\textit{empty}) \}$.

\begin{figure}[htbp]
  \centering
  \includegraphics[width=0.9 \textwidth]{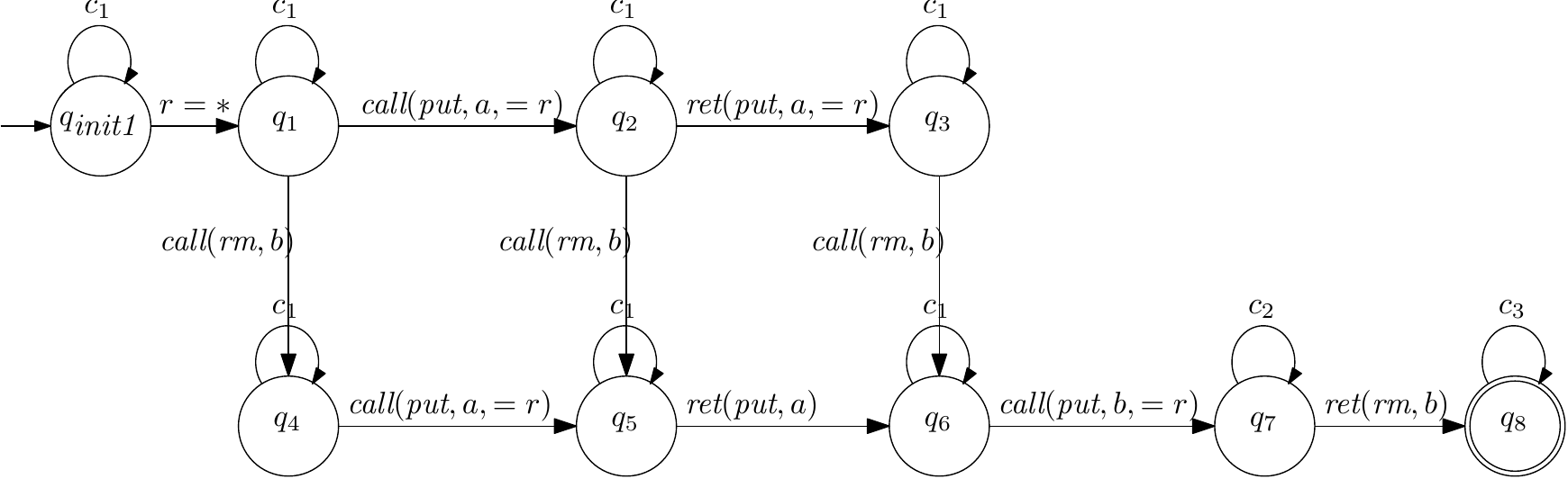}
  \caption{Automaton $\mathcal{A}_{\textit{SinPri}}^4$}
  \label{fig:automata for FIFO-4}
\end{figure}

Let $\mathcal{A}_{\textit{sinPri}}$ be the union of $\mathcal{A}_{\textit{SinPri}}^1, \mathcal{A}_{\textit{SinPri}}^2, \mathcal{A}_{\textit{SinPri}}^3, \mathcal{A}_{\textit{SinPri}}^4$. Let us prove Lemma \ref{lemma:automata for extended priority queue with single priority}.

\begin{lemma}\label{lemma:automata for extended priority queue with single priority}
Given a data-independent implementations $\mathcal{I}$ of extended priority queue, $\mathcal{I} \cap \mathcal{A}_{\textit{sinPri}} \neq \emptyset$, if and only if there exists $e \in \mathcal{I}_{\neq}$, $e' \in \textit{proj}(e)$, such that $e'$ is single-priority  without $\textit{rm}(\textit{empty})$, and $\textit{transToQueue}(e')$ does not linearizable to queue.
\end{lemma}

%


\begin {proof}

\cite{DBLP:conf/icalp/BouajjaniEEH15} states that, given a differentiated queue execution $e$ without $\textit{deq}(\textit{empty})$, $e$ is not linearizable with respect to queue, if one of the following cases holds for some $v_a,v_b$: (1) $\textit{deq}(v_b) <_{hb} \textit{enq}(v_b)$, (2) there are are no $\textit{enq}(v_b)$ and at least one $\textit{deq}(v_b)$, (3) there are are one $\textit{enq}(v_b)$ and more than one $\textit{deq}(v_b)$, and (4) $\textit{enq}(v_a) <_{\textit{hb}} \textit{enq}(v_b)$, and $\textit{deq}(v_b) <_{\textit{hb}} \textit{deq}(v_a)$, or $\textit{deq}(v_a)$ does not exists.

Let us prove the $\textit{only if}$ direction. Assume that there exists execution $e_0 \in \mathcal{I}$ and $e_0$ is accepted by $\mathcal{A}_{\textit{SinPri}}^1, \mathcal{A}_{\textit{SinPri}}^2, \mathcal{A}_{\textit{SinPri}}^3$ or $\mathcal{A}_{\textit{SinPri}}^4$. By data-independence, we can see that there exists a data-differentiated $e \in \mathcal{I}$ and renaming function, such that $e_0=r(e)$. Let $e'$ be obtained from $e$ by first removing $\textit{rm}(\textit{empty})$, and then,

\begin{itemize}
\setlength{\itemsep}{0.5pt}
\item[-] If $e_0$ is accepted by $\mathcal{A}_{\textit{SinPri}}^1$, $\mathcal{A}_{\textit{SinPri}}^2$ or $\mathcal{A}_{\textit{SinPri}}^3$: Then remove all values that are not renamed into $b$ by $r$.

\item[-] If $e_0$ is accepted by $\mathcal{A}_{\textit{SinPri}}^4$: Then remove all values that are not renamed into $a$ or $b$ by $r$.
\end{itemize}

It is obvious that $e' \in \textit{proj}(e)$. It is easy to see that $\textit{transToQueue}(e')$ satisfies one of above conditions, and then $\textit{transToQueue}(e')$ is not linearizable w.r.t queue.

Let us prove the $\textit{if}$ direction. Assume that exists $e \in \mathcal{I}_{\neq}$, $e' \in \textit{proj}(e)$, such that $e'$ is single-priority  without $\textit{rm}(\textit{empty})$, and $\textit{transToQueue}(e')$ does not linearizable to queue. Then we construct a renaming function $r$ as follows:

\begin{itemize}
\setlength{\itemsep}{0.5pt}
\item[-] If this is because case $1$, case $2$ or case $3$: $r$ maps $v_b$ into $b$ and maps all other values into $a$.

\item[-] If this is because case $4$: $r$ maps $v_a$ and $v_b$ into $a$ and $b$, respectively, and maps all other values into $d$.
\end{itemize}

Then it is not hard to see that $r(e) \in \mathcal{I}$ and it is accepted by $\mathcal{A}_{\textit{SinPri}}^1, \mathcal{A}_{\textit{SinPri}}^2, \mathcal{A}_{\textit{SinPri}}^3$ or $\mathcal{A}_{\textit{SinPri}}^4$. This completes the proof of this lemma. \qed
\end {proof}

\subsection{Proof of Lemma \ref{lemma:pri execution is enough}}
\label{sec:appendix proof of Lemma pri execution is enough}

The following lemma states that, from linearization of sub-histories, we can merge them and obtain a linearization (regardless of whether it belongs to sequential specification) of the whole history.

\begin{restatable}{lemma}{MergeTwoLinearization}
\label{lemma:merge two linearization}

Given an execution $e$, operation sets $S_1$, $S_2$ and sequences $l_1$ and $l_2$. Let $e_1 = e \vert{S_1}$ and $e_2 = e \vert{S_2}$. Assume that $e_1 \sqsubseteq l_1$, $e_2 \sqsubseteq l_2$, and $S_1 \cup S_2$ contains all operations of $e$. Then, there exists a sequence $l$, such that $e \sqsubseteq l$, $l \vert{S_1} = l_1$ and $l \vert{S_2} = l_2$.
\end{restatable}

\begin {proof}

Given an execution $e$ and a operation $o \in S_2$, let $\textit{MB}(o) = \{ o' \vert o' \in S_1$ and $o'$ happens before $o$ in $h \}$, let $\textit{SBI}(o) = \textit{min}\{ i \vert l_1[0,i]$ contains all elements of $\textit{MB}(o) \}$.

Let $l = s_1 \cdot l_2[1] \cdot \ldots \cdot l_2[n] \cdot s_{\textit{n+1}}$ be generated as follows, where $n = \vert l_2 \vert$:

\begin{itemize}
\setlength{\itemsep}{0.5pt}
\item[-] $s_1 = l_1[0, \textit{SBI}(l_2[1])]$,

\item[-] If $s_1 \cdot l_2[1] \cdot \ldots \cdot l_2[i]$ already contains $l_1(\textit{SBI}(l_2[\textit{i+1}]))$, then $s_{\textit{i+1}} = \epsilon$. Otherwise, $s_{\textit{i+1}}$ is a subsequence of $l_1$, which starts from the next of last elements of $s_1 \cdot l_2[1] \cdot \ldots \cdot l_2[i]$ in $l_1$ and ends in $l_1(\textit{SBI}(l_2[\textit{i+1}]))$.
\end{itemize}

It is obvious that $l \vert{S_1} = l_1$ and $l \vert{S_2} = l_2$, and it remains to prove that $e \sqsubseteq l$. We prove this by contradiction. Assume that $e \not \sqsubseteq l$. Then there must be two operations $o_1$, $o_2$ of $e$, such that $o_1 <_{hb} o_2$ in $e$ but $o_2$ before $0_1$ in $l$. Since $l \vert{S_1} = l_1$, $l \vert{S_2} = l_2$, and $h_1 \sqsubseteq l_1$, $h_2 \sqsubseteq l_2$, it is easy to see that it is impossible that $o_1,o_2 \in S_1$ or $o_1,o_2 \in S_2$. There are only two possibilities:

\begin{itemize}
\setlength{\itemsep}{0.5pt}
\item[-] $o_1 \in S_1 \wedge o_2 \in S_2$. Then we can see that $o_2=l_2[i]$ and $o_1 \in s_j$ for some $i < j$. Since $o_1 <_{hb} o_2$, we know that $o_1 \in \textit{SBI}(o_2)$. By the construction of $l$, we know that $o_1$ must be in $s_k$ for some $k \leq i$, contradicts that $o_1 \in s_j$ with $i < j$.

\item[-] $o_1 \in S_2 \wedge o_2 \in S_1$. Then we can see that $o_2 \in s_i$ and $o_1 = l_2[j]$ for some $i \leq j$. It is easy to see that this leads to contradiction when $i = j$. For the case of $i \neq j$, we need to satisfy the following requirements: (1) $o_1$ ($l_2[j]$) does not happen before $l_2[i]$, (2) $l_2[i]$ does not happen before $o_2$, (3) $o_2$ is either overlap or happens before $o' \in \textit{MB}(l_2[i])$, and (4) $o' <_{hb} l_2[i]$. By enumeration we can see that it is impossible that above four conditions be satisfied while $o_1 <_{hb} o_2$.
\end{itemize}

This completes the proof of this lemma. \qed
\end {proof}

With Lemma \ref{lemma:merge two linearization}, we can now prove Lemma \ref{lemma:pri execution is enough}.

$\newline$

{\noindent \bf Lemma \ref{lemma:pri execution is enough}}: Let $\Gamma\in \{\mathsf{UnmatchedMaxPriority}, \mathsf{MatchedMaxPriority}\}$ and $e$ a data-differentiated execution. Then, $e$ is $\Gamma$-linearizable iff $e\vert_{\preceq p}$ is $\Gamma$-linearizable for some maximal priority $p$ in $e$.

\begin {proof}

We deal with the case of $\Gamma = \mathsf{MatchedMaxPriority}^{>}$ , and other cases can be similarly dealt with.

To prove the $\textit{only if}$ direction, assume that $e \sqsubseteq l$, $\mathsf{MatchedMaxPriority\text{-}Seq}(l,x)$ holds and $p$ is the priority of $x$. We can see that $e \sqsubseteq u \cdot \textit{put}(x,p) \cdot v \cdot \textit{rm}(x) \cdot w$. Let $u'$, $v'$ and $w'$ be obtained from $u$, $v$ and $w$ by erasing all values with priority incomparable with $p$, respectively. Since the predicates in $\mathsf{MatchedMaxPriority\text{-}Seq}(l,x)$ does not restrict the values of priorities incomparable with $p$, is not hard to see that $e \sqsubseteq l'$ and $\mathsf{MatchedMaxPriority\text{-}Seq}(l',x)$ holds, where $l' = u' \cdot \textit{put}(x,\textit{pri}) \cdot v' \cdot \textit{rm}(x) \cdot w'$.

To prove the $\textit{if}$ direction, given $e' = e \vert_{\prec p}$. By assumption, $e' \sqsubseteq l$, $\mathsf{MatchedMaxPriority\text{-}Seq}(l,x)$ holds and $p$ is the priority of $x$, and then we can see that $e' \sqsubseteq l_1 = u \cdot \textit{put}(x,p) \cdot v \cdot \textit{rm}(x) \cdot w$. Let $O_c$ be the set of operations in $e$ that have priorities comparable with $p$, and Let $O_i$ be the set of operations in $e$ that have priorities incomparable with $p$. It is obvious that $l_1$ is the linearization of $e \vert_{O_c}$. By Lemma \ref{lemma:merge two linearization}, there exists sequence $l$, such that $e \sqsubseteq l$, and $l \vert_{O_c} = l_1$. Then $l = u' \cdot \textit{put}(x,p) \cdot v' \cdot \textit{rm}(x) \cdot w'$, where $u' \vert_{O_c} = u$, $v' \vert_{O_c} = v$ and $w' \vert_{O_c} = w$. Since $p$ is one of maximal priorities in $e$, and the predicates in $\mathsf{MatchedMaxPriority\text{-}Seq}(l,x)$ does not restrict the values with priority $O_i$, it is easy to see that $\mathsf{MatchedMaxPriority\text{-}Seq}(l,x)$ holds. \qed
\end {proof}

Therefore, from now on, it is safe to consider only the data-differentiated sequences with only one maximal priority.

\subsection{Proofs, Definitions and Register Automata in Subsection \ref{subsec:co-regular of EPQ1Lar}}
\label{sec:appendix proof and definition in section co-regular of EPQ1Lar}

Let us introduce $\textit{UVSet}(e,x)$, which intuitively contains all pairs of operations that should be putted before $\textit{rm}(x)$ when construction linearization of $x$ according to $\mathsf{MatchedMaxPriority}$. Let $\textit{UVSet}_1(e$, $x) = o \vert$ either $o <{\textit{hb}} \textit{put}(x,\_)$ or $\textit{rm}(x)$, or there exists $o'$ with the same value of $o$, such that $o' <{\textit{hb}} \textit{put}(x)$ or $\textit{rm}(x) \}$. For each $i > 1$, let $\textit{UVSet}_{\textit{i+1}}(e,x) = o \vert$ $o \notin \textit{UVSet}_k$ for each $k \leq i$, and either $o$ happens before some operation $o' \in \textit{UVSet}_i(e,x)$, or there exists $o''$ with the same value of $o$, and $o''$ happens before some operation $o' \in \textit{UVSet}_i(e,x) \}$. Let $\textit{UVSet}(e,x) = \textit{UVSet}_1(e,x) \cup \ldots$. Note that it is possible that $\textit{UVSet}_i(e,x) \cap \textit{UVSet}_j(e,x) = \emptyset$ for any $i \neq j$.

The following lemma states that $\textit{UVSet}(e,x)$ contains only matched $\textit{put}$ and $\textit{rm}$.

\begin{restatable}{lemma}{UVSetHasMatchedPutandRm}
\label{lemma:UVSet has matched put and rm}
Given a data-differentiated execution $e$ where $\mathsf{Has\text{-}MatchedMaxPriority}(e)$ holds, let $p$ be its maximal priority and $\textit{put}(x,p),\textit{rm}(x)$ are only operations of priority $p$ in $e$. Let $G$ be the graph representing the left-right constraint of $x$. Assume that $G$ has no cycle going through $x$. Then, $\textit{UVSet}(e,x)$ contains only matched $\textit{put}$ and $\textit{rm}$.
\end{restatable}
\begin {proof}

We prove this lemma by contradiction. Assume that there exists a value, such that $\textit{UVSet}(e,x)$ contains only its $\textit{put}$ and does not contain its $\textit{rm}$. Then we can see that there exists $d_1,\ldots,d_j$. Intuitively, $d_1,\ldots,d_j$ are elements in $\textit{UVSet}_1(e,x), \ldots, \textit{UVSet}_i(e,x)$, respectively. $\textit{UVSet}(e,x)$ contains $\textit{put}(d_j,\_)$ and does not contain $\textit{rm}(d_j)$. And each $d_i$ is the reason of $d_{\textit{i+1}} \in \textit{UVSet}_{\textit{i+1}}(e,x)$. Formally, we require that

\begin{itemize}
\setlength{\itemsep}{0.5pt}
\item[-] For each $1 \leq i \leq j$, operations of $d_i$ belongs to $\textit{UVSet}_i(e,x)$.

\item[-] For each $i \neq j$, $\textit{put}(d_i,\_),\textit{rm}(d_i) \in \textit{UVSet}_i(e,x)$. $\textit{put}(d_j,\_) \in \textit{UVSet}_j(e,x)$, and $e$ does not contain $\textit{rm}(d_j)$.

\item[-] An operation of $d_1$ happens before an operations of $x$. For each $1 < i \leq j$, an operation of $d_i$ happens an operation of $d_{\textit{i-1}}$.

\item[-] For each $k$ and $\textit{ind}$, if $k > \textit{ind+1}$, then no operation of $d_k$ happens before operation of $d_{\textit{ind}}$.
\end{itemize}

According to the definition of $\textit{UVSet}(e,x)$, it is easy to see that such $d_1,\ldots,d_j$ exists. Let us prove the following fact:

\noindent {\bf $\textit{fact}_1$}: Given $1 \leq i < j$, it can not be the case that $\textit{put}(d_i,\_)$ and $\textit{rm}(d_i)$ overlap.

Proof of $\textit{fact}_1$: We prove $\textit{fact}_1$ by contradiction. Assume that for some $i \neq j$, $\textit{put}(d_i,\_)$ and $\textit{rm}(d_i)$ overlap. Since $\textit{put}(d_i,\_), \textit{rm}(d_i) \in \textit{UVSet}_i(h,x)$, we know that an operation $o_i$ of $d_i$ happens before operation $o_{\textit{i-1}}$ of $d_{\textit{i-1}}$. Moreover, since $\textit{put}(d_i,\_)$ and $\textit{rm}(d_i)$ overlap, it is not hard to see that the call action of $\textit{put}(d_i,\_)$ and the call action of $\textit{rm}(d_i)$ is before the call action of $o_{\textit{i-1}}$. Since operations of $d_{\textit{i+1}}$ is in $\textit{UVSet}_{\textit{i+1}}(e,x)$, we know that an operation $o'_{\textit{i+1}}$ of $d_{\textit{i+1}}$ happens before operation $o'_i$ of $d_i$. Then, it is not hard to see that $o'_{\textit{i+1}}$ also happens before $o_{\textit{i-1}}$, which contradicts that for each $k > \textit{ind+1}$, no operation of $d_k$ happens before operation of $d_{\textit{ind}}$.

We already know that an operation of $d_1$ happens before an operation of $x$. By $\textit{fact}_1$, we can ensure that $\textit{put}(d_1,\_)$ happens before an operation of $x$, and then $d_1 \rightarrow x$ in $G$. For each $1 < i \leq j$, we know that an operation $o_i$ of $d_i$ happens before an operation $o_{\textit{i-1}}$ of $d_{\textit{i-1}}$. By $\textit{fact}_1$, we can ensure that $o_i=\textit{put}(d_i,\_)$ and $o_{\textit{i-1}}=\textit{rm}(d_{\textit{i-1}})$, and then $d_i \rightarrow d_{\textit{i-1}}$ in $G$. Since $h$ contains $\textit{put}(d_j,\_)$ and does not contain $\textit{rm}(d_j)$, we know that $x \rightarrow d_j$ in $G$. Then $G$ has a cycle going through $x$, contradicts that $G$ has no cycle going through $x$. \qed
\end {proof}

The following lemma states that $\textit{UVSet}(e,x)$ does not happen before $\textit{rm}(x)$ when the left-right constraint has no cycle going through $x$.

\begin{restatable}{lemma}{RmxDoesNotHappenBeforeUVSetForEPQ1Lar}
\label{lemma:Rmx does not happen before UVSet for EPQ1Lar}

Given a data-differentiated execution $e$ where $\mathsf{Has\text{-}MatchedMaxPriority}(e)$ holds, let $p$ be its maximal priority and $\textit{put}(x,p),\textit{rm}(x)$ are only operations of priority $p$ in $e$. Let $G$ be the graph representing the left-right constraint of x. Assume that $G$ has no cycle going through $x$. Then, $\textit{rm}(x)$ does not happen before any operation in $\textit{UVSet}(e,x)$.
\end{restatable}

\begin {proof}

We prove this lemma by induction, and prove that $\textit{rm}(x)$ does not happen before any operation in $\textit{UVSet}_1(e,x)$, in $\textit{UVSet}_2(e,x)$, $\ldots$. Note that, by Lemma \ref{lemma:UVSet has matched put and rm}, $\textit{UVSet}(e,x)$ contains only matched $\textit{put}$ and $\textit{rm}$, and it is easy to see that for each $i$, $\textit{UVSet}_i(e,x)$ contains only matched $\textit{put}$ and $\textit{rm}$.

\noindent (1) Let us prove that $\textit{rm}(x)$ does not happen before any operation in $\textit{UVSet}_1(e,x)$ by contradiction. Assume that $\textit{rm}(x) <_{hb} o$, where $o \in \textit{UVSet}_1(e,x)$ is an operation of value $d$.

We use a triple $(t_1,t_2,t_3)$ to represent related information. $t_1,t_2,t_3$ are chosen from $\{ \textit{put},\textit{rm} \}$. $t_1$ represents whether $o$ is a $\textit{put}$ operation or a $\textit{rm}$ operation. $t_2$ and $t_3$ is used for the reason of $o \in \textit{UVSet}_1(e,x)$: $o \in \textit{UVSet}_1(e,x)$, since an operation (of kind $t_2$) of $d$ happens before an operation (of kind $t_3$) of $x$. Let us consider all the possible cases of $(t_1,t_2,t_3)$:

\begin{itemize}
\setlength{\itemsep}{0.5pt}
\item[-] $(\textit{put},\textit{put},\textit{put})$: Then $\textit{rm}(x) <_{hb} \textit{put}(d,\_) <_{hb} \textit{put}(x,p)$, contradicts that $\textit{rm}(x)$ does not happen before $\textit{put}(x,p)$.

\item[-] $(\textit{put},\textit{put},\textit{rm})$: Then $\textit{rm}(x) <_{hb} \textit{put}(d,\_) <_{hb} \textit{rm}(x)$, contradicts that $\textit{rm}(x)$ does not happen before $\textit{rm}(x)$.

\item[-] $(\textit{put},\textit{rm},\textit{put})$: Then $( \textit{rm}(x) <_{hb} \textit{put}(d,\_) ) \wedge ( \textit{rm}(d) <_{hb} \textit{put}(x,p) )$. By interval order, we know that $( \textit{rm}(x) <_{hb} \textit{put}(x,p) ) \vee ( \textit{rm}(d) <_{hb} \textit{put}(d,\_) )$, which is impossible.

\item[-] $(\textit{put},\textit{rm},\textit{rm})$: Then $( \textit{rm}(x) <_{hb} \textit{put}(d,\_) ) \wedge ( \textit{rm}(d) <_{hb} \textit{rm}(x) )$. We can see that $\textit{rm}(d) <_{hb} \textit{rm}(x) <_{hb} \textit{put}(d,\_)$, which contradicts that $\textit{rm}(d)$ does not happen before $\textit{put}(d,\_)$.

\item[-] $(\textit{rm},\textit{put},\textit{put})$: Then $( \textit{rm}(x) <_{hb} \textit{rm}(d) ) \wedge ( \textit{put}(d,\_) <_{hb} \textit{put}(x,p) )$. We can see that $x$ and $d$ has circle in $G$, contradicts that $G$ has no cycle going through $x$.

\item[-] $(\textit{rm},\textit{put},\textit{rm})$: Then $( \textit{rm}(x) <_{hb} \textit{rm}(d) ) \wedge ( \textit{put}(d,\_) <_{hb} \textit{rm}(x) )$. We can see that $x$ and $d$ has circle in $G$, contradicts that $G$ has no cycle going through $x$.

\item[-] $(\textit{rm},\textit{rm},\textit{put})$: Then $\textit{rm}(x) <_{hb} \textit{rm}(d) <_{hb} \textit{put}(x,p)$, contradicts that $\textit{rm}(x)$ does not happen before $\textit{put}(x,p)$.

\item[-] $(\textit{rm},\textit{rm},\textit{rm})$: Then $\textit{rm}(x) <_{hb} \textit{rm}(d) <_{hb} \textit{rm}(x)$, contradicts that $\textit{rm}(x)$ does not happen before $\textit{rm}(x)$.
\end{itemize}

This completes the proof for $\textit{UVSet}_1(e,x)$.

\noindent (2) Assume we already prove that for some $j \geq 1$, $\textit{rm}(x)$ does not happen before any operation in $\textit{UVSet}_1(e,x) \cup \ldots \cup \textit{UVSet}_j(e,x)$. Let us prove that $\textit{rm}(x)$ does not happen before any operation in $\textit{UVSet}_{\textit{j+1}}(e,x)$ by contradiction. Assume that $\textit{rm}(x) <_{hb} o$, where $o \in \textit{UVSet}_{\textit{j+1}}(e,x)$ is an operation of value $d_{\textit{j+1}}$. We use a triple $(t_1,t_2,t_3)$ to represent related information. $t_1,t_2,t_3$ are chosen from $\{ \textit{put},\textit{rm} \}$. $t_1$ represents whether $o$ is a $\textit{put}$ operation or a $\textit{rm}$ operation. $t_2$ and $t_3$ is used for the reason of $o \in \textit{UVSet}_{\textit{j+1}}(e,x)$: $o \in \textit{UVSet}_{\textit{j+1}}(e,x)$, since an operation (of kind $t_2$) of $d_{\textit{j+1}}$ happens before an operation (of kind $t_3$) of $d_j$, where $\textit{put}(d_j,\_), \textit{rm}(d_j) \in \textit{UVSet}_j(e,x)$. Let us consider all the possible cases of $(t_1,t_2,t_3)$:

\begin{itemize}
\setlength{\itemsep}{0.5pt}
\item[-] $(\textit{put},\textit{put},\textit{put})$: Then $\textit{rm}(x) <_{hb} \textit{put}(d_{\textit{j+1}},\_) <_{hb} \textit{put}(d_j,\_)$. We can see that $( \textit{rm}(x) <_{hb} \textit{put}(d_j,\_) ) \wedge ( \textit{put}(d_j,\_) \in \textit{UVSet}_j(e,x) )$, which contradicts that $\textit{rm}(x)$ does not happen before any operation in $\textit{UVSet}_1(e,x) \cup \ldots \cup \textit{UVSet}_j(e,x)$.

\item[-] $(\textit{put},\textit{put},\textit{rm})$: Then $\textit{rm}(x) <_{hb} \textit{put}(d_{\textit{j+1}},\_) <_{hb} \textit{rm}(d_j,\_)$. We can see that $( \textit{rm}(x) <_{hb} \textit{rm}(d_j,\_) ) \wedge ( \textit{rm}(d_j) \in \textit{UVSet}_j(e,x) )$, which contradicts that $\textit{rm}(x)$ does not happen before any operation in $\textit{UVSet}_1(e,x) \cup \ldots \cup \textit{UVSet}_j(e,x)$.

\item[-] $(\textit{put},\textit{rm},\textit{put})$: Then $( \textit{rm}(x) <_{hb} \textit{put}(d_{\textit{j+1}},\_) ) \wedge ( \textit{rm}(d_{\textit{j+1}}) <_{hb} \textit{put}(d_j,\_) )$. By interval order, we know that $( \textit{rm}(x) <_{hb} \textit{put}(d_j,\_) ) \vee ( \textit{rm}(d_{\textit{j+1}}) <_{hb} \textit{put}(d_{\textit{j+1}},\_) )$, which is impossible.

\item[-] $(\textit{put},\textit{rm},\textit{rm})$: Then $( \textit{rm}(x) <_{hb} \textit{put}(d_{\textit{j+1}},\_) ) \wedge ( \textit{rm}(d_{\textit{j+1}}) <_{hb} \textit{rm}(d_j) )$. By interval order, we know that $( \textit{rm}(x) <_{hb} \textit{rm}(d_j) ) \vee ( \textit{rm}(d_{\textit{j+1}}) <_{hb} \textit{put}(d_{\textit{j+1}},\_) )$, which is impossible.

\item[-] $(\textit{rm},\textit{put},\textit{put})$: Then $( \textit{rm}(x) <_{hb} \textit{rm}(d_{\textit{j+1}}) ) \wedge ( \textit{put}(d_{\textit{j+1}},\_) <_{hb} \textit{put}(d_j,\_) )$. Let us consider the reason of $\textit{put}(d_j,\_), \textit{rm}(d_j) \in \textit{UVSet}_j(e,x)$:
    \begin{itemize}
    \setlength{\itemsep}{0.5pt}
    \item[-] If $( j > 1 ) \wedge ( \textit{put}(d_j,\_) <_{hb} o'' )$, where $o''$ is an operation of value $d_{\textit{j-1}}$ and $\textit{put}(d_{\textit{j-1}},\_), \textit{rm}(d_{\textit{j-1}}) \in \textit{UVSet}_{\textit{j-1}}(e,x)$: Then since $( \textit{put}(d_{\textit{j+1}},\_) <_{hb} \textit{put}(d_j,\_) ) \wedge ( \textit{put}(d_j,\_) <_{hb} o'' )$, we can see that $\textit{put}(d_{\textit{j+1}},\_) <_{hb} o''$, and then operations of $d_{\textit{j+1}}$ is in $\textit{UVSet}_j(e,x)$, contradicts that operations of $d_{\textit{j+1}}$ is in $\textit{UVSet}_{\textit{j+1}}(e,x)$.

    \item[-] If $( j = 1 ) \wedge ( \textit{put}(d_j,\_) <_{hb} o'' )$, where $o''$ is an operation of $x$: Similar to above case.

    \item[-] If $( j > 1 ) \wedge ( \textit{rm}(d_j) <_{hb} o'' )$, where $o''$ is an operation of value $d_{\textit{j-1}}$ and $\textit{put}(d_{\textit{j-1}},\_), \textit{rm}(d_{\textit{j-1}}) \in \textit{UVSet}_{\textit{j-1}}(e,x)$: Then since $( \textit{put}(d_{\textit{j+1}},\_) <_{hb} \textit{put}(d_j,\_) ) \wedge ( \textit{rm}(d_j) <_{hb} o'' )$, we can see that $( \textit{put}(d_{\textit{j+1}},\_) <_{hb} o'' ) \vee ( \textit{rm}(d_j) <_{hb} \textit{put}(d_j,\_) )$, which is impossible.

    \item[-] If $( j > 1 ) \wedge ( \textit{rm}(d_j) <_{hb} o'' )$, where $o''$ is an operation of $x$: Similar to above case.
    \end{itemize}

\item[-] $(\textit{rm},\textit{put},\textit{rm})$: Let $T_{\textit{ind}}$ be the set of sentences $\{ \textit{rm}(x) <_{hb} \textit{rm}(d_{\textit{j+1}}), \textit{put}(d_{\textit{j+1}},\_) <_{hb} \textit{rm}(d_j),\ldots, \textit{put}(d_{\textit{ind+1}},\_) <_{hb} \textit{rm}(d_{\textit{ind}}) \}$. Here each $d_i$ is a value of some operation in $\textit{UVSet}_i(e,x)$. Let us prove that from $T_j$ we can obtain contradiction by induction:

    {\bf Base case $1$}: From $T_1$ we can obtain contradiction.

    Let us prove base case $1$:

    \begin{itemize}
    \setlength{\itemsep}{0.5pt}
    \item[-] If $\textit{put}(d_1,\_)$ happens $o$, and $o$ is an operation of $x$. Then there is a cycle $x \rightarrow d_{\textit{j+1}} \rightarrow \ldots \rightarrow d_1 \rightarrow x$ in $G$, contradicts that $G$ has no cycle going through $x$.

    \item[-] If $\textit{rm}(d_1)$ happens before $o$, and $o$ is an operation of $x$. Then since $\textit{put}(d_2,\_) <_{hb} \textit{rm}(d_1)$ and $\textit{rm}(d_1) <_{hb} o$, we can see that $\textit{put}(d_2,\_) <_{hb} o$, and then $\textit{put}(d_2,\_) \in \textit{UVSet}_1(e,x)$, contradicts that $\textit{put}(d_2,\_) \in \textit{UVSet}_2(e,x)$.
    \end{itemize}

    {\bf Base case $2$}: From $T_2$ we can obtain contradiction.

    Let us prove base case $2$: If $\textit{rm}(d_2) <_{hb} o$, and $o$ is an operation of $d_1$, then since $( \textit{put}(d_3,\_) <_{hb} \textit{rm}(d_2) ) \wedge ( \textit{rm}(d_2) <_{hb} o )$, we know that $\textit{put}(d_3,\_) <_{hb} o$. This implies that $\textit{put}(d_3,\_) \in \textit{UVSet}_2(e,x)$, contradicts that $\textit{rm}(d_3,\_) \in \textit{UVSet}_3(e,x)$. Therefore, it is only possible that $\textit{put}(d_2,\_)$ happens before an operation of $d_1$.

    \begin{itemize}
    \setlength{\itemsep}{0.5pt}
    \item[-] If $\textit{put}(d_2,\_) <_{hb} \textit{put}(d_1,\_)$ and $\textit{put}(d_1,\_)$ happens before operations of $x$, then we know that $\textit{put}(d_2,\_)$ happens before operation of $x$, which is impossible.

    \item[-] If $\textit{put}(d_2,\_) <_{hb} \textit{put}(d_1,\_)$ and $\textit{rm}(d_1)$ happens before operations of $x$, then by interval order, we know that $\textit{put}(d_2,\_)$ happens before operation of $x$, or $\textit{rm}(d_1) <_{hb} \textit{put}(d_1,\_)$, which is impossible.

    \item[-] If $\textit{put}(d_2,\_) <_{hb} \textit{rm}(d_1)$ and $\textit{put}(d_1,\_)$ happens before operations of $x$, then $x \rightarrow d_{\textit{j+1}} \rightarrow \ldots \rightarrow d_1 \rightarrow x$ in $G$, contradicts that $G$ has no cycle going through $x$.

    \item[-] If $\textit{put}(d_2,\_) <_{hb} \textit{rm}(d_1)$ and $\textit{rm}(d_1)$ happens before operations of $x$, then we know that $\textit{put}(d_2,\_)$ happens before operation of $x$, which is impossible.
    \end{itemize}

    {\bf induction step}: Given $\textit{ind} \geq 3$, if from $T_{\textit{ind-1}}$ we can obtain contradiction, then from $T_{\textit{ind}}$ we can also contain contradiction.

    Prove of the induction step: Similarly as base case $2$, we can prove that it is only possible that $\textit{put}(d_{\textit{ind}},\_)$ happens before operations of $d_{\textit{ind-1}}$.

    \begin{itemize}
    \setlength{\itemsep}{0.5pt}
    \item[-] If $\textit{put}(d_{\textit{ind}},\_) <_{hb} \textit{put}(d_{\textit{ind-1}},\_)$ and $\textit{put}(d_{\textit{ind-1}},\_)$ happens before operations of $d_{\textit{ind-2}}$, then we know that $\textit{put}(d_{\textit{ind}})$ happens before operation of $d_{\textit{ind-2}}$, which is impossible.

    \item[-] If $\textit{put}(d_{\textit{ind}},\_) <_{hb} \textit{put}(d_{\textit{ind-1}},\_)$ and $\textit{rm}(d_{\textit{ind-1}})$ happens before operations of $d_{\textit{ind-2}}$, then by interval order, we know that $\textit{put}(d_{\textit{ind}},\_)$ happens before operation of $d_{\textit{ind-2}}$, or $\textit{rm}(d_{\textit{ind-1}}) <_{hb} \textit{put}(d_{\textit{ind-1}},\_)$, which is impossible.

    \item[-] If $\textit{put}(d_{\textit{ind}},\_) <_{hb} \textit{rm}(d_{\textit{ind-1}})$, then we obtain $T_{\textit{ind-1}}$, which already contain contradiction.
    \end{itemize}

    By base case $1$, base case $2$ and the induction step, it is easy to see that for each $i$, $T_i$ contains contradiction. Therefore, $T_j$, the case of $(\textit{rm},\textit{put},\textit{rm})$, contains contradiction.

\item[-] $(\textit{rm},\textit{rm},\textit{put})$: Then $( \textit{rm}(x) <_{hb} \textit{rm}(d_{\textit{j+1}}) ) \wedge ( \textit{rm}(d_{\textit{j+1}}) <_{hb} \textit{put}(d_j,\_) )$. We can see that $( \textit{rm}(x) <_{hb} \textit{put}(d_j,\_) ) \wedge ( \textit{put}(d_j,\_) \in \textit{UVSet}_j(e,x) )$, which contradicts that $\textit{rm}(x)$ does not happen before any operation in $\textit{UVSet}_1(e,x) \cup \ldots \cup \textit{UVSet}_j(e,x)$.

\item[-] $(\textit{rm},\textit{rm},\textit{rm})$: Then $( \textit{rm}(x) <_{hb} \textit{rm}(d_{\textit{j+1}}) ) \wedge ( \textit{rm}(d_{\textit{j+1}}) <_{hb} \textit{rm}(d_j) )$. We can see that $( \textit{rm}(x) <_{hb} \textit{rm}(d_j) ) \wedge ( \textit{rm}(d_j) \in \textit{UVSet}_j(e,x) )$, which contradicts that $\textit{rm}(x)$ does not happen before any operation in $\textit{UVSet}_1(e,x) \cup \ldots \cup \textit{UVSet}_j(e,x)$.
\end{itemize}

This completes the proof for $\textit{UVSet}_{\textit{j+1}}(e,x)$. Therefore, $\textit{rm}(x)$ does not happen before any operation in $\textit{UVSet}(e,x) = \textit{UVSet}_1(e,x) \cup \textit{UVSet}_2(e,x) \cup \ldots$. \qed
\end {proof}

With Lemma \ref{lemma:UVSet has matched put and rm} and Lemma \ref{lemma:Rmx does not happen before UVSet for EPQ1Lar}, we can now prove Lemma \ref{lemma:Lin Equals Constraint for EPQ1Lar}, which give the equivalent characterization of non $\mathsf{MatchedMaxPriority}$-linearizable executions.


$\newline$
{\noindent \bf Lemma \ref{lemma:Lin Equals Constraint for EPQ1Lar}}: Given a data-differentiated execution $e$ where $\mathsf{Has\text{-}MatchedMaxPriority}(e)$ holds, let $p$ be its maximal priority and $\textit{put}(x,p),\textit{rm}(x)$ are only operations of priority $p$ in $e$. Let $G$ be the graph representing the left-right constraint of $x$. $e$ is $\mathsf{MatchedMaxPriority}$-linearizable, if and only if $G$ has no cycle going through $x$.

\begin {proof}

To prove the $\textit{only if}$ direction, assume that $e \sqsubseteq l$ and $\mathsf{MatchedMaxPriority\text{-}Seq}(l,x)$ holds. Then $l = u \cdot \textit{put}(x,p) \cdot v \cdot \textit{rm}(x) \cdot w$, and let $U$, $V$ and $W$ be the set of operations of $u$, $v$ and $w$, respectively. Assume by contradiction that, there is a cycle $d_1 \rightarrow d_2 \rightarrow \ldots \rightarrow d_m \rightarrow x \rightarrow d_1$ in $G$. It is obvious that the priority of each $d_i$ is smaller than $p$. Then our proof proceeds as follows:

According to the definition of left-right constraint, there are two possibilities. The first possibility is that, $\textit{rm}(x)$ happens before $\textit{rm}(d_1)$. It is obvious that $\textit{rm}(d_1) \in W$, and then since $U \cup V$ contains matched $\textit{put}$ and $\textit{rm}$, we can see that $\textit{put}(d_1),\textit{rm}(d_1) \in W$. Then,

\begin{itemize}
\setlength{\itemsep}{0.5pt}
\item[-] Since $d_1 \rightarrow d_2$, by definition of $G$, we know that $\textit{put}(d_1)$ happens before $\textit{rm}(d_2)$. Since $\textit{put}(d_1) \in W$ and $U \cup V$ contains matched $\textit{put}$ and $\textit{rm}$, we know that $\textit{put}(d_2),\textit{rm}(d_2) \in W$. Similarly, for each $1 \leq i \leq m$, we know that $\textit{put}(d_i),\textit{rm}(d_i) \in W$.

\item[-] Since $d_m \rightarrow x$,
    \begin{itemize}
    \setlength{\itemsep}{0.5pt}
    \item[-] if $\textit{put}(d_m)$ happens before $\textit{put}(x)$, then we can see that $\textit{put}(d_m) \in U$, which contradicts that $\textit{put}(d_m) \in W$.

    \item[-] if $\textit{put}(d_m)$ happens before $\textit{rm}(x)$, then we can see that $\textit{put}(d_m) \in U \cup V$, which contradicts that $\textit{put}(d_m) \in W$.
    \end{itemize}
\end{itemize}

The second possibility is that, $e$ contains one $\textit{put}(d_1,\_)$ and no $\textit{rm}(d_1)$. Note that for each $j > 1$, $e$ contains $\textit{put}(d_j,\_)$ and $\textit{rm}(d_j)$. Since $d_m \rightarrow x$, is is obvious that $\textit{put}(d_m) \in U \cup V$. Since $U \cup V$ contains matched $\textit{put}$ and $\textit{rm}$, we know that $\textit{put}(d_m),\textit{rm}(d_m) \in U \cup V$. Then, since $d_{\textit{m-1}} \rightarrow d_m$, by definition of $G$, we know that $\textit{put}(d_{\textit{m-1}})$ happens before $\textit{rm}(d_m)$. Since $\textit{rm}(d_m) \in U \cup V$ and $U \cup V$ contains matched $\textit{put}$ and $\textit{rm}$, we know that $\textit{put}(d_{\textit{m-1}}),\textit{rm}(d_{\textit{m-1}}) \in U \cup V$. Similarly, for each $1 < i \leq m$, we know that $\textit{put}(d_i),\textit{rm}(d_i) \in U \cup V$, and also $\textit{put}(d_1)\in U \cup V$. However, there is one $\textit{put}(d_1,\_)$ and no $\textit{rm}(d_1)$ in $e$, contradicts that $U \cup V$ contains matched $\textit{put}$ and $\textit{rm}$.

This completes the proof of the $\textit{only if}$ direction.

To prove the $\textit{if}$ direction, assume that $G$ has no cycle going through $x$. Let $E_u$ be the set of operations that happen before $\textit{put}(x)$ in $e$. It is easy to see that $E_u \subseteq \textit{UVSet}(e,x)$. Let $E_v = \textit{UVSet}(e,x) \setminus E_u$. Let $E_e$ be the set of operations of $e$, and let $E_w = E_e \setminus \textit{UVSet}(e,x)$.

By Lemma \ref{lemma:UVSet has matched put and rm}, we can see that $E_u \cup E_v$ contains matched $\textit{put}$ and $\textit{rm}$ operations. It remains to prove that for $E_u$, $\{ \textit{put}(x,p) \}$, $E_v$, $\{ \textit{rm}(x) \}$, $E_w$, no elements of the latter set happens before elements of the former set. We prove this by showing that all the following cases are impossible:

\begin{itemize}
\setlength{\itemsep}{0.5pt}
\item[-] Case $1$: Some operation $o_w \in E_w$ happens before $\textit{rm}(x)$. Then we know that $o_w \in \textit{UVSet}(e,x) = E_u \cup E_v$, which contradicts that $o_w \in E_w$.

\item[-] Case $2$: Some operation $o_w \in E_w$ happens before some operation $o_{\textit{uv}} \in E_u \cup E_v$. Then we know that $o_w \in \textit{UVSet}(e,x) = E_u \cup E_v$, which contradicts that $o_w \in E_w$.

\item[-] Case $3$: Some operation $o_w \in E_w$ happens before $\textit{put}(x)$. Then we know that $o_w \in \textit{UVSet}(e,x) = E_u \cup E_v$, which contradicts that $o_w \in E_w$.

\item[-] Case $4$: $\textit{rm}(x)$ happens before some $o_{\textit{uv}} \in \textit{UVSet}(e,x) = E_u \cup E_v$. By Lemma \ref{lemma:Rmx does not happen before UVSet for EPQ1Lar} we know that this is impossible.

\item[-] Case $5$: $\textit{rm}(x)$ happens before $\textit{put}(x)$. This contradicts that each single-priority projection satisfy the FIFO property.

\item[-] Case $6$: Some operation $o_v \in E_v$ happens before $\textit{put}(x)$. Then we know that $o_v \in E_u$, which contradicts that $o_v \in E_v$.

\item[-] Case $7$: Some operation $o_v \in E_v$ happens before some operation $o_u \in E_u$. Then we know that $o_v \in E_u$, which contradicts that $o_v \in E_v$.

\item[-] Case $8$: $\textit{put}(x)$ happens before some operation $o_u \in E_u$. This is impossible.
\end{itemize}

This completes the proof of the $\textit{if}$ direction.

\qed
\end {proof}

Let us begin to represent register automata that is used for capture the existence of a data-differentiated execution $e$, let $p$ be a maximal priority of $e$, and $e'$ be the projection of $e$ into values with priorities comparable with $p$. We also require that $\mathsf{Has\text{-}MatchedMaxPriority}(e')$ holds and there exists a cycle going through the value with maximal priority in $e'$. By data-independence, we can obtain $e_r$ from $e$ by renaming function, which maps such value to be $b$, maps values that cover it to be $a$, and maps other values into $d$. There are four possible enumeration of call and return actions of $\textit{put}(b,p)$ and $\textit{rm}(b)$. For each of them, we generate a register automaton.

For the case when $e_r \vert_{b} = \textit{call}(\textit{put},b,p) \cdot \textit{ret}(\textit{put},b,p) \cdot \textit{call}(\textit{rm},b) \cdot \textit{ret}(\textit{rm},b)$, we generate register automaton $\mathcal{A}_{\textit{l-lar}}^1$, as shown in \figurename~\ref{fig:automata APQ1Lar-1}. Here $C_1 = C \cup \{ \textit{ret}(\textit{rm},a) \}$, $C_2 = C \cup \{ \textit{call}(\textit{put},a,<r) \}$, $C_3 = C_2 \cup \{ \textit{ret}(\textit{rm},a) \}$, where $C = \{ \textit{call}(\textit{put},\top,\textit{true}),\textit{ret}(\textit{put},\top,\textit{true}), \textit{call}(\textit{rm},\top)$, $\textit{ret}(\textit{rm},\top),$ $\textit{call}(\textit{rm},\textit{empty}),\textit{ret}(\textit{rm},\textit{empty})$. The differentiated branch in $\mathcal{A}_{\textit{l-lar}}^1$ comes from the positions of the first $\textit{ret}(\textit{put},a,\_)$.

\begin{figure}[htbp]
  \centering
  \includegraphics[width=1 \textwidth]{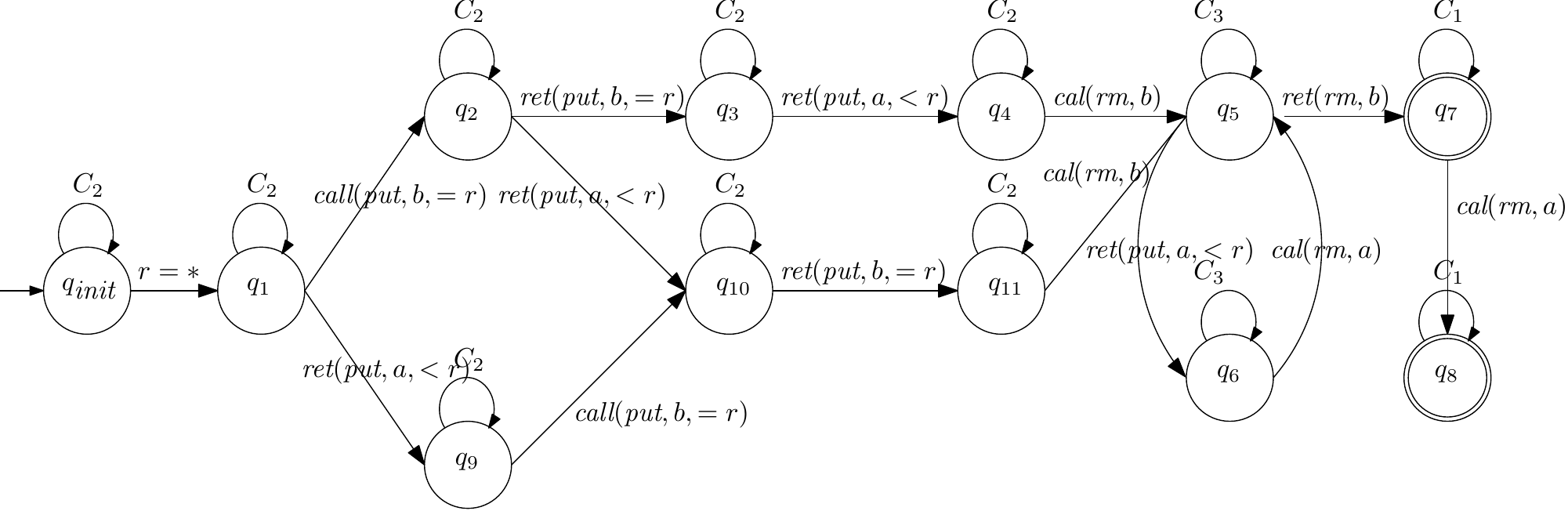}
  \caption{Automaton $\mathcal{A}_{\textit{l-lar}}^1$}
  \label{fig:automata APQ1Lar-1}
\end{figure}

$\mathcal{A}_{\textit{l-lar}}^1$ is used to recognize conditions in \figurename~\ref{fig:his for APQ1Lar-1}. Here for simplicity, we only draw operation of $b$, and the first $\textit{ret}(\textit{put},a,\_)$.

\begin{figure}[htbp]
  \centering
  \includegraphics[width=1 \textwidth]{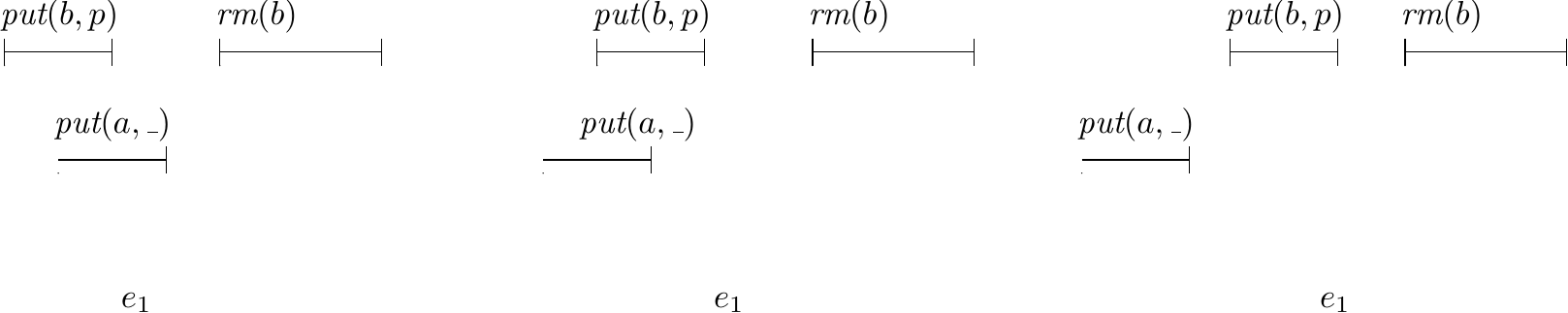}
  \caption{Conditions recognized by $\mathcal{A}_{\textit{l-lar}}^1$}
  \label{fig:his for APQ1Lar-1}
\end{figure}

For the case when $e_r \vert_{b} = \textit{call}(\textit{put},b,p) \cdot \textit{call}(\textit{rm},b) \cdot \textit{ret}(\textit{put},b,p) \cdot \textit{ret}(\textit{rm},b)$, we generate register automaton $\mathcal{A}_{\textit{l-lar}}^2$, as shown in \figurename~\ref{fig:automata APQ1Lar-2}. Here $C_1,C_2,C_3$ is the same as that in $\mathcal{A}_{\textit{l-lar}}^1$. The differentiated branch in $\mathcal{A}_{\textit{l-lar}}^2$ comes from the positions of the first $\textit{ret}(\textit{put},a,\_)$.

\begin{figure}[htbp]
  \centering
  \includegraphics[width=1 \textwidth]{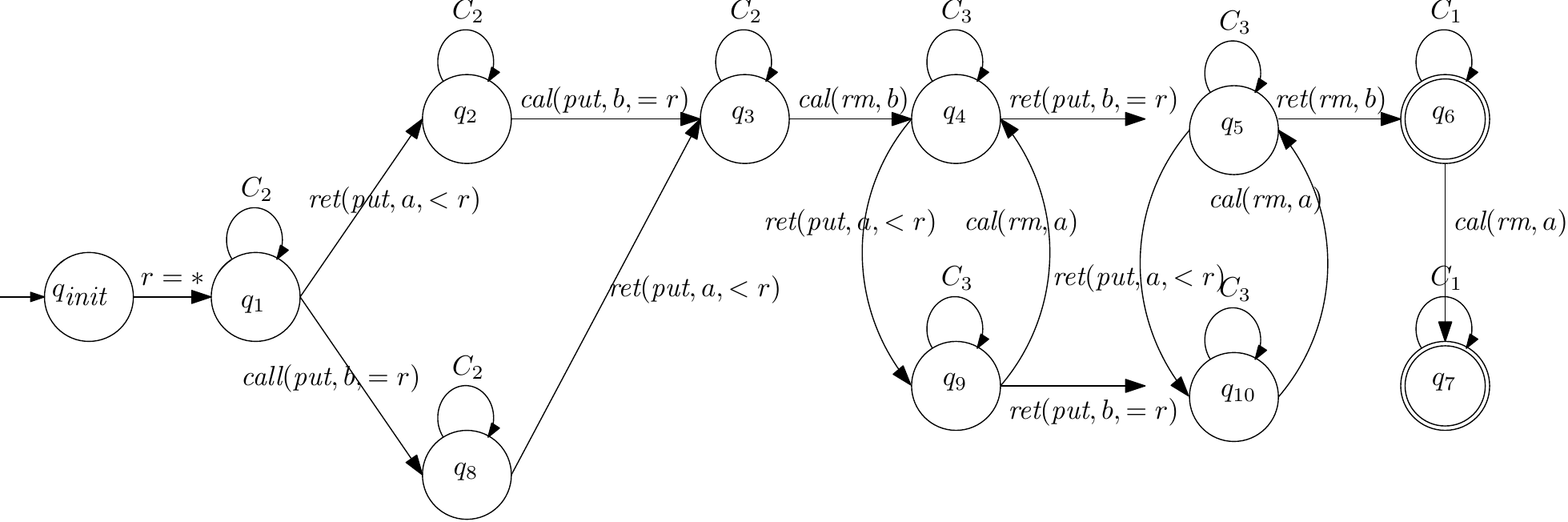}
  \caption{Automaton $\mathcal{A}_{\textit{l-lar}}^2$}
  \label{fig:automata APQ1Lar-2}
\end{figure}

$\mathcal{A}_{\textit{l-lar}}^2$ is used to recognize conditions in \figurename~\ref{fig:his for APQ1Lar-2}. Here for simplicity, we only draw operation of $b$, and the first $\textit{ret}(\textit{put},a,\_)$.

\begin{figure}[htbp]
  \centering
  \includegraphics[width=0.7 \textwidth]{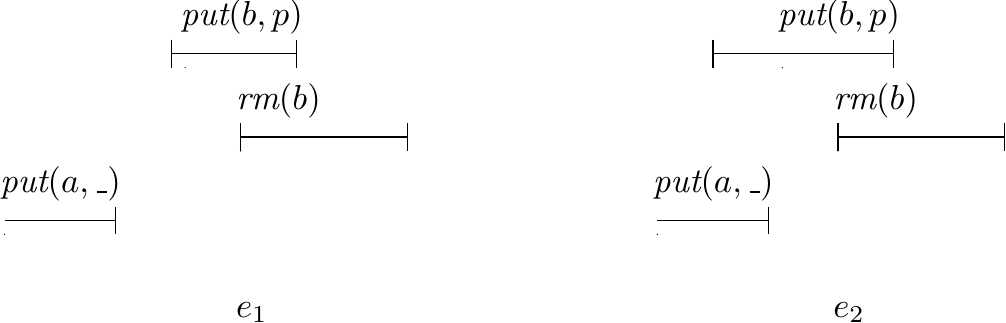}
  \caption{Conditions recognized by $\mathcal{A}_{\textit{l-lar}}^2$}
  \label{fig:his for APQ1Lar-2}
\end{figure}

For the case when $e_r \vert_{b} = \textit{call}(\textit{rm},b) \cdot \textit{call}(\textit{put},b,p) \cdot \textit{ret}(\textit{put},b,p) \cdot \textit{ret}(\textit{rm},b)$, we generate register automaton $\mathcal{A}_{\textit{l-lar}}^3$, as shown in \figurename~\ref{fig:automata APQ1Lar-3}. Here $C_1,C_2,C_3$ is the same as that in $\mathcal{A}_{\textit{l-lar}}^1$. The differentiated branch in $\mathcal{A}_{\textit{l-lar}}^3$ comes from the positions of the first $\textit{ret}(\textit{put},a,\_)$.

\begin{figure}[htbp]
  \centering
  \includegraphics[width=1 \textwidth]{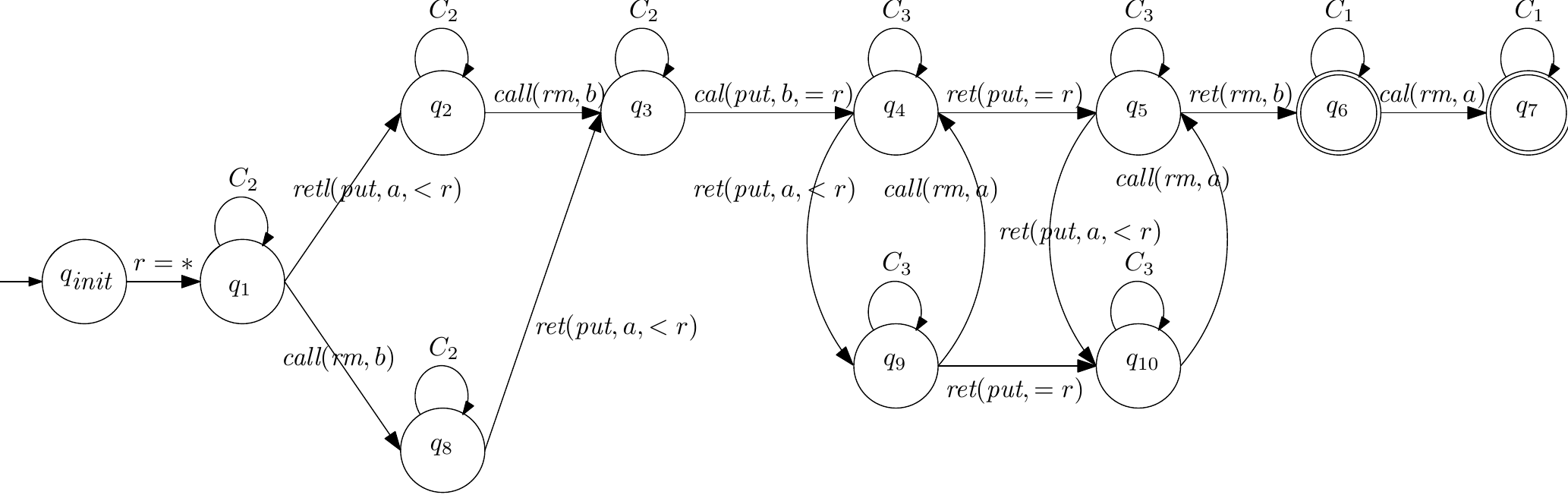}
  \caption{Automaton $\mathcal{A}_{\textit{l-lar}}^3$}
  \label{fig:automata APQ1Lar-3}
\end{figure}

$\mathcal{A}_{\textit{l-lar}}^3$ is used to recognize conditions in \figurename~\ref{fig:his for APQ1Lar-3}. Here for simplicity, we only draw operation of $b$, and the first $\textit{ret}(\textit{put},a,\_)$.

\begin{figure}[htbp]
  \centering
  \includegraphics[width=0.7 \textwidth]{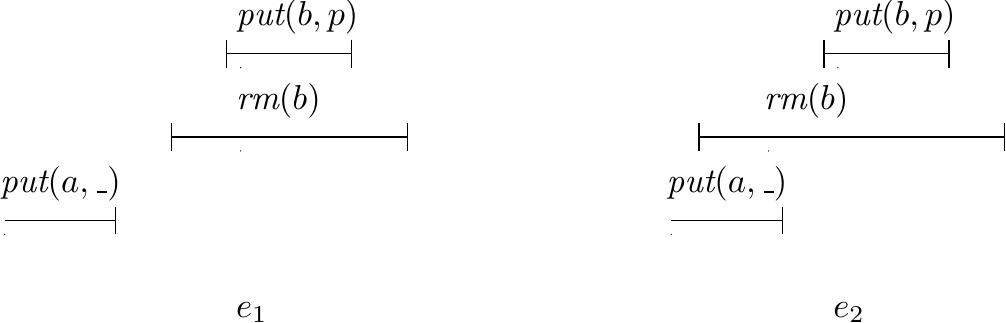}
  \caption{Conditions recognized by $\mathcal{A}_{\textit{l-lar}}^3$}
  \label{fig:his for APQ1Lar-3}
\end{figure}

For the case when $e_r \vert_{b} = \textit{call}(\textit{rm},b) \cdot \textit{call}(\textit{put},b,p) \cdot \textit{ret}(\textit{rm},b) \cdot \textit{ret}(\textit{put},b,p)$, we generate register automaton $\mathcal{A}_{\textit{l-lar}}^4$, as shown in \figurename~\ref{fig:automata APQ1Lar-4}. Here $C_1,C_2,C_3$ is the same as that in $\mathcal{A}_{\textit{l-lar}}^1$, and $C_4 = C_1 \cup \{ \textit{ret}(\textit{put},b,=r) \}$. The differentiated branch in $\mathcal{A}_{\textit{l-lar}}^4$ comes from the positions of the first $\textit{ret}(\textit{put},a,\_)$.

\begin{figure}[htbp]
  \centering
  \includegraphics[width=0.9 \textwidth]{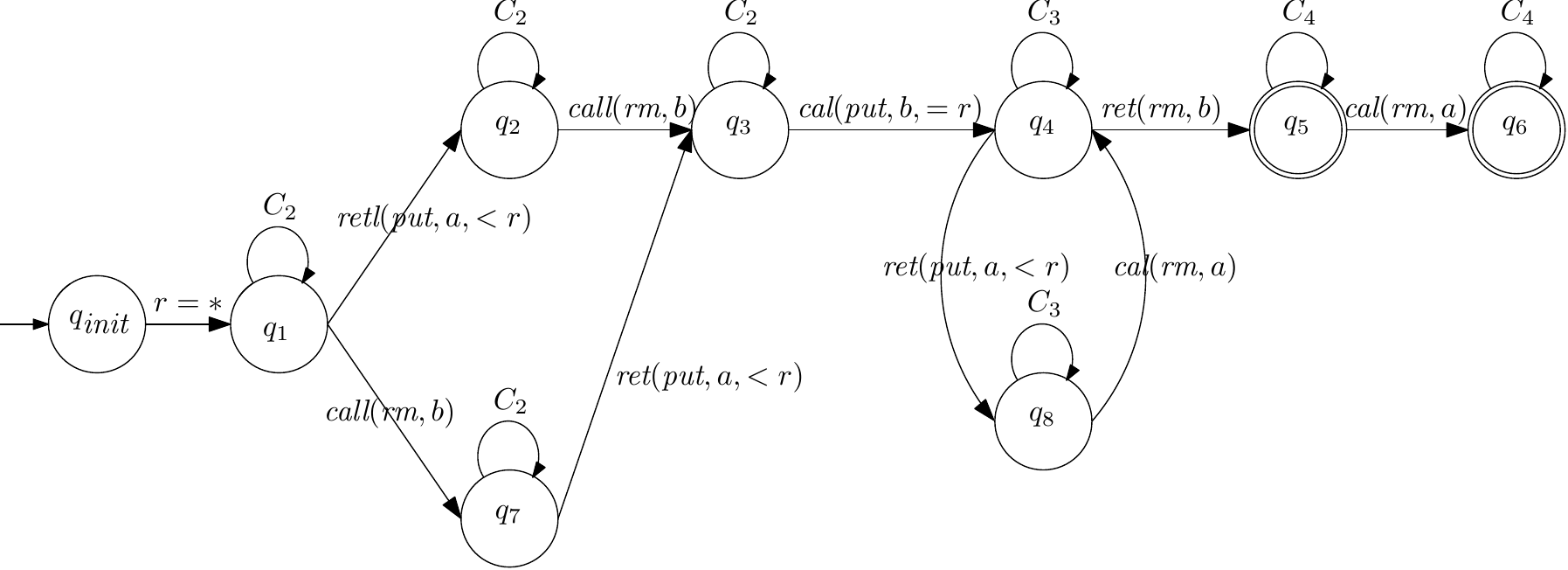}
  \caption{Automaton $\mathcal{A}_{\textit{l-lar}}^4$}
  \label{fig:automata APQ1Lar-4}
\end{figure}

$\mathcal{A}_{\textit{l-lar}}^4$ is used to recognize conditions in \figurename~\ref{fig:his for APQ1Lar-4}. Here for simplicity, we only draw operation of $b$, and the first $\textit{ret}(\textit{put},a,\_)$.

\begin{figure}[htbp]
  \centering
  \includegraphics[width=0.7 \textwidth]{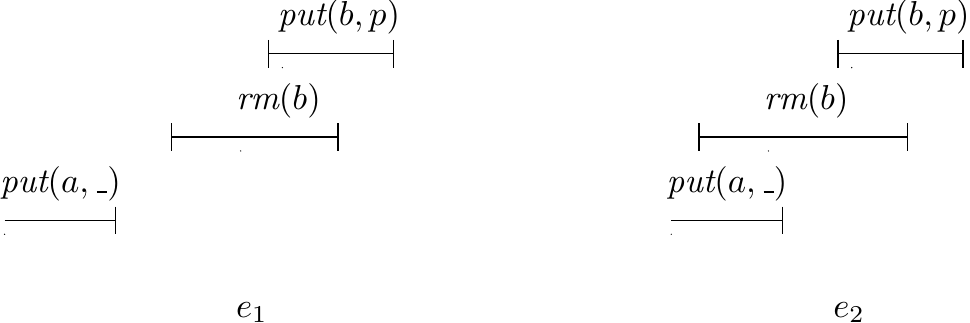}
  \caption{Conditions recognized by $\mathcal{A}_{\textit{l-lar}}^4$}
  \label{fig:his for APQ1Lar-4}
\end{figure}

Let $\mathcal{A}_{\textit{1-lar}}$ be the union of $\mathcal{A}_{\textit{l-lar}}^1, \mathcal{A}_{\textit{l-lar}}^2, \mathcal{A}_{\textit{l-lar}}^3$ and $\mathcal{A}_{\textit{l-lar}}^4$. The following lemma states that $\mathcal{A}_{\textit{1-lar}}$ is $\mathsf{MatchedMaxPriority}^{>}$-complete.

\begin{restatable}{lemma}{EPQOneLarisCoRegular}
\label{lemma:EPQ1Lar is co-regular}
$\mathcal{A}_{\textit{1-lar}}$ is $\mathsf{MatchedMaxPriority}^{>}$-complete.
\end{restatable}

\begin {proof}

We need to prove that, given a data-independent implementation $\mathcal{I}$. $\mathcal{A}_{\textit{1-lar}} \cap \mathcal{I} \neq \emptyset$ if and only if there exists $e \in \mathcal{I}$ and $e' \in \textit{proj}(e)$ such that $e'$ is not $\mathsf{MatchedMaxPriority}^{>}$-linearizable.

By Lemma \ref{lemma:pri execution is enough} and Lemma \ref{lemma:Lin Equals Constraint for EPQ1Lar}, we need to prove the following fact:

\noindent {\bf $\textit{fact}_1$}: Given a data-independent implementation $\mathcal{I}$. $\mathcal{A}_{\textit{1-lar}} \cap \mathcal{I} \neq \emptyset$ if and only if there exists $e \in \mathcal{I}$ and $e' \in \textit{proj}(e)$, $\mathsf{Has\text{-}MatchedMaxPriority}(e')$ holds, $x$ is the value with maximal priority in $e'$, $e'$ has only one maximal priority, and there is a cycle going through $x$ in $G$, where $G$ is the left-right constraint of $x$ in $e'$.

\noindent The $\textit{only if}$ direction: Let us consider the case of $\mathcal{A}_{\textit{l-lar}}^1$. Assume that $e_1 \in \mathcal{I}$ is accepted by $\mathcal{A}_{\textit{l-lar}}^1$. By data-independence, there exists data-differentiated execution $e \in \mathcal{I}$ and renaming function $r_1$, such that $e_1 = r_1(e)$. Assume that $r_1$ maps $d$ into $b$ and maps $f_1,\ldots,f_m$ into $a$. Let $e'$ be obtained from $e$ by projection into $\{ d, f_1,\ldots,f_m \}$. Assume that the priority of $b$ is $p$. It is easy to see that $e'$ has only one maximal priority, $\mathsf{Has\text{-}MatchedMaxPriority}(e')$ holds, and there is a cycle going through $d$ in $G$, where $G$ is the left-right constraint of $e'$. The case of $\mathcal{A}_{\textit{l-lar}}^2$, $\mathcal{A}_{\textit{l-lar}}^3$ and $\mathcal{A}_{\textit{l-lar}}^4$ can be similarly proved.

\noindent The $\textit{if}$ direction: Given such $e$, $e'$ and $x$. Let renaming function $r$ maps $x$ into $b$, maps values cover $x$ into $a$, and maps other values into $d$. By data-independence, $r(e) \in \mathcal{I}$. Then depending on the cases of $r(e) \vert_{b}$, we can see that $r(e)$ is accepted by $\mathcal{A}_{\textit{l-lar}}^1$, $\mathcal{A}_{\textit{l-lar}}^2$, $\mathcal{A}_{\textit{l-lar}}^3$ or $\mathcal{A}_{\textit{l-lar}}^4$. \qed
\end {proof}

\subsection{Proofs, Definitions and Register Automata in Subsection \ref{subsec:co-regular of EPQ1Equal}}
\label{sec:appendix proof and definition in section co-regular of EPQ1Equal}

Let $\textit{Items}(e,p)$ be the set of values with priority $p$ in execution $e$. The following lemma states a method to build the linearization of $e$ w.r.t $\mathsf{MatchedMaxPriority}^{=}$.

\begin{restatable}{lemma}{MaximalInPBadGPMakePQ1Equal}
\label{lemma:maximal in pb and gap-point make a candidate of EPQ1Equal}
Given a data-differentiated execution $e$ where $\mathsf{Has\text{-}MatchedMaxPriority}^{=}(e)$ holds and $p$ is its maximal priority. If there exists a value $x$ with priority $p$, such that for each $y \in \textit{Items}(e,p)$, (1) $x\not<_{\textit{pb}}y$, and (2) the right-most gap-point of $x$ is after $\textit{call}(\textit{put},y,p)$ and $\textit{call}(\textit{rm},y)$. Then $e$ is $\mathsf{MatchedMaxPriority}^{=}$-linearizable.
\end{restatable}

\begin {proof}

Let $o$ be the right-most gap-point of $x$. We locate linearization points of each operation as follows:

\begin{itemize}
\setlength{\itemsep}{0.5pt}
\item[-] Locate the linearization point of $\textit{rm}(x)$ at $o$,

\item[-] If $\textit{put}(x,p)$ overlaps with $\textit{rm}(x)$, then locate the linearization point of $\textit{put}(x,p)$ just before the linearization point of $\textit{rm}(x)$. Otherwise, $\textit{put}(x,p) <_{\textit{hb}} \textit{rm}(x)$, and we locate the linearization point of $\textit{put}(x,p)$ just before its return action.

\item[-] Locate linearization points of operation of each $y \in \textit{Items}(e,p)$ (except for $x$) just after the call action of the operation.

\item[-] For value $z$ with priority smaller than $p$. If both $\textit{call}(\textit{put},z,\_)$ and $\textit{call}(\textit{rm},z)$ is before $o$, then locate the linearization points of $\textit{put}(z,\_)$ and $\textit{rm}(z)$ just after their call actions. If both $\textit{ret}(\textit{put},z)$ and $\textit{ret}(\textit{rm},z)$ (if exists) is after $o$, then locate the linearization points of $\textit{put}(z,\_)$ and $\textit{rm}(z)$ just before their return actions. Otherwise, $x$ is in interval of $z$, which contradicts the definition of gap-point, and is impossible.
\end{itemize}

Let $l$ be the sequence of linearization points constructed above. It is obvious that $e \sqsubseteq l$. Since for each $y \in \textit{Items}(e,p)$, $o$ is after $\textit{call}(\textit{put},y,\_)$ and $\textit{call}(\textit{rm},x)$, we can see that $\textit{rm}(x)$ is after $\textit{put}(y,p)$ and $\textit{rm}(y)$ in $l$. It is obvious that $\textit{put}(x,p)$ is before $\textit{rm}(x)$ in $l$. Since $x$ does not $<_{\textit{pb}}$ to $y$, we can see that no $\textit{put}(y,p)$ happens before $\textit{put}(x,p)$. Then it is easy to see that $\textit{put}(x,p)$ is after $\textit{put}(y,p)$ in $l$. Since $\mathsf{Has\text{-}MatchedMaxPriority}^{=}(e)$ holds, all other values in $\textit{Items}(e,p)$ has matched $\textit{put}$ and $\textit{rm}$, and it is easy to see that their $\textit{put}$ and $\textit{rm}$ (except for that of $x$) are all before $\textit{rm}(x)$ in $l$.

For value $z$ with priority smaller than $p$, we can see that there are only two possibilities: (1) $\textit{put}(z,\_)$ and $\textit{rm}(z)$ are both before $\textit{rm}(x)$ in $l$, and (2) $\textit{put}(z,\_)$ and $\textit{rm}(z)$ (if exists) are after before $\textit{rm}(x)$ in $l$. Therefore, before $\textit{rm}(x)$ in $l$, the $\textit{put}$ and $\textit{rm}$ of $z$ are matched.

Therefore, it is easy to see that $\mathsf{MatchedMaxPriority^{=}\text{-}Seq}(l,x)$ holds and $e$ is $\mathsf{MatchedMaxPriority}^{=}$-linearizable. \qed
\end {proof}

With Lemma \ref{lemma:maximal in pb and gap-point make a candidate of EPQ1Equal}, we can prove the following lemma, which states that getting rid of case in \figurename~\ref{fig:introduce pb order} is enough for ensure that $\mathsf{MatchedMaxPriority}^{=}(e)$ holds.


$\newline$
{\noindent \bf Lemma \ref{lemma:EPQ1Equal as pb order and gap-point}}: Let $e$ be a data-differentiated execution which contains only one maximal priority $p$ such that $\mathsf{Has\text{-}MatchedMaxPriority}(e)$ holds.
Then, $e$ is not linearizable w.r.t $\mathsf{MatchedMaxPriority}^{=}$ iff $e$ contains two values $x$ and $y$ of maximal priority $p$ such that $y <_{\textit{pb}}^* x$, and the rightmost gap-point of $x$ is strictly smaller than the index of $\textit{call}(\textit{put},y,p)$ or $\textit{call}(\textit{rm},y)$.

\begin {proof}
To prove the $\textit{if}$ direction, let $e_{x,y}$ be the execution that is obtained from $e$ by erasing all actions of values that has same priority as $x$, except for actions of $x$ and $y$. It is obvious that $\mathsf{Has\text{-}MatchedMaxP-}$ $\mathsf{riority}^{=}(e_{x,y})$ holds. Since $y <_{\textit{pb}}^* x$, we can see that $x$ should be chosen as $\alpha$ in $\mathsf{MatchedMaxPriority}^{=}$ .

According to Lemma \ref{lemma:Lin Equals Constraint for EPQ1Lar} (Here we temporarily forget the existence of $y$), the only possible position for locating linearizaton point of $\textit{rm}(x)$ is at gap-point of $x$. Otherwise, if the linearizaton point of $\textit{rm}(x)$ is chosen at a position that is not a gap-point of $x$, then there exists unmatched operation before $\textit{rm}(x)$ with smaller priority. Since the rightmost gap-point of $x$ is before $\textit{call}(\textit{put},y,p)$ or $\textit{call}(\textit{rm},y)$, if we locate linearizaton point of $\textit{rm}(x)$ at gap-point of $x$, then $\textit{rm}(x)$ will be before $\textit{call}(\textit{put},y,p)$ or $\textit{call}(\textit{rm},x)$.

Therefore, for every sequence $l = u \cdot \textit{put}(x,p) \cdot v \cdot \textit{rm}(x) \cdot w$, if $e_{x,y} \sqsubseteq l$, then either $u \cdot v$ contains some unmatched operations of priority smaller than $p$, or $w$ contains $\textit{put}(y,p)$ or $\textit{rm}(y)$. In both cases, $\mathsf{MatchedMaxPriority^{=}\text{-}Seq}(l,x)$ does not hold.

To prove the $\textit{only if}$ direction, we prove its contrapositive. Assume we already know that for each $x$ and $y$ has maximal priority in $e$, if $y <_{\textit{pb}}^* x$, then the rightmost gap-point of $x$ is after $\textit{call}(\textit{put},y,p)$ and $\textit{call}(\textit{rm},x)$. We need to prove that $e$ is $\mathsf{MatchedMaxPriority}^{=}$-linearizable. Recall that we already assume that each single-priority execution has FIFO property, and value with larger priority is not covered by values with smaller priority.

Our proof proceed as follows:

\begin{itemize}
\setlength{\itemsep}{0.5pt}
\item[-] Let $e_{p}$ be the projection of $e$ into operations of priority $p$. Since each single-priority execution has FIFO property, there exists sequence $l_{p}$, such that $e_{p} \sqsubseteq l_{p}$, and when we treat $\textit{put}$ as $\textit{enq}$ and $\textit{rm}$ as $\textit{deq}$, $l_{p}$ belongs to queue.

\item[-] Let $a_1$ be the last inserted value of $l_{p}$.

    Step $1$: Check whether for each $b \in \textit{Items}(e,p)$, (1) $a_1$ does not $<_{\textit{pb}}$ to $b$, and (2) the right-most gap-point of $a$ is after $\textit{call}(\textit{put},b,p)$ and $\textit{call}(\textit{rm},b)$.

    It is easy to see that $a_1$ is of priority $p$, and $a_1$ does not $<_{\textit{pb}}$ to any $b \in \textit{Items}(e,p)$. If for each $b \in \textit{Items}(e,p)$, the rightmost gap-point of $a_1$ is after $\textit{call}(\textit{put},b,p)$ and $\textit{call}(\textit{rm},b)$. Then by Lemma \ref{lemma:maximal in pb and gap-point make a candidate of EPQ1Equal}, we can obtain that $e$ is $\mathsf{MatchedMaxPriority}^{=}$-linearizable.

\item[-] Otherwise, there exists $a_2 \in \textit{Items}(e,p)$, such that the rightmost gap-point of $a_1$ is before $\textit{call}(\textit{put},a_2,p)$ or $\textit{call}(\textit{rm},a_2)$ in $e$. We can see that each gap-point of $a_2$ is after the rightmost gap-point of $a_1$.
    By assumption, we know that $a_2$ does not $<_{\textit{pb}}$ to $a_1$.

    \begin{itemize}
    \setlength{\itemsep}{0.5pt}
    \item[-] If for each value $b \in \textit{Items}(e,p)$, $a_2$ does not $<_{\textit{pb}}$ to $b$. Then we go to step $1$ and treat $a_2$ similarly as $a_1$.
    \item[-] Otherwise, there exists $a_3$ with priority $p$ such that $a_2 <_{\textit{pb}}^* a_3$.

    Since $l_{p}$ has FIFO property, it is easy to see that there is no cycle in $<_{\textit{pb}}$ order. It is safe to assume that $a_3$ is maximal in the sense of $<_{\textit{pb}}^*$. Or we can say, there does not exists $a_4$, such that $a_3 <_{\textit{pb}}^* a_4$.

    By assumption,we know that the rightmost gap-point of $a_3$ is after $\textit{call}(\textit{put},a_2,p)$ and $\textit{call}(\textit{rm},a_2)$. Therefore, we can see that the rightmost gap-point of $a_3$ is after the rightmost gap-point of $a_1$. Then we go to step $1$ and treat $a_3$ similarly as $a_1$.
    \end{itemize}
\end{itemize}

Let $a^i$ be the $a_1$ in the $\textit{i-th}$ loop of our proof. It is not hard to see that, given $i<j$, the rightmost gap-point of $a^j$ is after the rightmost gap-point of $a^i$. Therefore, the loop finally stop at some $a^f$. $a^f$ satisfies the check of Step $1$. By Lemma \ref{lemma:maximal in pb and gap-point make a candidate of EPQ1Equal}, this implies that $e$ is $\mathsf{MatchedMaxPriority}^{=}$-linearizable. This completes the proof of $\textit{if}$ direction. \qed
\end {proof}

According to the definition of $<_{\textit{ob}}^*$, if $a <_{\textit{pb}}^* b$, then there exists $a_1,\ldots,a_m$, such that $a <_{\textit{pb}} a_1 <_{\textit{pb}} \ldots <_{\textit{pb}} a_m <_{\textit{pb}} b$. The following lemma states that, the number of intermediate values $a_i$ is in fact bounded.


$\newline$
{\noindent \bf Lemma \ref{lemma:ob order has bounded length}}: Let $e$ be a data-differentiated execution such that $a <_{\textit{pb}} a_1 <_{\textit{pb}} \ldots <_{\textit{pb}} a_m <_{\textit{pb}} b$ holds for some set of values $a$, $a_1$,$\ldots$,$a_m$, $b$. Then, one of the following holds:
\begin{itemize}
\setlength{\itemsep}{0.5pt}
\item[-] $a <_{\textit{pb}}^A b$, $a <_{\textit{pb}}^B b$, or $a <_{\textit{pb}}^C b$,

\item[-] $a <_{\textit{pb}}^A a_i <_{\textit{pb}}^B b$ or $a <_{\textit{pb}}^B a_i <_{\textit{pb}}^A b$, for some $i$.
\end{itemize}

\begin {proof}

Our proof proceed as follows:

\begin{itemize}
\setlength{\itemsep}{0.5pt}
\item[-] ($<_{\textit{pb}}^A \cdot <_{\textit{pb}}^A$,$<_{\textit{pb}}^B \cdot <_{\textit{pb}}^B$ and $<_{\textit{pb}}^C \cdot <_{\textit{pb}}^C$): If $c_3 <_{\textit{pb}}^A c_2 <_{\textit{pb}}^A c_1$, then $\textit{put}(c_3,\_)$ happens before $\textit{put}(c_2,\_)$, and $\textit{put}(c_2,\_)$ happens before $\textit{put}(c_1,\_)$. Therefore, it is obvious that $\textit{put}(c_3,\_)$ happens before $\textit{put}(c_1,\_)$ and $c_3 <_{\textit{pb}}^A c_1$.

    Similarly, if $c_3 <_{\textit{pb}}^B c_2 <_{\textit{pb}}^B c_1$, then $c_3 <_{\textit{pb}}^B c_1$.

    If $c_3 <_{\textit{pb}}^C c_2 <_{\textit{pb}}^C c_1$: Since $c_2 <_{\textit{pb}}^C c_1$, $\textit{ret}(\textit{rm},c_2)$ is before $\textit{call}(\textit{put},c_1,\_)$. Since $\textit{rm}(c_2)$ does not happen before $\textit{put}(c_2,\_)$, $\textit{call}(\textit{put},c_2,\_)$ is before $\textit{ret}(\textit{rm},c_2)$. Since $c_3 <_{\textit{pb}}^C c_2$, $\textit{ret}(\textit{rm},c_3)$ is before $\textit{call}(\textit{put},c_2,\_)$. Therefore, $\textit{ret}(\textit{rm},c_3)$ is before $\textit{call}(\textit{put},c_1,\_)$, and $c_3 <_{\textit{pb}}^C c_1$.

    Therefore, when we meet successive $<_{\textit{pb}}^A$, it is safe to leave only the first and the last elements and ignore intermediate elements. Similar cases hold for $<_{\textit{pb}}^B$ and $<_{\textit{pb}}^C$.

\item[-] $<_{\textit{pb}}^A$ and $<_{\textit{pb}}^C$:

    \begin{itemize}
    \setlength{\itemsep}{0.5pt}
    \item[-] ($<_{\textit{pb}}^A \cdot <_{\textit{pb}}^C$): If $c_3 <_{\textit{pb}}^A c_2 <_{\textit{pb}}^C c_1$. Since $c_2 <_{\textit{pb}}^C c_1$, $\textit{ret}(\textit{rm},c_2)$ is before $\textit{call}(\textit{put},c_1,\_)$. Since $\textit{rm}(c_2)$ does not happen before $\textit{put}(c_2,\_)$, $\textit{call}(\textit{put},c_2,\_)$ is before $\textit{ret}(\textit{rm},c_2)$. Since $c_3 <_{\textit{pb}}^A c_2$, $\textit{ret}(\textit{put},c_3,\_)$ is before $\textit{call}(\textit{put},c_2,\_)$. Therefore, $\textit{ret}(\textit{put},c_3)$ is before $\textit{call}(\textit{put},c_1,\_)$, and $c_3 <_{\textit{pb}}^A c_1$.

    \item[-] ($<_{\textit{pb}}^C \cdot <_{\textit{pb}}^A$): If $c_3 <_{\textit{pb}}^C c_2 <_{\textit{pb}}^A c_1$. Since $c_2 <_{\textit{pb}}^A c_1$, $\textit{ret}(\textit{put},c_2,\_)$ is before $\textit{call}(\textit{put},c_1,\_)$. It is obvious that $\textit{call}(\textit{put},c_2,\_)$ is before $\textit{ret}(\textit{put},c_2,\_)$. Since $c_3 <_{\textit{pb}}^C c_2$, $\textit{ret}(\textit{rm},c_3)$ is before $\textit{call}(\textit{put},c_2,\_)$. Therefore, $\textit{ret}(\textit{rm},c_3)$ is before $\textit{call}(\textit{put},c_1,\_)$, and $c_3 <_{\textit{pb}}^C c_1$.

    \end{itemize}

\item[-] $<_{\textit{pb}}^B$ and $<_{\textit{pb}}^C$:

    \begin{itemize}
    \setlength{\itemsep}{0.5pt}
    \item[-] ($<_{\textit{pb}}^B \cdot <_{\textit{pb}}^C$): If $c_3 <_{\textit{pb}}^B c_2 <_{\textit{pb}}^C c_1$. Since $c_2 <_{\textit{pb}}^C c_1$, $\textit{ret}(\textit{rm},c_2)$ is before $\textit{call}(\textit{put},c_1,\_)$. It is obvious that $\textit{call}(\textit{rm},c_2)$ is before $\textit{ret}(\textit{rm},c_2)$. Since $c_3 <_{\textit{pb}}^B c_2$, $\textit{ret}(\textit{rm},c_3)$ is before $\textit{call}(\textit{rm},c_2)$. Therefore, $\textit{ret}(\textit{rm},c_3)$ is before $\textit{call}(\textit{put},c_1,\_)$, and $c_3 <_{\textit{pb}}^C c_1$.

    \item[-] ($<_{\textit{pb}}^C \cdot <_{\textit{pb}}^B$): If $c_3 <_{\textit{pb}}^C c_2 <_{\textit{pb}}^B c_1$. Since $c_2 <_{\textit{pb}}^B c_1$, $\textit{ret}(\textit{rm},c_2)$ is before $\textit{call}(\textit{rm},c_1)$. Since $\textit{rm}(c_2)$ does not happen before $\textit{put}(c_2,\_)$, $\textit{call}(\textit{put},c_2,\_)$ is before $\textit{ret}(\textit{rm},c_2)$. Since $c_3 <_{\textit{pb}}^C c_2$, $\textit{ret}(\textit{rm},c_3)$ is before $\textit{call}(\textit{put},c_2,\_)$. Therefore, $\textit{ret}(\textit{rm},c_3)$ is before $\textit{call}(\textit{rm},c_1)$, and $c_3 <_{\textit{pb}}^B c_1$.
    \end{itemize}

\item[-]  ($<_{\textit{pb}}^A \cdot <_{\textit{pb}}^B \cdot <_{\textit{pb}}^A$): If $c_4 <_{\textit{pb}}^A c_3 <_{\textit{pb}}^B c_2 <_{\textit{pb}}^A c_1$:
    \begin{itemize}
    \setlength{\itemsep}{0.5pt}
    \item[-] If $\textit{call}(\textit{rm},c_2)$ is before $\textit{call}(\textit{put},c_1,\_)$: Since $c_3 <_{\textit{pb}}^B c_2$, $\textit{ret}(\textit{rm},c_3)$ is before $\textit{call}(\textit{rm},c_2)$. Then $\textit{ret}(\textit{rm},c_3)$ is before $\textit{call}(\textit{put},c_1,\_)$, and $c_3 <_{\textit{pb}}^C c_1$. This implies that $c_4 <_{\textit{pb}}^A c_3 <_{\textit{pb}}^C c_1$. According to the fact for $<_{\textit{pb}}^A \cdot <_{\textit{pb}}^C$, we know that $c_4  <_{\textit{pb}}^A c_1$.

    \item[-] If $\textit{call}(\textit{rm},c_2)$ is after $\textit{call}(\textit{put},c_1,\_)$: Since $c_2 <_{\textit{pb}}^A c_1$, $\textit{ret}(\textit{put},c_2,\_)$ is before $\textit{call}(\textit{put},c_1,\_)$. Since $c_3 <_{\textit{pb}}^B c_2$, $\textit{rm}(c_3)$ happens before $\textit{rm}(c_2)$, and then we know that $\textit{put}(c_2,\_)$ can not happen before $\textit{put}(c_3,\_)$. Since $\textit{put}(c_2,\_)$ does not happen before $\textit{put}(c_3,\_)$, $\textit{call}(\textit{put},c_3,\_)$ is before $\textit{ret}(\textit{put},c_2,\_)$. Since $c_4 <_{\textit{pb}}^A c_3$, $\textit{ret}(\textit{put},c_4,\_)$ is before $\textit{call}(\textit{put},c_3,\_)$. Therefore, $\textit{ret}(\textit{put},c_4,\_)$ is before $\textit{call}(\textit{put},c_1,\_)$, and $c_4 <_{\textit{pb}}^A c_1$.
    \end{itemize}

\item[-]  ($<_{\textit{pb}}^B \cdot <_{\textit{pb}}^A \cdot <_{\textit{pb}}^B$): If $c_4 <_{\textit{pb}}^B c_3 <_{\textit{pb}}^A c_2 <_{\textit{pb}}^B c_1$: Since $c_2 <_{\textit{pb}}^B c_1$, $\textit{ret}(\textit{rm},c_2)$ is before $\textit{call}(\textit{rm},c_1)$. Since $c_3 <_{\textit{pb}}^A c_2$, we can see that $\textit{put}(c_3,\_) <_{\textit{hb}} \textit{put}(c_2,\_)$. Since each single-priority execution has FIFO property, we know that $\textit{rm}(c_2)$ does not happen before $\textit{rm}(c_3)$, and thus, $\textit{call}(\textit{rm},c_3)$ is before $\textit{ret}(\textit{rm},c_2)$. Since $c_4 <_{\textit{pb}}^B c_3$, $\textit{ret}(\textit{rm},c_4)$ is before $\textit{call}(\textit{rm},c_3)$. Therefore, $\textit{ret}(\textit{rm},c_4)$ is before $\textit{call}(\textit{rm},c_1)$, and $c_4 <_{\textit{pb}}^B c_1$.

\end{itemize}

Based on above results, given $a <_{\textit{pb}}^{b_1} a_1 <_{\textit{pb}} \ldots <_{\textit{pb}}^{b_m} a_m <_{\textit{pb}}^{b_{\textit{m+1}}} b$, where each $b_i$ is in $\{ A,B,C \}$, we can merge relations, until we get one of the following facts:

\begin{itemize}
\setlength{\itemsep}{0.5pt}
\item[-] $a <_{\textit{pb}}^A b$, $a <_{\textit{pb}}^B b$ or $a <_{\textit{pb}}^C b$,

\item[-] $a <_{\textit{pb}}^A a_i <_{\textit{pb}}^B b$, or $a <_{\textit{pb}}^B a_i <_{\textit{pb}}^A b$, for some $i$,
\end{itemize}

This completes the proof of this lemma. \qed
\end {proof}

There are many enumerations of operations of $a$, $b$ and $a_1$ that may makes $a <_{\textit{pb}}^* b$. The following lemma states that since some of them is not consistent with the requirements of gap-point, the number of potential enumerations can be further reduced into only five.

\begin{restatable}{lemma}{FiveEnmuerationisEnoughForEPQOneEqual}
\label{lemma:five enumeration is enough for EPQ1Equal}
Given a data-differentiated $p$-execution $e$ where $\mathsf{Has\text{-}MatchedMaxPriority}^{=}(e)$ holds. Let $a$ and $b$ be values with maximal priority $p$. Assume that $a <_{\textit{pb}}^* b$, and the rightmost gap-point of $b$ is before $\textit{call}(\textit{put},a,p)$ or $\textit{call}(\textit{rm},a)$. Then, there are five possible enumeration of operations of $a$, $b$, $a_1$ (if exists), where $a_1$ is the possible intermediate value for obtain $a <_{\textit{pb}}^* b$.
\end{restatable}
\begin {proof}

Let us prove by consider all the possible reasons of $a <_{\textit{pb}}^* b$. According to Lemma \ref{lemma:ob order has bounded length}, we need to consider five reasons: Let $o$ be the right-most gap-point of $b$.

\begin{itemize}
\setlength{\itemsep}{0.5pt}
\item[-] Reason $1$, $a <_{\textit{pb}}^A b$:

    Since $a <_{\textit{pb}}^A b$, $\textit{put}(a,p) <_{\textit{hb}} \textit{put}(b,p)$. Since $o$ is after $\textit{call}(\textit{put},b,p)$, and thus, after $\textit{call}(\textit{put},a,p)$, we can see that $o$ is before $\textit{call}(\textit{rm},b)$.

    Since single-priority execution must satisfy the FIFO property, $\textit{rm}(b) \not <_{\textit{hb}} \textit{rm}(a)$, and thus, $\textit{call}(\textit{rm},a)$ is before $\textit{ret}(\textit{rm},b)$. If $\textit{call}(\textit{rm},a)$ is before $\textit{call}(\textit{rm},b)$, then $o$ is also a gap-point of $a$ and contradicts our assumption. So we know that $\textit{call}(\textit{rm},a)$ is after $\textit{call}(\textit{rm},b)$. If $\textit{ret}(\textit{rm},a)$ is before $\textit{ret}(\textit{rm},b)$, since we already assume that there exists gap-point of $a$, this gap-point is also a gap-point of $b$, and is after $o$, which contradicts that $o$ is the rightmost gap-point of $b$. Therefore, $\textit{ret}(\textit{rm},a)$ is after $\textit{ret}(\textit{rm},b)$.

    According to above discussion, there are two possible enumeration of operations of $a$ and $b$, as shown in \figurename~\ref{fig:history enumeration 1 for PQ1Equal} and \figurename~\ref{fig:history enumeration 2 for PQ1Equal}. Here we explicitly draw the leftmost gap-point of $a$ as $o'$. Since the position of $\textit{ret}(\textit{put},b,\_)$ does not influence the correctness, we can simply ignore it.

\item[-] Reason $2$, $a <_{\textit{pb}}^B b$:

    Since $a <_{\textit{pb}}^B b$, $\textit{ret}(\textit{rm},a)$ is before $\textit{call}(\textit{rm},b)$. Since $o$ is after $\textit{call}(\textit{rm},b)$, we can see that $o$ is before $\textit{call}(\textit{put},a,p)$. This implies that $\textit{ret}(\textit{rm},a)$ is before $\textit{call}(\textit{put},a,p)$, and then $\textit{rm}(a) <_{\textit{hb}} \textit{put}(a)$, which is impossible. Therefore, we can safely ignore this reason.

\item[-] Reason $3$, $a <_{\textit{pb}}^C b$:

    Since $a <_{\textit{pb}}^B b$, $\textit{ret}(\textit{rm},a)$ is before $\textit{call}(\textit{put},b,p)$. Since $o$ is after $\textit{call}(\textit{put},b)$, we can see that $o$ is before $\textit{call}(\textit{put},a,p)$. This implies that $\textit{ret}(\textit{rm},a)$ is before $\textit{call}(\textit{put},a,p)$, and then $\textit{rm}(a) <_{\textit{hb}} \textit{put}(a)$, which is impossible. Therefore, we can safely ignore this reason.

\item[-] Reason $4$, $a <_{\textit{pb}}^A a_1 <_{\textit{pb}}^B b$:

    Since $a_1 <_{\textit{pb}}^B b$, $\textit{rm}(a_1) <_{\textit{hb}} \textit{rm}(b)$, and $\textit{ret}(\textit{rm},a_1)$ is before $\textit{call}(\textit{rm},b)$. Since $\textit{rm}(a_1)$ does not happen before $\textit{put}(a_1)$, $\textit{call}(\textit{put},a_1,p)$ is before $\textit{ret}(\textit{rm},a_1)$. Since $a <_{\textit{pb}}^A a_1$, $\textit{ret}(\textit{put},a,p)$ is before $\textit{call}(\textit{put},a_1,p)$. Therefore, $\textit{ret}(\textit{put},a,p)$ is before $\textit{call}(\textit{rm},b)$. Since $\textit{call}(\textit{rm},b)$ is before $o$, we can see that $o$ is before $\textit{call}(\textit{rm},a)$.

    If $\textit{call}(\textit{rm},a)$ is after $\textit{ret}(\textit{rm},b)$, then $e \vert_{ \{ a,a_1,b \} }$ violates the FIFO property. Therefore, $\textit{call}(\textit{rm},a)$ is before $\textit{ret}(\textit{rm},b)$. Similarly as the case of reason $1$, we can see that $\textit{ret}(\textit{rm},b)$ is before $\textit{ret}(\textit{rm},a)$.

    According to above discussion, there are three possible enumeration of operations of $a$, $a_1$ and $b$, as shown in \figurename~\ref{fig:history enumeration 3 for PQ1Equal}, \figurename~\ref{fig:history enumeration 4 for PQ1Equal} and \figurename~\ref{fig:history enumeration 5 for PQ1Equal}. Here we explicitly draw the leftmost gap-point of $a$ as $o'$. Since the position of $\textit{ret}(\textit{put},a_1,p)$ and $\textit{call}(\textit{put},a,p)$ do not influence the correctness, we can simply ignore it. We also ignore $\textit{call}(\textit{put},b,p)$ and $\textit{ret}(\textit{put},b,\_)$, since the only requirements of them are (1) $\textit{rm}(b) \not <_{\textit{hb}} \textit{put}(b)$ and (2) $\textit{call}(\textit{put},b,p)$ is before $o$.

\item[-] Reason $5$, $a <_{\textit{pb}}^B a_1 <_{\textit{pb}}^A b$:

    Since $a_1 <_{\textit{pb}}^A b$, $\textit{ret}(\textit{put},a_1)$ is before $\textit{call}(\textit{put},b,p)$. Since $\textit{call}(\textit{put},b,p)$ is before $o$, we can see that $\textit{ret}(\textit{put},a_1,\_)$ is before $o$.

    \begin{itemize}
    \setlength{\itemsep}{0.5pt}
    \item[-] If $o$ is before $\textit{call}(\textit{rm},a)$: Then $o$ is obviously before $\textit{ret}(\textit{rm},a)$. Since $a <_{\textit{pb}}^B a_1$, $\textit{ret}(\textit{rm},a)$ is before $\textit{call}(\textit{rm},a_1)$. Then we can see that, $o$ is before $\textit{call}(\textit{rm},a_1)$, and remember that $a_1 <_{\textit{pb}}^A b$. Then we can goto the case of reason $1$ and treat $a_1$ as $a$. Therefore, we can safely ignore this.

    \item[-] If $o$ is before $\textit{call}(\textit{put},a,p)$: Since $\textit{rm}(a)$ does not happen before $\textit{put}(a,p)$, we can see that $\textit{call}(\textit{put},a,p)$ is before $\textit{ret}(\textit{rm},a)$, and then $o$ is before $\textit{ret}(\textit{rm},a)$. Then similarly as above case, we can see that $o$ is before $\textit{call}(\textit{rm},a_1)$, and $a_1 <_{\textit{pb}}^A b$. Then we can goto the case of reason $1$ and treat $a_1$ as $a$. Therefore, we can safely ignore this.
    \end{itemize}
\end{itemize}

This completes the proof of this lemma. \qed
\end {proof}

\begin{figure}[htbp]
  \centering
  \includegraphics[width=0.4 \textwidth]{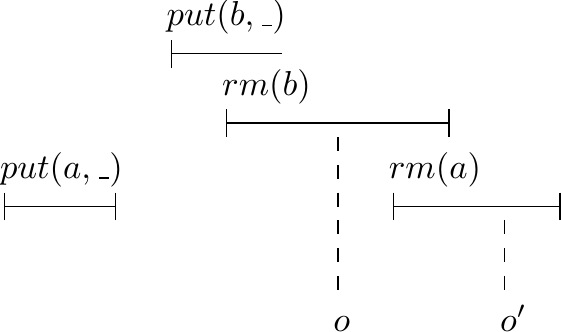}
  \caption{The first possible enumeration.}
  \label{fig:history enumeration 1 for PQ1Equal}
\end{figure}

\begin{figure}[htbp]
  \centering
  \includegraphics[width=0.4 \textwidth]{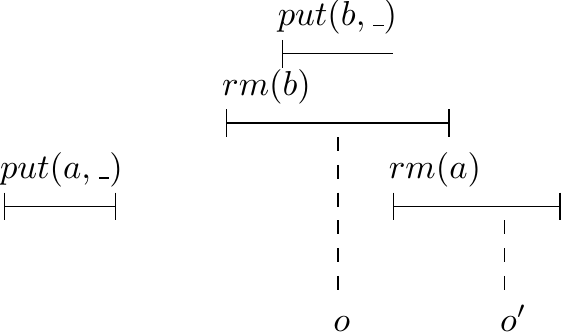}
  \caption{The second possible enumeration.}
  \label{fig:history enumeration 2 for PQ1Equal}
\end{figure}

\begin{figure}[htbp]
  \centering
  \includegraphics[width=0.4 \textwidth]{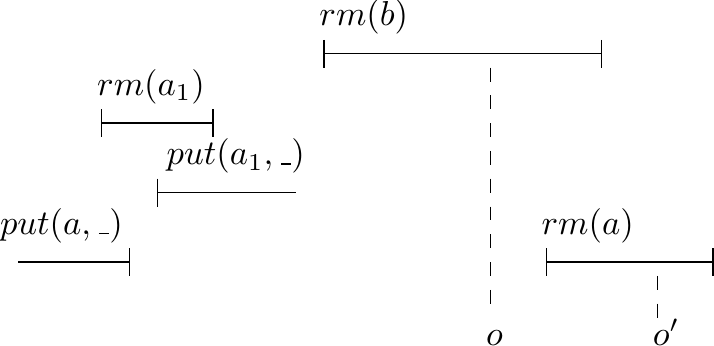}
  \caption{The third possible enumeration.}
  \label{fig:history enumeration 3 for PQ1Equal}
\end{figure}

\begin{figure}[htbp]
  \centering
  \includegraphics[width=0.4 \textwidth]{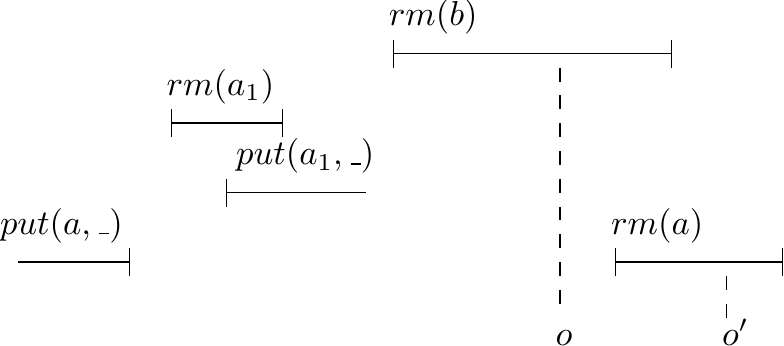}
  \caption{The forth possible enumeration.}
  \label{fig:history enumeration 4 for PQ1Equal}
\end{figure}

\begin{figure}[htbp]
  \centering
  \includegraphics[width=0.4 \textwidth]{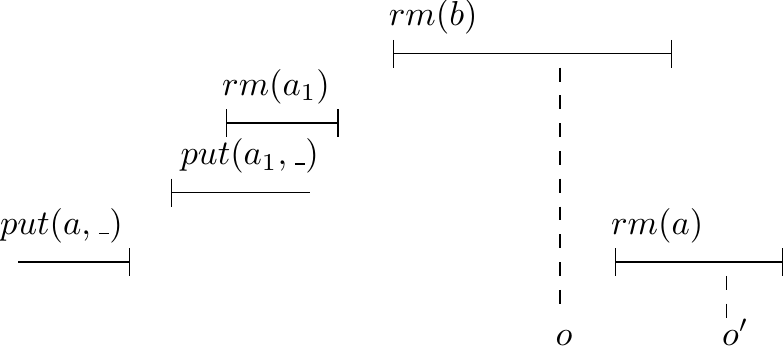}
  \caption{The fifth possible enumeration.}
  \label{fig:history enumeration 5 for PQ1Equal}
\end{figure}

Given a data-differentiated execution $e$ with only one maximal priority, two actions $\textit{act}_1$, $\textit{act}_2$ of maximal priority in $e$, and assume that $\textit{act}_1$ is before $\textit{act}_2$ in $e$.
we say that $\textit{act}_1$, $\textit{act}_2$ is covered by values $d_1,\ldots,d_m$ in $e$, if the priorities of $d_1,\ldots,d_m$ is smaller than that of $\textit{act}_1$ and $\textit{act}_2$, and

\begin{itemize}
\setlength{\itemsep}{0.5pt}
\item[-] $\textit{ret}(\textit{put},d_m,\_)$ is before $\textit{act}_1$,

\item[-] For each $i < 1 \leq m$,$\textit{put}(d_{\textit{i-1}},\_) <_{\textit{hb}} \textit{rm}(d_i)$,

\item[-] $\textit{act}_2$ is before $\textit{call}(\textit{rm},d_1)$.
\end{itemize}

According to Lemma \ref{lemma:EPQ1Equal as pb order and gap-point}, Lemma \ref{lemma:ob order has bounded length} and Lemma \ref{lemma:five enumeration is enough for EPQ1Equal}, it is not hard to prove that, given a data-differentiated execution $e$ with only one maximal priority and $\mathsf{Has\text{-}MatchedMaxPriority}^{=}(e)$ holds. $e$ is not $\mathsf{MatchedMaxPriority}$-linearizable, if and only if, one of enumerations holds in $e$ (permit renaming), while $\textit{call}(\textit{rm},a)$ and $\textit{ret}(\textit{rm},b)$ is covered by some $d_1,\ldots,d_m$, $\textit{call}(\textit{rm},b)$ is before $\textit{ret}(\textit{put},d_m,\_)$, and $\textit{call}(\textit{rm},d_1)$ is before $\textit{ret}(\textit{rm},a)$. We say that such $d_1,\ldots,d_m$ constitute the rightmost gap of $b$.

Let us begin to represent several register automata that is used to capture the existence of a data-differentiated execution $e$, $e$ has a projection $e'$, $e'$ has only one maximal priority, $\mathsf{Has\text{-}MatchedMax-}$ $\mathsf{Priority}(e')$ holds, there exists values $a$ and $b$ with maximal priority in $e'$, $a <_{\textit{pb}}^* b$, and the rightmost gap-point of $b$ is before $\textit{call}(\textit{put},a,\_)$ or $\textit{call}(\textit{rm},a)$.

An automaton $\mathcal{A}_{\textit{l-eq}}^1$ is given in \figurename~\ref{fig:automata for first enumeration of PQ1Equal}, and it is constructed for the first enumeration in \figurename~\ref{fig:history enumeration 1 for PQ1Equal}. Here we rename the values that cover $\textit{call}(\textit{rm},a)$ and $\textit{ret}(\textit{rm},b)$ into $d$, and rename the remanning values into $\top$. In this figure, $C = \{ \textit{call}(\textit{put},\top,\textit{true}),\textit{ret}(\textit{put},\top,\textit{true})$, $\textit{call}(\textit{rm},\top), \textit{ret}(\textit{rm},\top),$ $\textit{call}(\textit{rm},\textit{empty}),\textit{ret}(\textit{rm},\textit{empty}) \}$, $C_1 = C \cup \{ \textit{call}(\textit{put},d,<r) \}$, $C_2 = C_1 \cup \{ \textit{ret}(\textit{put},b,=r) \}$, $C_3 = C_2 \cup \{ \textit{call}(\textit{put},d,<r),\textit{ret}(\textit{rm},d) \}$, $C_4 = C \cup \{ \textit{ret}(\textit{put},b,=r), \textit{ret}(\textit{rm},d) \}$.

\begin{figure}[htbp]
  \centering
  \includegraphics[width=0.8 \textwidth]{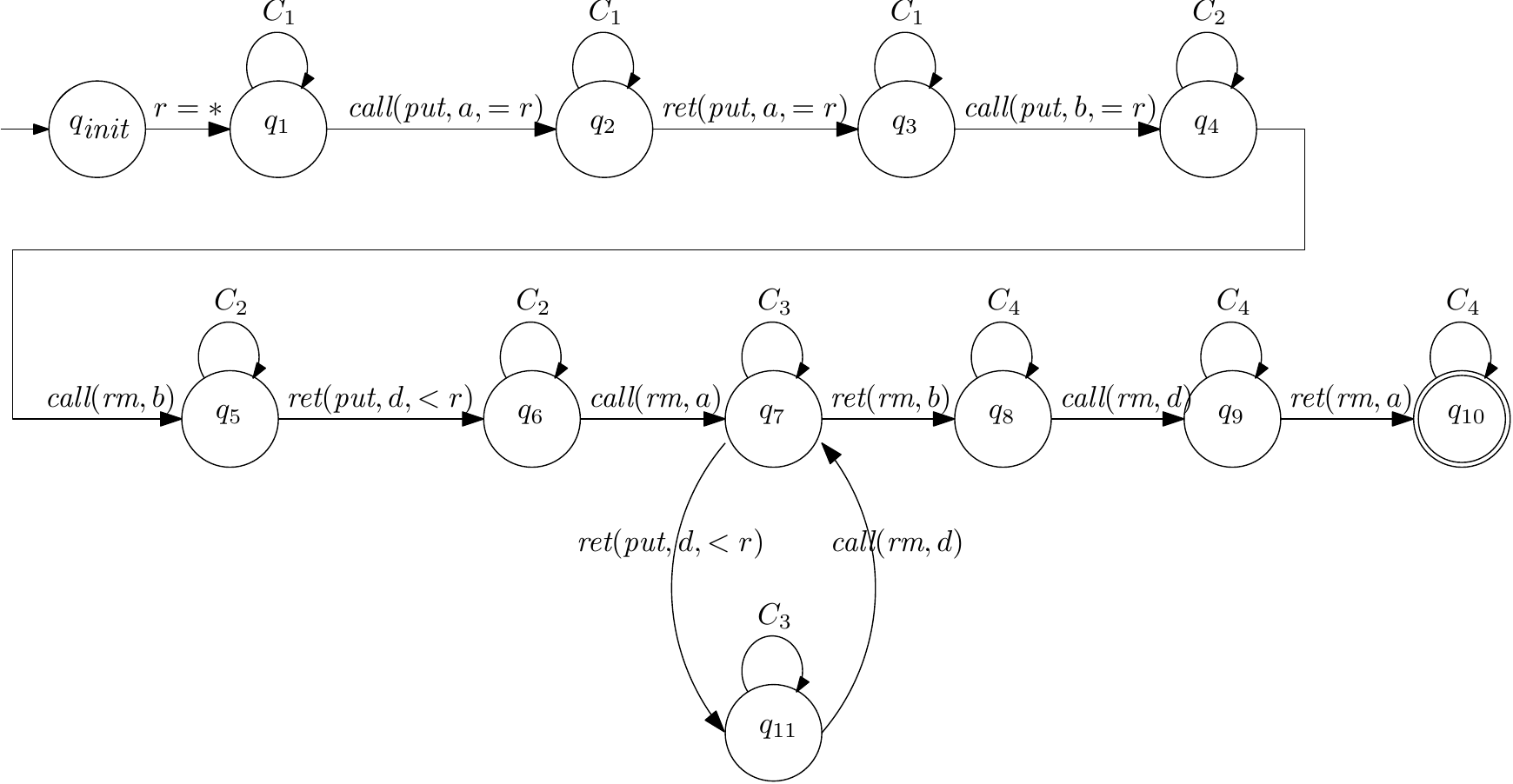}
  \caption{Automaton $\mathcal{A}_{\textit{l-eq}}^1$}
  \label{fig:automata for first enumeration of PQ1Equal}
\end{figure}

An automaton $\mathcal{A}_{\textit{l-eq}}^2$ is given in \figurename~\ref{fig:automata for second enumeration of PQ1Equal}, and it is constructed for the second enumeration in \figurename~\ref{fig:history enumeration 2 for PQ1Equal}. In \figurename~\ref{fig:automata for second enumeration of PQ1Equal}, $C_1$, $C_2$, $C_3$ and $C_4$ is same as that in \figurename~\ref{fig:automata for first enumeration of PQ1Equal}.

\begin{figure}[htbp]
  \centering
  \includegraphics[width=0.8 \textwidth]{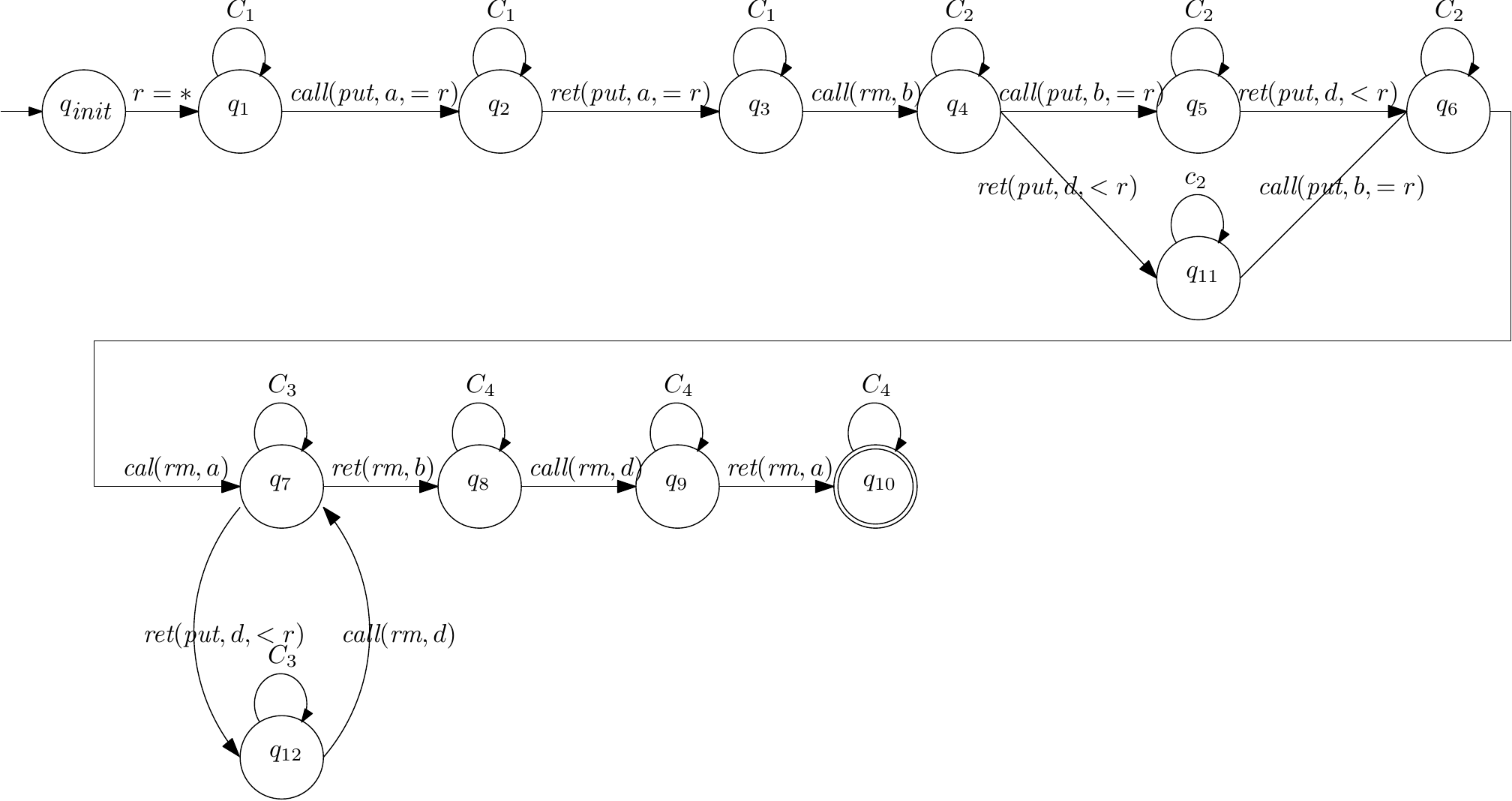}
  \caption{Automaton $\mathcal{A}_{\textit{l-eq}}^2$}
  \label{fig:automata for second enumeration of PQ1Equal}
\end{figure}

For the third enumeration in \figurename~\ref{fig:history enumeration 3 for PQ1Equal}. Since we want to ensure that $a$ and $b$ are putted only once, we need to explicitly record the positions of $\textit{call}(\textit{put},a,p)$ and $\textit{call}(\textit{put},b,p)$. Since the positions of $\textit{call}(\textit{put},a,p)$ and $\textit{call}(\textit{put},b,p)$ are not fixed, there are finite possible cases to consider, as shown below:

\begin{itemize}
\setlength{\itemsep}{0.5pt}
\item[-] If $\textit{call}(\textit{put},b,p)$ is after $\textit{call}(\textit{rm},b)$ and before $\textit{call}(\textit{rm},a)$: There are two possible positions of $\textit{call}(\textit{put},a,p)$: (1) before $\textit{call}(\textit{rm},a_1)$, and (2) after $\textit{call}(\textit{rm},a_1)$, and before $\textit{ret}(\textit{put},a,p)$.

\item[-] If $\textit{call}(\textit{put},b,p)$ is after $\textit{ret}(\textit{rm},a_1)$ and before $\textit{call}(\textit{rm},b)$: same as above case.

\item[-] If $\textit{call}(\textit{put},b,p)$ is after $\textit{call}(\textit{put},a_1,p)$ and before $\textit{ret}(\textit{rm},a_1)$: same as above case.

\item[-] If $\textit{call}(\textit{put},b,p)$ is after $\textit{ret}(\textit{put},a,p)$ and before $\textit{call}(\textit{put},a_1,p)$: same as above case.

\item[-] If $\textit{call}(\textit{put},b,p)$ is after $\textit{call}(\textit{rm},a_1)$ and before $\textit{ret}(\textit{put},a,p)$: There are three possible positions of $\textit{call}(\textit{put},a,p)$: (1) after $\textit{call}(\textit{put},b,p)$ and before $\textit{ret}(\textit{put},a,p)$, (2) after $\textit{call}(\textit{rm},a_1)$ and before $\textit{call}(\textit{put},b,p)$, and (3) before $\textit{call}(\textit{rm},a_1)$.

\item[-] If $\textit{call}(\textit{put},b,p)$ is before $\textit{call}(\textit{rm},a_1)$: There are three possible positions of $\textit{call}(\textit{put},a,p)$: (1) after $\textit{call}(\textit{rm},a_1)$ and before $\textit{ret}(\textit{put},a,p)$, (2) after $\textit{call}(\textit{put},b,p)$ and before $\textit{call}(\textit{rm},a_1)$, and (3) before $\textit{call}(\textit{put},b,p)$.
\end{itemize}

Therefore, there are fourteen possible cases that satisfy the third enumeration in \figurename~\ref{fig:history enumeration 3 for PQ1Equal}. For each case, we construct an finite automaton. Let $\mathcal{A}_{\textit{1-eq}}^{3}$ be the union of register automata that is constructed for above fourteen cases. For example, for the case $\textit{ca}_1$ when $\textit{call}(\textit{put},a,p)$ is before $\textit{call}(\textit{rm},a_1)$, $\textit{call}(\textit{put},b,p)$ is after $\textit{ret}(\textit{rm},a_1)$, and $\textit{call}(\textit{put},b,p)$ is before $\textit{call}(\textit{rm},b)$, we construct a finite automaton $\mathcal{A}_{\textit{l-eq}}^{\textit{3-1}}$ in \figurename~\ref{fig:automata for ca1 of third enumeration of Rpr2}. In \figurename~\ref{fig:automata for ca1 of third enumeration of Rpr2}, let $C$ and $C_1$ the same as that in \figurename~\ref{fig:automata for first enumeration of PQ1Equal}. Let $C_2 = C_1 \cup \{ \textit{ret}(\textit{put},a_1,=r) \}$, $C_3 = C_2 \cup \{ \textit{ret}(\textit{put},b,=r) \}$, $C_4 = C_3 \cup \{ \textit{call}(\textit{put},d,<r), \textit{ret}(\textit{rm},d) \}$, and $C_5 = C \cup \{ \textit{ret}(\textit{put},b,=r), \textit{ret}(\textit{put},a_1,=r), \textit{ret}(\textit{rm},d) \}$. Other register automata can be similarly constructed.

\begin{figure}[htbp]
  \centering
  \includegraphics[width=0.8 \textwidth]{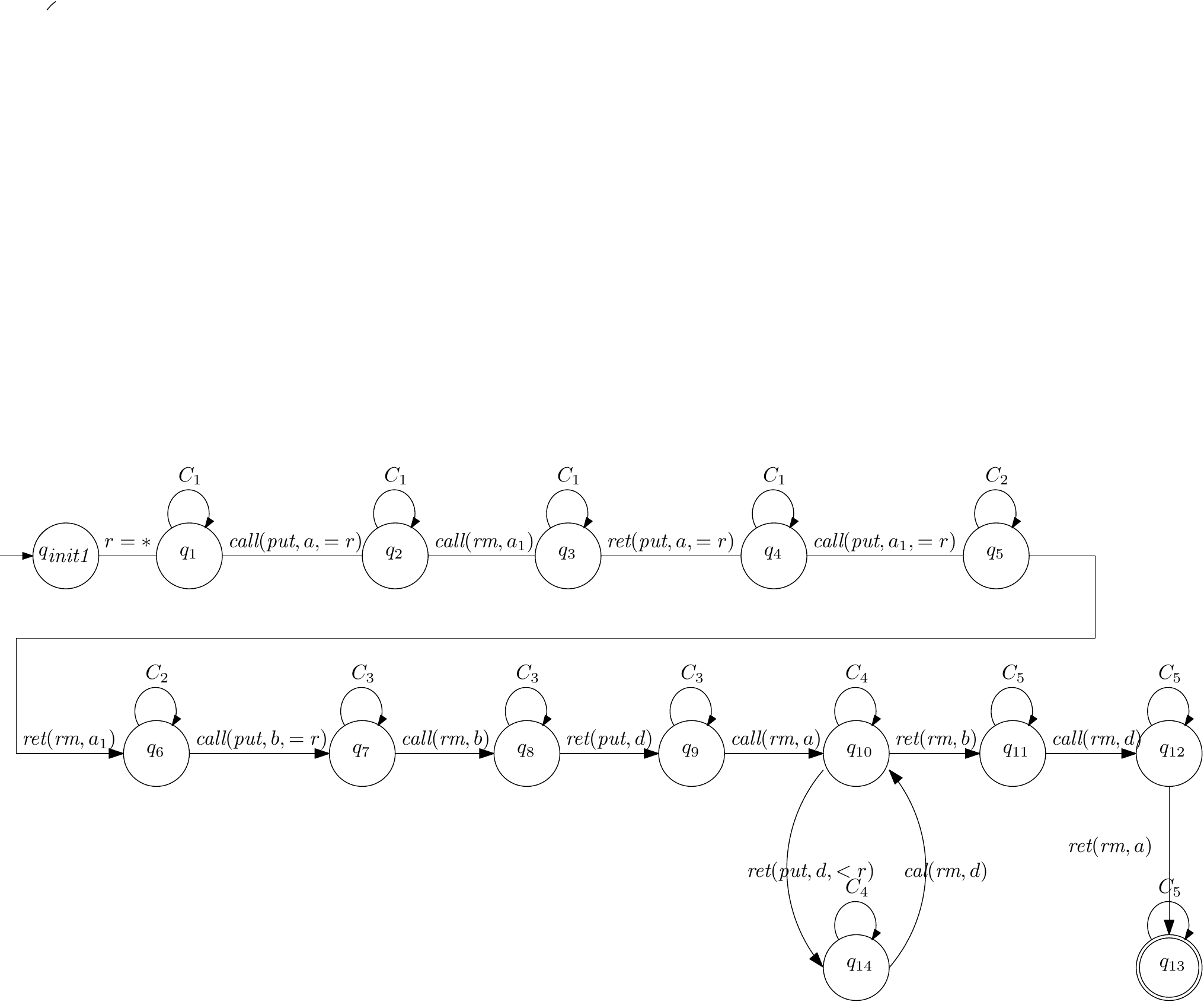}
  \caption{Automaton $\mathcal{A}_{\textit{l-eq}}^{\textit{3-1}}$}
  \label{fig:automata for ca1 of third enumeration of Rpr2}
\end{figure}

Similarly, we construct sets $\mathcal{A}_{\textit{1-eq}}^{4}$ and $\mathcal{A}_{\textit{1-eq}}^{5}$ of union of register automata for the forth enumeration in \figurename~\ref{fig:history enumeration 4 for PQ1Equal} and the fifth enumeration in \figurename~\ref{fig:history enumeration 5 for PQ1Equal}, respectively.

Let $\mathcal{A}_{\textit{1-eq}}$ be the union of $\mathcal{A}_{\textit{l-eq}}^1, \mathcal{A}_{\textit{l-eq}}^2,\mathcal{A}_{\textit{1-eq}}^{3}, \mathcal{A}_{\textit{1-eq}}^{4}$ and $\mathcal{A}_{\textit{1-eq}}^{5}$. The following lemma states that $\mathcal{A}_{\textit{1-eq}}$ is $\mathsf{MatchedMaxPriority}^{=}$-complete.

\begin{restatable}{lemma}{EPQOneEqualIsCoRegular}
\label{lemma:EPQ1Equal is co-regular}
$\mathcal{A}_{\textit{1-eq}}$ is $\mathsf{MatchedMaxPriority}^{=}$-complete.
\end{restatable}

\begin {proof}

We need to prove that, given a data-independent implementation $\mathcal{I}$. $\mathcal{A}_{\textit{1-eq}} \cap \mathcal{I} \neq \emptyset$ if and only if there exists $e \in \mathcal{I}$ and $e' \in \textit{proj}(e)$ such that $e'$ is not $\mathsf{MatchedMaxPriority}^{=}$-linearizable.

By Lemma \ref{lemma:pri execution is enough} and Lemma \ref{lemma:EPQ1Equal as pb order and gap-point}, we need to prove the following fact:

\noindent {\bf $\textit{fact}_1$}: Given a data-independent implementation $\mathcal{I}$. $\mathcal{A}_{\textit{1-eq}} \cap \mathcal{I} \neq \emptyset$ if and only if there exists $e \in \mathcal{I}$ and $e' \in \textit{proj}(e)$, $\mathsf{Has\text{-}MatchedMaxPriority}^{=}(e')$ holds, $x$ is the value with maximal priority in $e'$, $e'$ has only one maximal priority, $a$ and $b$ are two values with maximal priority $p$ in $e'$, $a <_{\textit{pb}}^* b$ in $e'$, and the rightmost gap-point of $b$ is before $\textit{call}(\textit{put},a,p)$ or $\textit{call}(\textit{rm},a)$ in $e'$.

\noindent The $\textit{only if}$ direction: Assume that $e_1 \in \mathcal{I}$ is accepted by some register automata in $\mathcal{A}_{\textit{1-eq}}$. By data-independence, there exists data-differentiated execution $e_2 \in \mathcal{I}$ and a renaming function $r$, such that $e_1=r(e_2)$. Since $e_1$ is accepted by some register automata in  $\mathcal{A}_{\textit{1-eq}}$, let $x$, $y$ and $z$ (if exists) be the values that are renamed into $b$, $a$ and $a_1$ (if exists) by $r$, respectively, and let $d_1,\ldots,d_m$ be the values that are renamed into $d$ by $r$.

let $e'' = e_2 \vert_{ \{ x,y,z,d_1,\ldots,d_m \} }$. It is obvious that $e'' \in \textit{proj}(e_2)$, $e''$ has only one maximal priority, and $\mathsf{Has\text{-}MatchedMaxPriority}^{=}(e'')$ holds. According to our construction of automata in $\textit{Auts}_{\textit{1-eq}}$, it is not hard to see that $x$ and $y$ has maximal priority in $h_2$, $y <_{\textit{pb}}^* x$, and the rightmost gap-point of $x$ is before $\textit{call}(\textit{put},y,p)$ or $\textit{call}(\textit{rm},y)$ in $e''$.

\noindent The $\textit{if}$ direction: Assume that there exists $e \in \mathcal{I}_{\neq},e' \in \textit{proj}(e)$, such that $\mathsf{Has\text{-}MatchedMaxPriority}^{=}(e')$ holds $e'$ has only one maximal priority, $a'$ and $b'$ are two values with maximal priority $p$ in $e'$, $a' <_{\textit{pb}}^* b'$ in $e'$, and the rightmost gap-point of $b'$ is before $\textit{call}(\textit{put},a',p)$ or $\textit{call}(\textit{rm},a')$ in $e'$. By data-independence, we can obtain execution $e_1$ as follows: (1) rename $a'$ and $b'$ into $a$ and $b$, respectively, (2) for the values $d_1,\ldots,d_m$ that constitute the rightmost gap of $b'$, we rename them into $d$, (3) if $a' <_{\textit{pb}}^A a'_1 <_{\textit{pb}}^B b$, we rename $a'_1$ into $a_1$, and (4) rename the other values into $\top$. It is easy to see that $\mathsf{Has\text{-}MatchedMaxPriority}^{=}(e_1)$ holds, $a$ and $b$ has maximal priority in $e_1$, $a <_{\textit{pb}}^* b$ in $e_1$, and the rightmost gap-point of $b$ is before $\textit{call}(\textit{put},a,p)$ or $\textit{call}(\textit{rm},a)$ in $e_1$. By Lemma \ref{lemma:five enumeration is enough for EPQ1Equal}, there are five possible enumeration of operations of $a$, $b$, $a_1$ (if exists). Then

\begin{itemize}
\setlength{\itemsep}{0.5pt}
\item[-] If $a <_{\textit{pb}}^* b$ because of the first enumeration, it is easy to see that $e_1$ is accepted by $\mathcal{A}_{\textit{l-eq}}^1$.

\item[-] If $a <_{\textit{pb}}^* b$ because of the second enumeration, it is easy to see that $e_1$ is accepted by $\mathcal{A}_{\textit{l-eq}}^2$.

\item[-] If $a <_{\textit{pb}}^* b$ because of the third enumeration, it is easy to see that $e_1$ is accepted by some register automaton in $\mathcal{A}_{\textit{1-eq}}^{3}$.

\item[-] If $a <_{\textit{pb}}^* b$ because of the forth enumeration, it is easy to see that $e_1$ is accepted by some register automaton in $\mathcal{A}_{\textit{1-eq}}^{4}$.

\item[-] If $a <_{\textit{pb}}^* b$ because of the fifth enumeration, it is easy to see that $e_1$ is accepted by some register automaton in $\mathcal{A}_{\textit{1-eq}}^{5}$.
\end{itemize}

This completes the proof of this lemma. \qed
\end {proof}

\subsection{$\mathsf{UnmatchedMaxPriority}^{>}(e)$ Always Holds}
\label{subsec:appendix co-regular of EPQ2Lar}

\begin{restatable}{lemma}{EPQ2LarIsAlwaysCoRegular}
\label{lemma:EPQ2Lar is always co-regular}

Given a data-differentiated execution $e$ with only one maximal priority, if $\mathsf{Has\text{-}}\mathsf{Unmat-}$ $\mathsf{chedMaxPriority}^{>}(e)$ holds, then $e \sqsubseteq l$ for some $l$ where $\mathsf{UnmatchedMaxPriority}^{>}\mathsf{\text{-}Seq}(l,x)$ holds.
\end{restatable}

\begin {proof}
Since $\mathsf{Has\text{-}}\mathsf{UnmatchedMaxPriority}^{>}(e)$ holds, the actions with maximal priority in $e$ is some unmatched $\textit{put}$. Therefore, no matter how we locate linearization points, we can always obtain a sequence $l$ of operations that contains unmatched $\textit{put}$ with maximal priority, and this satisfy the requirements of $\mathsf{UnmatchedMaxPriority}^{>}$. This completes the proof of this lemma. \qed
\end {proof}

\subsection{$\mathsf{UnmatchedMaxPriority}^{=}(e)$ Always Holds}
\label{subsec:appendix co-regular of EPQ2Equal}

\begin{restatable}{lemma}{EPQ2EqualAsHappenBefore}
\label{lemma:EPQ2Equal as happen before}
Given a data-differentiated execution $e$ with only one priority and $\mathsf{Has\text{-}}\mathsf{Unmatched-}$ $\mathsf{MaxPriority}^{=}(e)$ holds. $e$ is not $\mathsf{UnmatchedMaxPriority}^{=}(e)$-linearizable, if and only if there exists $x$ and $y$ with maximal priority $p$, $x$ has unmatched $\textit{put}$, $y$ has matched $\textit{put}$ and $\textit{rm}$, and $\textit{put}(x,p) <_{\textit{hb}} \textit{put}(y,p)$.
\end{restatable}

\begin {proof}

The $\textit{if}$ direction is obvious.

To prove the $\textit{only if}$ direction, we prove its contrapositive. Assume that for each pair of $x$ and $y$ with maximal priority in $e$, if $x$ has unmatched $\textit{put}$, $y$ has matched $\textit{put}$ and $\textit{rm}$, then $\textit{put}(x,p)$ does not happen before $\textit{put}(y,p)$. We need to prove that $e$ is $\mathsf{UnmatchedMaxPriority}^{=}(e)$-linearizable.

Let $x_1,\ldots,x_m$ be the set of values with priority $p$ and has unmatched $\textit{put}$ in $e$, let $y_1,\ldots,y_n$ be the set of values with priority $p$ and has matched $\textit{put}$ and $\textit{rm}$ in $e$. By assumption, we know that $\textit{call}(\textit{put},y_i,p)$ is before $\textit{ret}(\textit{put},x_j,p)$ for each $i,j$. Then we explicitly construction the linearization of $e$ by locating the linearization points of $e$ as follows:

\begin{itemize}
\setlength{\itemsep}{0.5pt}
\item[-] For each $x_i$, locate the linearization point of $\textit{put}(x_i,p)$ just before its return action.

\item[-] For each $y_j$, locate the lineariztion point of $\textit{put}(y_j,p)$ jest after its call action.

\item[-] For other operations, locate their linearization points at an arbitrary location after its call action and before its return action.
\end{itemize}

Let $l$ be the sequence of linearization points. It is easy to see that $e \sqsubseteq l$. Since linearization points of $\textit{put}(x_i,p)$ is after the linearization point of $\textit{put}(y_j,p)$ for each $i,j$, it is easy to see that $\mathsf{UnmatchedMaxPriority}^{=}\mathsf{\text{-}Seq}(l,x')$ holds where $x' = x_{\_}$ and its return action is after other values in $x_1,\ldots,x_m$. This completes the proof of this lemma. \qed
\end {proof}

Lemma \ref{lemma:EPQ2Equal as happen before} shows how to check non $\mathsf{UnmatchedMaxPriority}^{=}(e)$-linearizable exections. However, the case in Lemma \ref{lemma:EPQ2Equal as happen before} violates our assumption that each single-priority execution is FIFO. Therefore, we know that $\mathsf{UnmatchedMaxPriority}^{=}(e)$ always holds, as states by the following lemma.

\begin{restatable}{lemma}{EPQ2EqualIsAlwaysCoRegular}
\label{lemma:EPQ2Equal is always co-regular}
Given a data-differentiated execution $e$ with only one maximal priority, if $\mathsf{Has\text{-}}\mathsf{Unmat-}$ $\mathsf{chedMaxPriority}^{=}(e)$ holds, then $e \sqsubseteq l$ for some $l$ where $\mathsf{UnmatchedMaxPriority}^{=}\mathsf{\text{-}Seq}(l,x)$ holds.
\end{restatable}

\begin {proof}

According to Lemma \ref{lemma:EPQ2Equal as happen before}, if $\mathsf{Has\text{-}}\mathsf{UnmatchedMaxPriority}^{=}(e)$ holds and $e$ is not $\mathsf{Unmatched-}$ $\mathsf{MaxPriority}^{=}$-linearizable, then there exists $x$ and $y$ with maximal priority $p$, $x$ has unmatched $\textit{put}$, $y$ has matched $\textit{put}$, and $\textit{put}(x,p) <_{\textit{hb}} \textit{put}(y,p)$. Let $e_1 = e \vert_{ \{ x,y \} }$. It is obvious that $e_1$ does not satisfy FIFO property. This contradicts the assumption that every single-priority execution has FIFO property, and thus, we can safely ignore this case. \qed
\end {proof}

\subsection{Proofs, Definitions and Register Automata for $\mathsf{EmptyRemove}$}
\label{subsec:co-regular of EPQ3}

In this subsection we construct $\mathsf{EmptyRemove}$-complete register automata. The notion of left-right constraint of $\textit{rm}(\textit{empty})$ is inspired by left-right constraint of queue \cite{DBLP:conf/icalp/BouajjaniEEH15}.

\begin{definition}\label{def:left-right constraint for rmEmpty operation}
Given a data-differentiated execution $e$, and $o = \textit{rm}(\textit{empty})$ of $e$. The left-right constraint of $o$ is the graph $G$ where:

\begin{itemize}
\setlength{\itemsep}{0.5pt}
\item[-] the nodes are the values of $e$ or $o$, to which we add a node,

\item[-] there is an edge from value $d_1$ to $o$, if $\textit{put}(d_1,\_)$ happens before $o$,

\item[-] there is an edge from $o$ to value $d_1$, if $o$ happens before $\textit{rm}(d_1)$ or $\textit{rm}(d_1)$ does not exists in $h$,

\item[-] there is an edge from value $d_1$ to value $d_2$, if $\textit{put}(d_1,\_)$ happens before $\textit{rm}(d_2,\_)$.
\end{itemize}
\end{definition}

Given a data-differentiated execution $e$ and $o = \textit{rm}(\textit{empty})$ of $e$, it is obvious that $\mathsf{Has\text{-}}\mathsf{EmptyRe-}$ $\mathsf{move}^{=}(e)$ holds. Let $\textit{USet}_1(e,o) = \{ \textit{op} \vert$ $\textit{op}$ is an operation of some value, and either $\textit{op} <_{\textit{hb}} o$, or there is $\textit{op}'$ with the same value of $\textit{op}$, such that $\textit{op}' <_{\textit{hb}} o \}$. For each $i \geq 1$, let $\textit{USet}_{\textit{i+1}}(e,o) = \{ \textit{op} \vert$ $\textit{op}$ is an operation of some value, $\textit{op}$ is not in $\textit{USet}_k(e,o)$ for each $k \leq i$, and either $\textit{op}$ happens before some $o' \in \textit{USet}_i(e,o)$, or there is $\textit{op}''$ with the same value of $o$ and $\textit{op}''$ happens before some $o' \in \textit{USet}_i(e,o) \}$. We can see that $\textit{USet}_i(e,o) \cap \textit{USet}_j(e,o) = \emptyset$ for any $i \neq j$. Let $\textit{USet}(e,o) = \textit{USet}_1(e,o) \cup \textit{USet}_2(e,o) \cup \ldots$.

Similarly as $\textit{UVSet}$, we can prove the following two lemmas for $\textit{USet}$.

\begin{restatable}{lemma}{USetHasMatchedPutandRm}
\label{lemma:USet has matched put and rm}

Given a data-differentiated execution $e$ where $\mathsf{Has\text{-}EmptyRemove}(e)$ holds. Let $o$ be a $\textit{rm}(\textit{empty})$ of $e$. Let $G$ be the graph representing the left-right constraint of $o$. Assume that $G$ has no cycle going through $o$. Then, $\textit{USet}(e,o)$ contains only matched $\textit{put}$ and $\textit{rm}$.
\end{restatable}

This Lemma can be similarly proved as Lemma \ref{lemma:UVSet has matched put and rm}.

\begin{restatable}{lemma}{RmxDoesNotHappenBeforeUSetForEPQ3}
\label{lemma:Rmx does not happen before USet for EPQ3}
Given a data-differentiated execution $e$ where $\mathsf{Has\text{-}EmptyRemove}(e)$ holds. Let $o$ be a $\textit{rm}(\textit{empty})$ of $e$. Let $G$ be the graph representing the left-right constraint of $o$. Assume that $G$ has no cycle going through $o$. Then, $o$ does not happen before any operation in $\textit{USet}(e,o)$.
\end{restatable}

This Lemma can be similarly proved as Lemma \ref{lemma:Rmx does not happen before UVSet for EPQ1Lar}.

Then we can prove that getting rid of cycle though $o$ in left-right constraint is enough for ensure linearizable w.r.t $\textit{MS}(\seqPQ_3)$, as stated by the following lemma.

\begin{restatable}{lemma}{LINEqualsConstraintforEPQ3}
\label{lemma:Lin Equals Constraint for EPQ3}
Given a data-differentiated execution $e$ where $\mathsf{Has\text{-}EmptyRemove}(e)$ holds. $e$ is not $\mathsf{EmptyRemove}$-linearizable, if and only if there exists $o = \textit{rm}(\textit{empty})$ in $e$, $G$ has a cycle going through $o$, where $G$ is the graph representing the left-right constraint of $o$.
\end{restatable}

\begin {proof}

To prove the $\textit{if}$ direction, assume that there is such a cycle. Assume by contradiction that $e \sqsubseteq l= u \cdot o \cdot v$ and $\mathsf{EmptyRemove\text{-}Seq}(l,o)$ holds. Let $U$ and $V$ be the set of operations in $u$ and $v$. Let the cycle be $d_1 \rightarrow d_2 \rightarrow \ldots \rightarrow d_m \rightarrow o \rightarrow d_1$ in $G$. Since $d_m \rightarrow o$, $\textit{put}(d_m,\_)$ happens before $o$, and it is easy to see that $\textit{put}(d_m,\_)$ is in $U$. Since $U$ contains matched $\textit{put}$ and $\textit{rm}$, we can see that operations of $d_m$ is in $U$. Similarly, we can see that operations of $d_{\textit{m-1}},\ldots,d_1$ is in $U$. If $\textit{rm}(d_1)$ does not exists, then this contradicts that $U$ contains matched $\textit{put}$ and $\textit{rm}$. Else, if $\textit{rm}(d_1)$ exists, since $o$ happens before $\textit{rm}(d_1)$, we can see that $\textit{rm}(d_1) \in V$, which contradicts that $\textit{rm}(d_1) \in U$. This completes the proof of the $\textit{if}$ direction.

To prove the $\textit{only if}$ direction, we prove its contrapositive. Assume that for each such $o$ and $G$, $G$ has no cycle going through $o$. Let $O$ be the set of operations of $e$, except for $\textit{rm}(\textit{empty})$. Let $O_L = \textit{USet}(e,o)$, $O_R = O \setminus O_L$.

By Lemma \ref{lemma:USet has matched put and rm}, we can see that $O_L = \textit{USet}(e,o)$ contains only matched $\textit{put}$ and $\textit{rm}$. Let $O'_L$ be the union of $O_L$ and all the $\textit{rm}(\textit{empty})$ that happens before some operations in $O_L \cup \{ o \}$. Let $O'_R$ be the union of $O_R$ and the remanning $\textit{rm}(\textit{empty})$. It remains to prove that for $O'_L$, $\{ o \}$, $O'_R$, no elements of the latter set happens before elements of the former set. We prove this by showing that all the following cases are impossible:

\begin{itemize}
\setlength{\itemsep}{0.5pt}
\item[-] Case $1$: If some operation $o_r \in O'_R$ happens before $o$. Then we can see that $o_r \in \textit{USet}(e,o)$ or is a $\textit{rm}(\textit{empty})$ that happens before $o$, and then $o_r \in O'_L$, which contradicts that $o_r \in O'_R$.

\item[-] Case $2$: If some operation $o_r \in O'_R$ happens before some operation $o_l \in O'_L$. Then we know that $o_r \in \textit{USet}(e,o)$ or is a $\textit{rm}(\textit{empty})$ that happens before some operations in $O_L \cup \{ o \}$, and then $o_r \in O'_L$, which contradicts that $o_r \in O'_R$.

\item[-] Case $3$: If $o$ happens before some $o_l \in O'_L$. If $o_l \in \textit{USet}(e,o)$, then by Lemma \ref{lemma:Rmx does not happen before USet for EPQ3} we know that this is impossible. Else, $o_l$ is a $\textit{rm}(\textit{empty})$ that happens before some operations in $O_L \cup \{ o \}$, and $o$ happens before some operations in $O_L \cup \{ o \}$, which is impossible by Lemma \ref{lemma:Rmx does not happen before USet for EPQ3}.
\end{itemize}

This completes the proof of the $\textit{only if}$ direction.

\qed
\end {proof}

Let us begin to represent a register automaton that is used for capture the case that, in a sub-execution $e'$ of an execution $e$, $\mathsf{Has\text{-}EmptyRemove}(e')$ holds, $e'$ is not $\mathsf{EmptyRemove}$-linearizable, and the reason is that there is a cycle going through some $\textit{rm}(\textit{empty})$ $o$ in the left-right constraint of $o$. The automaton is $\mathcal{A}_{\seqPQ}^3$, which is given in \figurename~\ref{fig:automata for PQ3}. In \figurename~\ref{fig:automata for PQ3}, let $C = \{ \textit{call}(\textit{put},\top,\textit{true}),\textit{ret}($ $\textit{put},\top,\textit{true}), \textit{call}(\textit{rm},\top), \textit{ret}(\textit{rm},\top),\textit{call}(\textit{rm},\textit{empty}),\textit{ret}(\textit{rm},\textit{empty}) \}$, $C_1 = C \cup \{ \textit{call}(\textit{put},b,\textit{true}) \}$, $C_2 = C_1 \cup \{ \textit{ret}(\textit{rm},b) \}$, and $C_3 = C \cup \{ \textit{ret}(\textit{rm},b) \}$.

\begin{figure}[htbp]
  \centering
  \includegraphics[width=0.8 \textwidth]{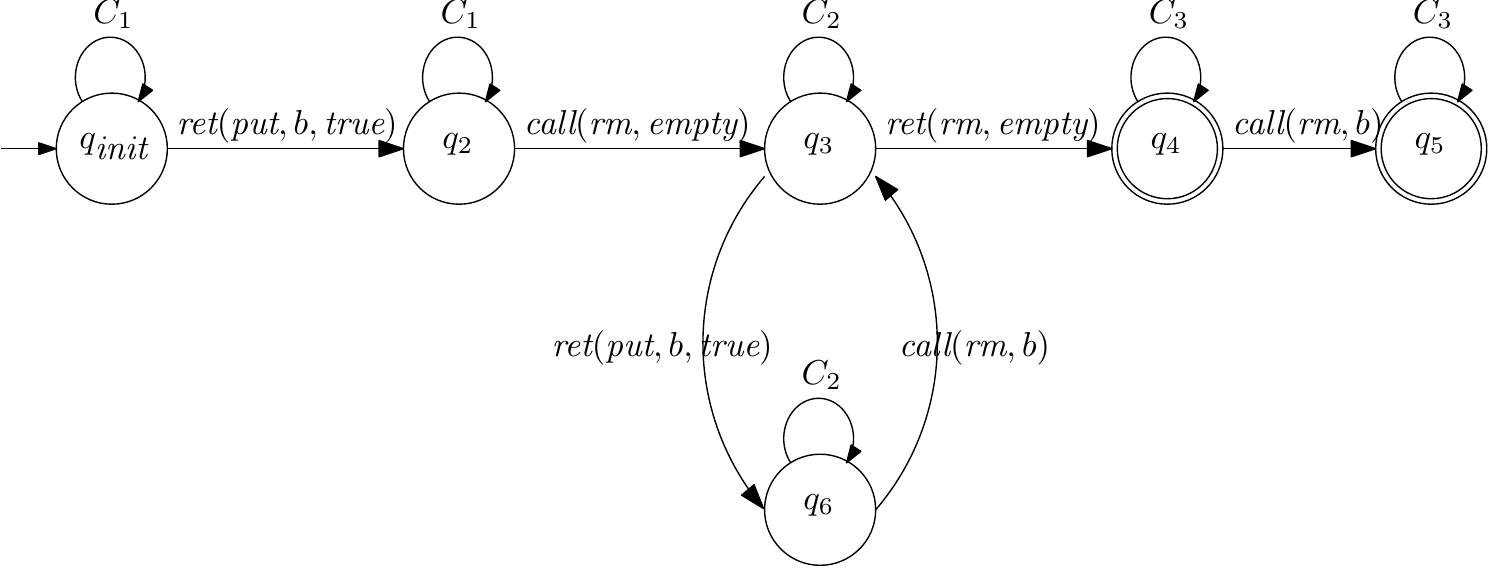}
  \caption{Automaton $\mathcal{A}_{\seqPQ}^3$}
  \label{fig:automata for PQ3}
\end{figure}

Given a data-differentiated execution $e$, we say that $o = \textit{rm}(\textit{empty})$ in $e$ is covered by values $d_1,\ldots,d_m$, if

\begin{itemize}
\setlength{\itemsep}{0.5pt}
\item[-] $\textit{put}(d_m,\_)$ happens before $o$,

\item[-] For each $i < 1 \leq m$,$\textit{put}(d_{\textit{i-1}},\_)$ happens before $\textit{rm}(d_i)$,

\item[-] $o$ happens before $\textit{rm}(d_1)$, or $\textit{rm}(d_1)$ does not exists in $e$
\end{itemize}

According to the definition of left-right constraint for $o$, in a data-differentiated execution $e$, there is a cycle going through $o$, if and only if there exists values $d_1,\ldots,d_m$, such that $o$ is covered by $d_1,\ldots,d_m$.

\begin{restatable}{lemma}{EPQ3IsCoRegular}
\label{lemma:EPQ3 is co-regular}
$\mathcal{A}_{\seqPQ}^3$ is $\mathsf{EmptyRemove}^{=}$-complete.
\end{restatable}

\begin {proof}

We need to prove that, given a data-independent implementation $\mathcal{I}$. $\mathcal{A}_{\seqPQ}^3 \cap \mathcal{I} \neq \emptyset$ if and only if there exists $e \in \mathcal{I}$ and $e' \in \textit{proj}(e)$ such that $e'$ is not $\mathsf{EmptyRemove}^{=}$-linearizable.

By Lemma \ref{lemma:Lin Equals Constraint for EPQ3}, we need to prove the following fact:

\noindent {\bf $\textit{fact}_1$}: Given a data-independent implementation $\mathcal{I}$. $\mathcal{A}_{\seqPQ}^3 \cap \mathcal{I} \neq \emptyset$ if and only if $\exists e \in \mathcal{I}_{\neq},e' \in \textit{proj}(e)$, $\mathsf{Has\text{-}EmptyRemove}(e')$ holds, $o = \textit{rm}(\textit{empty})$ is in $e'$, and $o$ is covered by some values $d_1,\ldots,d_m$ in $e'$.

\noindent The $\textit{only if}$ direction: Assume that $e_1 \in \mathcal{I}$ is accepted by $\mathcal{A}_{\seqPQ}^3$. By data-independence, there exists data-differentiated execution $e_2 \in \mathcal{I}$ and a renaming function $r$, such that $e_1=r(e_2)$. Let $d_1,\ldots,d_m$ be the values in $e_2$ such that $r(d_i)=b$ for each $1 \leq i \leq m$. Let $e_3 = e_2 \vert_{ \{ o, d_1, \ldots, d_m \} }$. It is obvious that $e_3 \in \textit{proj}(e_2)$ and $\textit{last}(e_3) = \seqPQ_3$. It is easy to see that $o$ is covered by $d_1,\ldots,d_m$.

\noindent The $\textit{if}$ direction: Assume that there exists such $e$, $e'$, $o$ and $d_1,\ldots,d_m$. Then, let $e_1$ be obtained from $e$ by renaming $d_1,\ldots,d_m$ into $b$ and renaming other values into $\top$. By data-independence, $e_1 \in \mathcal{I}$. It is easy to see that $e_1$ is accepted by $\mathcal{A}_{\seqPQ}^3$.

This completes the proof of this lemma. \qed
\end {proof}

\subsection{Proofs and Definitions in Section \ref{subsec:combine step-by-step linearizability and co-regular}}
\label{subsec:appendix proof and definition in subsection decidability result}

In this subsection, we show how to prove that the verification of linearizability of priority queue to be PSPACE-complete for a fixed number of threads, and EXPSPACE-complete for unbounded number of processes.

When considering decidability and complexity, we consider only finite variables with finite data domain. In such case, the priority queue implementations can only model priority queue with bounded capacity. Then, at some time, the priority queue may become full and then put can not succeed until some item is removed. However, we assume that each method can always make progress in isolation, which means that $\textit{put}$ must return some value when the priority queue is full. Therefore, we introduce a specific value $\textit{full}$. A $\textit{put}$ method returns $\textit{full}$ indicates that the priority queue is full now. To fit our results of priority queue in former section, when checking linearizability of priority queue, we only consider executions without any $\textit{put}(\_,\_,\textit{full})$. Since does not mean we consider only executions with bounded number of operations, since an execution can contain unbounded number of pairs of put and remove operations by linearizable w.r.t pairs of put and remove operations.

\noindent {\bf Vector Addition Systems with States (VASS):} A VASS \cite{conf/esop/BouajjaniEEH13} $\mathcal{A}=(Q_\mathcal{A},\rightarrow_\mathcal{A})$ contains a finite set $Q_\mathcal{A}$ of states and a finite set $\rightarrow_\mathcal{A}$ of transitions. Each transition is chosen from $Q_\mathcal{A} \times \mathbb{Z}^s \times Q_\mathcal{A}$, where $s$ is a positive integer. A vector that has a single non-zero component in $\{1,-1\}$ is called a unit vector. We assume, w.l.o.g., that all vectors of transitions are unit vectors.

The operational semantics of a VASS is defined as a LTS. A configuration $(q,\vec n)$ of VASS $\mathcal{A}$ contains a state $q \in Q_\mathcal{A}$ with a vector $\vec n \in \mathbb{N}^s$. $\vec n$ can be considered as values of several counters. We use $\vec n(i)$ to denote the $i$-th component of $\vec n$. The transitions between configurations of VASS is defined as follows: $(q_1, \vec n_1) \rightarrow_{\mathcal{A}} (q_2, \vec n_2)$, if there exists vector $\vec n$, such that (1) $(q_1,\vec n,q_2) \in \rightarrow_\mathcal{A}$, (2) if $\vec n(j)=1$ for some $j$, then $\vec n_2$ is obtain from $\vec n_1$ by increasing the $i$-th component by $1$, and (3) if $\vec n(j)=-1$ for some $j$, then $\vec n_2$ is obtain from $\vec n_1$ by decreasing the $i$-th component by $1$.

The state-reachability problem of VASS is to determine whether some configuration with state $q$ is reachable from a given configuration $(q_0,\vec 0)$. Here $\vec 0$ is a specific vector for which all components are $0$ and $q_0$ is the initial state of VASS.

\noindent {\bf Reducing priority Queue into PSPACE and EXPSPACE problems:} We describe a class $\mathcal{C}$ of data-independent implementations for which linearizability w.r.t. $\seqPQ$ is decidable. The implementations in $\mathcal{C}$ allow an unbounded number of values but a bounded number of priorities. Each method manipulates a finite number of local variables which store Boolean values, or data values from $\mathbb{D}$. Methods communicate through a finite number of shared variables that also store Boolean values, or data values from $\mathbb{D}$. Data values may be assigned, but never used in program predicates (e.g., in the conditions of if-then-else and while statements) so as to ensure data independence. This class captures typical implementations, or finite-state abstractions thereof, e.g., obtained via predicate abstraction. Since the $\Gamma$-complete automata $A(\Gamma)$ uses a fixed set $D=\{a,b,a_1,d,e,\top\}$ of values, we have that $\mathcal{C}\cap A(\Gamma)\neq\emptyset$ for some $\Gamma$ iff $\mathcal{C}_D\cap A(\Gamma)\neq\emptyset$ where $\mathcal{C}_D$ is the subset of $\mathcal{C}$ that uses only values in $D$.

We model the set $\mathcal{C}_D$ as the executions of a Vector Addition System with States (VASS) similarly to Bouajjani et al. \cite{conf/esop/BouajjaniEEH13}. A VASS has a finite set of states, and manipulates a finite set of counters holding non-negative values. The idea is that the states of the VASS represent the global variables of $\mathcal{C}_D$. This set is finite as we have bounded the data values. Each counter of the VASS then represents the number of threads which are at a particular control location within a method, with a certain valuation of the local variables. Here again, as the data values are bounded, there is a finite number of valuations. When a thread moves from a control location to another, or updates its local variables, we decrement, resp., increment, the counter corresponding the old, resp., the new, control location or valuation.

For a fixed set of priorities $P$, the register automata $A(\Gamma)$ can be transformed to finite-state automata since the number of possible valuations of the registers is bounded. In this way, we reduce checking linearizability of priority queue into the the EXPSPACE-complete problem of checking state-reachability in a VASS. When we consider only finite number of threads, the value of counters are bounded and then this problem is PSPACE \cite{DBLP:conf/ac/Esparza96}.

\noindent {\bf Reducing PSPACE and EXPSPACE problems into linearizability of priority queue for bounded number of threads:} We show how to reduce the state-reachability problem of VASS into checking linearizability of priority queue. Our construction is based on the reduction of state-reachability problem of counter machine and VASS into linearizability in \cite{conf/esop/BouajjaniEEH13}.

There are already several linearizable lock-free implementations of priority queue \cite{DBLP:conf/ipps/SundellT03,DBLP:conf/ppopp/LiuS12}. It is not hard to modify them to fit finite variables with finite data domain and it is obvious that they are still linearizable then. Let $\textit{Imp-PQ}_{\textit{lf}}$ be one such implementations. It is safe to assume that $\textit{Imp-PQ}_{\textit{lf}}$ is given in the form of pseudo-code. We use $\textit{put}_{\textit{lf}}$ and $\textit{rm}_{\textit{lf}}$ to explicitly denote the pseudo-code of put and remove method of $\textit{Imp-PQ}_{\textit{lf}}$.

A counter machine is a VASS that can additional detect if the value of some counter is $0$. In \cite{conf/esop/BouajjaniEEH13}, they construct a library $\mathcal{L}_{\mathcal{A}}$ that simulate executions of counter machine $\mathcal{A}$. This library use finite variables with finite data domain. It contains the following six methods:

\begin{itemize}
\setlength{\itemsep}{0.5pt}
\item[-] $M_{(q,\vec n,q')}$, which represents a transition resulting of $(q,\vec n,q')$.

\item[-] $M_{(q,i,q')}$, which represents a transition that from state $q$ to state $q'$ while detecting the value of counter $i$ to be $0$.

\item[-] $M_{q_f}$. This method returns if the simulating reaches some state with state $q_f$.

\item[-] $M_{\textit{inc-i}}$, which simulates increasing counter $i$ with $1$.

\item[-] $M_{\textit{dec-i}}$, which simulates decreasing counter $i$ with $1$.

\item[-] $M_{\textit{zero-i}}$, which simulates testing whether the value of counter $i$ to be $0$.
\end{itemize}

The first three methods do their work with the help of the latter three methods. For example, a $M_{(q,\vec n,q')}$ method (that increases counter $j$) will activate a $M_{inc-j}$ method and a $M_{dec-j}$ method.

Let $D' = \{ (q,\vec n,q') \vert q,q' \in Q_{\mathcal{A}} - \{q_f\}, \vec n \in \mathbb{Z}^s \} \cup \{ q_f \} \cup \{\textit{inc-i},\textit{dec-i} \vert 1 \leq i \leq s\}$. Let $\textit{RAs}$ be the register automata obtained as follows:

\begin{itemize}
\setlength{\itemsep}{0.5pt}
\item[-] Given $\Gamma\in \{\mathsf{EmptyRemove}$, $\mathsf{UnmatchedMaxPriority}$, $\mathsf{MatchedMaxPriority}\}$, we modify the $\Gamma$-complete automata while the original value of $\top$ is replaced by values $\top$ and all values in $D'$.

\item[-] For register automata that checks FIFO of sub-executions of operations with a same priority, if the original value under consideration is $\{ b \}$ (the first three automata) or $\{a,b\}$ (the forth automaton), then we replace them with $\{ b \} \cup D'$ and $\{a,b\} \cup D'$, respectively.
\end{itemize}

Given a VASS $\mathcal{A}$, we construct a priority queue implementation $\textit{PQ}_{\mathcal{A}}$ as follows:

\begin{itemize}
\setlength{\itemsep}{0.5pt}
\item[-] A new memory location $\textit{reachQf}$ is introduced. Its initial value is set to false.

\item[-] The $\textit{rm}$ method first check the value of $\textit{reachQf}$. If it is true, then it returns a random value. Otherwise, it works as $\textit{rm}_{\textit{lf}}$. The pseudo-code of $\textit{rm}$ is shown in Algorithm \ref{alg:new rm}.
\end{itemize}

\begin{algorithm}[t]
\KwOut{A removed value}

\If {$\textit{reachQf} = \textit{true}$}
{\KwRet a random value from $D'$\;}
\Else
{\KwRet $\textit{rm}_{\textit{lf}}$()\;}

\caption{$\textit{rm}$ of $\textit{PQ}_{\mathcal{A}}$}
\label{alg:new rm}
\end{algorithm}

\begin{itemize}
\setlength{\itemsep}{0.5pt}
\item[-] The $\textit{put}$ method uses items from $D' \cup \{ a,b,a_1,d,e,\top \}$. It first uses $\textit{put}_{\textit{lf}}$ to do the work of priority queue, and then choose a role and begin to simulate transitions of $\mathcal{A}$, until $q_f$ is reached. Then, it set the flag $\textit{reachQf}$ to be $\textit{true}$. The pseudo-code of $\textit{put}$ is shown in Algorithm \ref{alg:new put}.
\end{itemize}

In \cite{conf/esop/BouajjaniEEH13}, to use a library to simulate counter machine $\mathcal{A}$, their library $\mathcal{L}_{\mathcal{A}}$ contains the following variables: (1) variable $q$ that uses value from $Q_{\mathcal{A}}$, (2) vector $\textit{req}[]$ and $\textit{ack}[]$, the index is chosen from $1$ to the number of unit vectors of length $s$, (3) $\textit{dec}[]$, where the index is chosen form $1$ to $s$.

Let us explain how $\textit{put}$ method do to act according to its chosen role. We use the case of $\textit{trans}(q,\vec n,q')$ as an example while other cases can be similarly obtained. For the role $\textit{trans}(q,\vec n,q')$, it works as $M_{(q,\vec n,q')}$ in $\mathcal{L}_{\mathcal{A}}$. Or we can say, it do the following work:

\begin{itemize}
\setlength{\itemsep}{0.5pt}
\item[-] atomically do $\textit{wait}(q)$ and $\textit{signal}(\textit{req}[\vec n])$,
\item[-] atomically do $\textit{wait}(\textit{ack}[\vec n])$ and $\textit{signal}(q')$,
\end{itemize}

Note that we want to obtain a implementation where each method can always make progress in isolation. To obtain such an implementation, we need to slightly modify the $\textit{wait}$ command as follows:

\begin{itemize}
\setlength{\itemsep}{0.5pt}
\item[-] If a wait command find itself blocked, then it can still proceed. However, in this case the $\textit{put}$ method will return a specific value $\textit{fail}$. Moreover, when deciding linearizability, we do not consider executions with return value $\textit{fail}$.
\end{itemize}

\begin{algorithm}[t]
\KwIn {An item $\textit{itm}$ and its priority $p$}
\KwOut{$\textit{full}$ if the priority queue is full}

$x = \textit{put}_{\textit{lf}}(\textit{itm},p)$\;

random select a role from $\{ \textit{trans}(q,\vec n,q'), \textit{reach-}q_f, \textit{inc-i},\textit{dec-i} \}$\;

works according to chosen role\;

\If {$q_f$ is reached}
{$\textit{reachQf} := \textit{true}$\;}

\If {$x = \textit{full}$}
{\KwRet $x$\;}
\Else
{\KwRet\;}
\caption{$\textit{put}$ of $\textit{PQ}_{\mathcal{A}}$}
\label{alg:new put}
\end{algorithm}

Follows the proof in \cite{conf/esop/BouajjaniEEH13}, we can prove that $(q_f,\_)$ is reachable in $\mathcal{A}$, if and only if $\textit{PQ}_{\mathcal{A}}$ is not linearizable w.r.t $\seqPQ$. In this way, we reduce the state-reachability problem of VASS into checking linearizablity of priority queue.

With above discussions, we can now prove the following theorem:

$\newline$
{\noindent \bf Theorem \ref{theorem:complexity of priority queue}}: Verifying whether an implementation in $\mathcal{C}$ is linearizable w.r.t. $\seqPQ$ is PSPACE-complete for a fixed number of threads, and EXPSPACE-complete otherwise.

\end{document}